\documentclass[onecolumn]{IEEEtran}
\usepackage{amsmath,amssymb,bm,color,amsthm,mathrsfs}
\usepackage[noadjust]{cite}
\usepackage[thicklines]{cancel}
\usepackage{stfloats}
\usepackage{tikz}
\usetikzlibrary{arrows.meta,calc,positioning}
\usepackage{enumerate}
\usepackage[linesnumbered,ruled]{algorithm2e}
\usepackage{hyperref}
\usepackage{cleveref}

\newtheorem{theorem}{Theorem}[section]
\newtheorem{lemma}{Lemma}[section]

\newtheorem{fact}{Fact}[section]
\newtheorem{proposition}{Proposition}[section]
\newtheorem{corollary}{Corollary}[section]
\newtheorem{definition}{Definition}[section]

\newtheorem{example}{Example}[section]
\newtheorem{remark}{Remark}[section]
\newtheorem{construction}{Construction}[section]

\newcommand\nc\newcommand
\nc{\cA}{\mathcal{A}}\nc{\cB}{\mathcal{B}}\nc{\cC}{\mathcal{C}}\nc{\cD}{\mathcal{D}}
\nc{\cE}{\mathcal{E}}\nc{\cF}{\mathcal{F}}\nc{\cG}{\mathcal{G}}\nc{\cH}{\mathcal{H}}
\nc{\cI}{\mathcal{I}}\nc{\cJ}{\mathcal{J}}\nc{\cK}{\mathcal{K}}\nc{\cL}{\mathcal{L}}
\nc{\cM}{\mathcal{M}}\nc{\cN}{\mathcal{N}}\nc{\cO}{\mathcal{O}}\nc{\cP}{\mathcal{P}}
\nc{\cQ}{\mathcal{Q}}\nc{\cR}{\mathcal{R}}\nc{\cS}{\mathcal{S}}\nc{\cT}{\mathcal{T}}
\nc{\cU}{\mathcal{U}}\nc{\cV}{\mathcal{V}}\nc{\cW}{\mathcal{W}}\nc{\cX}{\mathcal{X}}
\nc{\cY}{\mathcal{Y}}\nc{\cZ}{\mathcal{Z}}

\nc{\bba}{\mathbf{a}}\nc{\bbb}{\mathbf{b}}\nc{\bbc}{\mathbf{c}}\nc{\bbd}{\mathbf{d}}
\nc{\bbe}{\mathbf{e}}\nc{\bbf}{\mathbf{f}}\nc{\bbg}{\mathbf{g}}\nc{\bbh}{\mathbf{h}}
\nc{\bbi}{\mathbf{i}}\nc{\bbj}{\mathbf{j}}\nc{\bbk}{\mathbf{k}}\nc{\bbl}{\mathbf{l}}
\nc{\bbm}{\mathbf{m}}\nc{\bbn}{\mathbf{n}}\nc{\bbo}{\mathbf{o}}\nc{\bbp}{\mathbf{p}}
\nc{\bbq}{\mathbf{q}}\nc{\bbr}{\mathbf{r}}\nc{\bbs}{\mathbf{s}}\nc{\bbt}{\mathbf{t}}
\nc{\bbu}{\mathbf{u}}\nc{\bbv}{\mathbf{v}}\nc{\bbw}{\mathbf{w}}\nc{\bfx}{\mathbf{x}}
\nc{\bby}{\mathbf{y}}\nc{\bbz}{\mathbf{z}}

\nc{\bbA}{\mathbf{A}}\nc{\bbB}{\mathbf{B}}\nc{\bbC}{\mathbf{C}}\nc{\bbD}{\mathbf{D}}
\nc{\bbE}{\mathbf{E}}\nc{\bbF}{\mathbf{F}}\nc{\bbG}{\mathbf{G}}\nc{\bbH}{\mathbf{H}}
\nc{\bbI}{\mathbf{I}}\nc{\bbJ}{\mathbf{J}}\nc{\bbK}{\mathbf{K}}\nc{\bbL}{\mathbf{L}}
\nc{\bbM}{\mathbf{M}}\nc{\bbN}{\mathbf{N}}\nc{\bbO}{\mathbf{O}}\nc{\bbP}{\mathbf{P}}
\nc{\bbQ}{\mathbf{Q}}\nc{\bbR}{\mathbf{R}}\nc{\bbS}{\mathbf{S}}\nc{\bbT}{\mathbf{T}}
\nc{\bbU}{\mathbf{U}}\nc{\bbV}{\mathbf{V}}\nc{\bbW}{\mathbf{W}}\nc{\bfX}{\mathbf{X}}
\nc{\bbY}{\mathbf{Y}}\nc{\bbZ}{\mathbf{Z}}

\nc{\sA}{\mathsf{A}}\nc{\sB}{\mathsf{B}}\nc{\sC}{\mathsf{C}}\nc{\sD}{\mathsf{D}}
\nc{\sE}{\mathsf{E}}\nc{\sF}{\mathsf{F}}\nc{\sG}{\mathsf{G}}\nc{\sH}{\mathsf{H}}
\nc{\sI}{\mathsf{I}}\nc{\sJ}{\mathsf{J}}\nc{\sK}{\mathsf{K}}\nc{\sL}{\mathsf{L}}
\nc{\sM}{\mathsf{M}}\nc{\sN}{\mathsf{N}}\nc{\sO}{\mathsf{O}}\nc{\sP}{\mathsf{P}}
\nc{\sQ}{\mathsf{Q}}\nc{\sR}{\mathsf{R}}\nc{\sS}{\mathsf{S}}\nc{\sT}{\mathsf{T}}
\nc{\sU}{\mathsf{U}}\nc{\sV}{\mathsf{V}}\nc{\sW}{\mathsf{W}}\nc{\sX}{\mathsf{X}}
\nc{\sY}{\mathsf{Y}}\nc{\sZ}{\mathsf{Z}}

\nc{\bma}{\bm{a}}\nc{\bmb}{\bm{b}}\nc{\bmc}{\bm{c}}\nc{\bmd}{\bm{d}}
\nc{\bme}{\bm{e}}\nc{\bmf}{\bm{f}}\nc{\bmg}{\bm{g}}
\nc{\bmh}{\bm{h}}\nc{\bmi}{\bm{i}}\nc{\bmj}{\bm{j}}
\nc{\bmk}{\bm{k}}\nc{\bml}{\bm{l}}\nc{\bmm}{\bm{m}}
\nc{\bmn}{\bm{n}}\nc{\bmo}{\bm{o}}\nc{\bmp}{\bm{p}}
\nc{\bmq}{\bm{q}}\nc{\bmr}{\bm{r}}\nc{\bms}{\bm{s}}\nc{\bmt}{\bm{t}}\nc{\bmu}{\bm{u}}\nc{\bmv}{\bm{v}}\nc{\bmw}{\bm{w}}\nc{\bmx}{\bm{x}}\nc{\bmy}{\bm{y}}\nc{\bmz}{\bm{z}}

\nc{\bmA}{\bm{A}}\nc{\bmB}{\bm{B}}\nc{\bmC}{\bm{C}}\nc{\bmD}{\bm{D}}
\nc{\bmE}{\bm{E}}\nc{\bmF}{\bm{F}}\nc{\bmG}{\bm{G}}
\nc{\bmH}{\bm{H}}\nc{\bmI}{\bm{I}}\nc{\bmJ}{\bm{J}}
\nc{\bmK}{\bm{K}}\nc{\bmL}{\bm{L}}\nc{\bmM}{\bm{M}}
\nc{\bmN}{\bm{N}}\nc{\bmO}{\bm{O}}\nc{\bmP}{\bm{P}}
\nc{\bmQ}{\bm{Q}}\nc{\bmR}{\bm{R}}\nc{\bmS}{\bm{S}}\nc{\bmT}{\bm{T}}\nc{\bmU}{\bm{U}}\nc{\bmV}{\bm{V}}\nc{\bmW}{\bm{W}}\nc{\bmX}{\bm{X}}\nc{\bmY}{\bm{Y}}\nc{\bmZ}{\bm{Z}}

\newcommand{\mathset}[1]{\left\{#1\right\}}
\newcommand{\multiset}[1]{\left\{\left\{#1\right\}\right\}}
\newcommand{\abs}[1]{\left|#1\right|}
\newcommand{\ceilenv}[1]{\left\lceil #1 \right\rceil}
\newcommand{\floorenv}[1]{\left\lfloor #1 \right\rfloor}
\newcommand{\parenv}[1]{\left( #1 \right)}
\newcommand{\sparenv}[1]{\left[ #1 \right]}

\nc{\set}[1]{\llbracket #1 \rrbracket}

\newcommand{\dth}[1]{\mathsf{d}_H\left(#1\right)}
\newcommand{\Wth}[2]{\mathsf{wt}_{#1}\left(#2\right)}
\newcommand{\Sum}[1]{\mathsf{Sum}\left(#1\right)}
\newcommand{\vt}{\mathsf{VT}}

\newcommand{\e}{\mathsf{e}}
\newcommand{\temp}{\mathsf{temp}}
\newcommand{\Enc}{\textup{Enc}}
\newcommand{\Dec}{\textup{Dec}}
\newcommand{\expan}[2]{\mathsf{expan}_{#1}\left(#2\right)}
\newcommand{\T}{\mathsf{T}}
\newcommand{\E}{\mathbb{E}}
\newcommand{\ins}{\operatorname{ins}}

\newcommand{\F}{\mathbb{F}}

\newcommand{\cfcite}[2]{\textit{cf.}~\cite[#2]{#1}}

\title{Bounds and Constructions of Codes for Ordered Composite DNA Sequences}
\author{Zuo~Ye, Yuling~Li, Zhaojun~Lan and Gennian~Ge%
\thanks{This research is supported by the National Key Research and Development Program of China under Grant 2025YFC3409900, the National Natural Science Foundation of China under Grant 12231014 and Grant 12501466, Beijing Scholars Program, and Xiaomi Young Scholars Program.

Z. Ye is with the Institute of Mathematics and Interdisciplinary Sciences, Xidian University,
Xian 710126, China (email: yezuo@xidian.edu.cn).

Y. Li, Z. Lan, and G. Ge are with the School of Mathematical Sciences, Capital Normal University, Beijing 100048, China (emails: 2240501022@cnu.edu.cn, 2200501014@cnu.edu.cn, gnge@zju.edu.cn).
}
}

\begin{document}
\maketitle

\begin{abstract}
This paper extends the foundational work of Dollma \emph{et al} on codes for the $q$-ary
$k$-resolution ordered composite DNA channel. In this model, a codeword is a $k\times n$ matrix over $\{0,\ldots,q-1\}$ with nondecreasing columns. Each row is transmitted through a separate channel
that may introduce substitutions or deletions. An $(e_1,\ldots,e_k)$-composite-error correcting code (CECC) corrects at most $e_i$ substitutions in row $i$ for every $i$, whereas an $e$-CECC corrects at most $e$ substitutions in total across the $k$ rows. The corresponding $(e_1,\ldots,e_k)$- and $e$-composite-deletion correcting codes (CDCCs) correct exactly $e_i$ deletions in each row $i$ and exactly $e$ deletions in total, respectively. In addition to these four models, we introduce the $t$-$(e_1,\ldots,e_t)$-composite-error and composite-deletion models for $1\le t\le k$. In these models, $t$ rows are selected, but their indices are not specified in advance. In increasing
order of row index, the $i$-th selected row undergoes at most $e_i$ substitutions or exactly $e_i$ deletions, respectively, and all other rows are error-free.

For arbitrary fixed $q,k\ge2$ and error parameters, we derive upper bounds on code sizes in all six models. We also construct $t$-$(1,\ldots,1)$-CECCs and $t$-$(1,\ldots,1)$-CDCCs for every $1\le t\le k$. Taking $t=1$ gives $1$-CECCs and $1$-CDCCs,
which also serve as $(1,0,\ldots,0)$-CECCs and CDCCs; taking $t=k$ gives $(1,\ldots,1)$-CECCs and CDCCs. The redundancies of these constructions have optimal leading terms as $n\to\infty$. For the CECCs and for the CDCCs with $t<k$, the redundancies match the corresponding lower bounds up to an additive constant. These constructions have explicit encoding and decoding algorithms, and the CECC encoders are systematic. Finally, we provide explicit message encoding and decoding algorithms for the binary $(1,0,\ldots,0)$-CECC constructed by Dollma \emph{et al}.
\end{abstract}


\section{Introduction}\label{sec_introduction}
%%%%%%%%%%%%%%%%%%%%%%%%%%%%%%%%%%%%%%%%%%%%%%%%%
The exponential growth of global digital data has intensified the search for high‑density, long‑term storage alternatives. DNA‑based storage has emerged as a promising candidate, owing to its exceptional information density, longevity, and physical stability, alongside continuous advances in synthesis and sequencing technologies \cite{science2012,nature2013,Yazdi2015TMBMC,YANIV2017Science}. In conventional DNA storage, data is encoded into sequences over the four‑letter alphabet $\{A, C, G, T\}$. Then DNA strands, which will be stored in suitable containers, are synthesized according to these sequences. To retrieve the original
data, the stored DNA strands are sampled and sequenced.

Despite substantial progress in DNA-based data storage, the high cost of DNA synthesis remains a major barrier to practical deployment \cite{YANIV2017Science}. The synthesis cost is closely related to the total number of synthesis cycles required to write a given amount of data \cite{Leon2019NB}. This has motivated efforts to reduce the number of synthesis cycles \cite{Lenz2020ISIT,Makarychev2022IT,Ohad2023IT,Abu-Sini2023ISIT,Chrisnata2023ISIT,Tuan2024ISIT,Immink2024TMBMSC,Yajuan2025ISIT,Leon2019NB,Yeongjae2019SR,Yan2023SR,Press2024SR}.

One approach is to increase the amount of information represented at each synthesized position by using \emph{composite DNA letters} \cite{Leon2019NB,Yeongjae2019SR}. In a standard DNA strand, each position is represented by a single letter from the DNA alphabet $\mathset{A,C,G,T}$. In contrast, a composite letter represents a position by a predetermined mixture of $A$, $C$, $G$ and $T$ in a given ratio $\sigma=\parenv{p_A,p_C,p_G,p_T}$, where $p_A,p_C,p_G,p_T\ge 0$ and $p_A+p_C+p_G+p_T=1$. In other words, a fraction $p_b$ of the mixture corresponds to base $b$, for $b\in\mathset{A,C,G,T}$. Therefore, a sequence of composite letters (called a composite sequence) corresponds to multiple standard DNA sequences. This is feasible because the DNA synthesis process typically produces numerous copies of standard strands for each designed sequence. We say that the letter $\sigma$ has resolution $k$ if it can be written as $\sigma=(k_A/k,k_C/k,k_G/k,k_T/k)$
with nonnegative integers $k_A,k_C,k_G,k_T$ summing to $k$. There are $\binom{k+3}{3}$ such letters. Using them as an alphabet can therefore increase the theoretical information represented per synthesis cycle beyond the two bits represented by a single DNA base.

The use of composite letters also brings new challenges, including the reconstruction of composite sequences from sequencing reads. To overcome this challenge, researchers have proposed combinatorial composite letters, a variant of composite letters, and investigated related coding problems \cite{Yan2023SR,Press2024SR,Omer2024ISIT, Zuo2025DCC}; others have examined the sequencing coverage‑depth problem \cite{Preuss2024TMBMSC,Roman2025TCOMM,Cohen2025JSAIT}.

Another challenge lies in designing codes tailored to the composite channel. Because each composite letter specifies a mixture of DNA bases, a single composite sequence can give rise to many distinct standard DNA strands. For example, consider the composite sequence $\bms=AMCMTM$, where $M$ is a mixture of $G$ and $T$ (that is, $p_G,p_T>0$ and $p_A=p_C=0$). Then even in the noiseless case, the synthesis process may produce up to eight distinct standard sequences:
$$
\begin{array}{cccc}
    AGCGTG,&AGCGTT,&AGCTTG,&ATCGTG,\\
    AGCTTT,&ATCGTT,&ATCTTG,&ATCTTT.
\end{array}
$$
As the length of the composite sequence grows, the number of possible standard strands can increase rapidly. The large number of possible strands and the uncertainty about which strands are synthesized complicate code design. Several coding problems arising from this synthesis method have been studied \cite{Wenkai2025IT,Frederik2024ISIT,Frederik2025ISIT,Cohen2025ISTC,Tuan2025ISIT,WangChen2025ITW}.

To overcome this challenge, Dollma \emph{et al} recently \cite{BesartDollma202509} introduced the \emph{$k$-resolution ordered composite DNA channel}. In this model, the standard DNA strands to be produced are assigned to $k$ ordered groups, indexed by $1,\ldots,k$, with each group representing a fraction $1/k$ of a composite letter. The group membership and order remain fixed throughout synthesis. At a position with composite letter $(k_A/k,k_C/k,k_G/k,k_T/k)$, the synthesis process appends $A$ to every strand in the first $k_A$ groups, $C$ to every strand in the next $k_C$ groups, $G$ to every strand in the next $k_G$ groups, and $T$ to every strand in the remaining $k_T$ groups. All strands within a group are thus identical before errors occur. For code design, each group is represented by a single strand.
The resulting $k$ representative strands are transmitted through $k$ independent noisy channels, with the channel index of each received strand known to the decoder. For a fixed resolution $k$, the number of representative strands is $k$, regardless of the composite sequence length.

Identifying $A,C,G,T$ with $0,1,2,3$, respectively, gives a matrix description of the model. A composite sequence of length $n$ is represented by a $k\times n$ matrix whose rows are the representative strands. Every column is nondecreasing because bases are assigned to the groups in the order $A<C<G<T$. Each column is an \emph{ordered composite letter}, and each row is sent through a separate channel that may introduce errors. The decoder receives the output of each channel together with its channel index and aims to recover the transmitted matrix. Replacing the four DNA bases by the alphabet $\Sigma_q=\{0,1,\ldots,q-1\}$ gives the \emph{$q$-ary $k$-resolution ordered composite DNA channel}. Here, $q$ is the alphabet size within each row, $k$ is the number of ordered groups, and $n$ is the number of composite letters.

The authors of \cite{BesartDollma202509} conducted a preliminary study of this channel model for the binary case ($q=2$). They defined four error types: $(e_1,\ldots,e_k)$-composite-error/deletion and $e$-composite-error/deletion. An $(e_1, e_2, \ldots, e_{k})$-composite-error/composite-deletion means that for each $1 \le i \le k$, the $i$-th channel introduces at most $e_i$ substitutions/exactly $e_i$ deletions, while an $e$-composite-error/deletion means that the $k$ channels collectively introduce up to $e$ substitutions/exactly $e$ deletions. Their work presented initial bounds and code constructions for these models.

\subsection{Our Contributions}
%%%%%%%%%%%%%%%%%%%%%%%%%%%%%%%%%%%%%%%%%%%%%%%%
This paper provides a deeper investigation of composite-error correcting codes (CECCs) and composite-deletion correcting codes (CDCCs) for arbitrary alphabet size $q\ge2$ and resolution $k\ge2$. In the following, denote $Q_{q,k}=\binom{k+q-1}{q-1}$, which is the number of nondecreasing column vectors of length $k$ over $\Sigma_q$. Let $n$ be the code length. All asymptotic statements below concern $n\to\infty$ with $q$, $k$, and the error parameters fixed.

\subsubsection{\textbf{Additional Error Models}}
%%%%%%%%%%%%%%%%%%%%%%%%%%%%%%%%%%%%
We introduce the $t$-$(e_1,\ldots,e_t)$-composite-error and $t$-$(e_1,\ldots,e_t)$-composite-deletion models for $1\le t\le k$ and nonnegative integers $e_1,\ldots,e_t$. In these models, $t$ channels with indices $j_1<\cdots<j_t$ are selected, but their indices are not specified in advance. Channel $j_i$ introduces at most $e_i$ substitutions or exactly $e_i$ deletions, respectively, for every $i\in[t]$, and all other channels are error-free. Clearly, the $t$-$(e_1,\ldots,e_t)$ models are more general than the corresponding $(e_1,\ldots,e_k)$ models.

\vspace{5pt}
\subsubsection{\textbf{Upper Bounds}}
%%%%%%%%%%%%%%%%%%%%%%%%%%%%%%%%%%%%%%%%%%%
In \Cref{sec_CECCbound}, we derive upper bounds on the sizes of CECCs. These bounds imply that the redundancies of an $(e_1,\ldots,e_k)$-CECC and an $e$-CECC are respectively at least $\parenv{\sum_{i=1}^ke_i}\log_{Q_{q,k}}(n)-O(1)$ and $e\log_{Q_{q,k}}(n)-O(1)$, for any fixed $q,k\ge2$, $e$ and $(e_1,\ldots,e_k)$. As a corollary, the redundancy of a $t$-$(e_1,\ldots,e_t)$-CECC is at least $\parenv{\sum_{i=1}^te_i}\log_{Q_{q,k}}(n)-O(1)$.

In \Cref{sec_bounddeltion}, we establish upper bounds for $(e_1,\ldots,e_k)$-CDCCs. With $e_0=\min\{e_1,\ldots,e_k\}$ and $R=\sum_{i=1}^k(e_i-e_0)$, these bounds give the redundancy lower bound $e_0\log_{Q_{q,k}}(n)+R\log_{Q_{q,k}}\ln(n)-O(1)$.
We also obtain bounds for $e$-CDCCs and $t$-$(e_1,\ldots,e_t)$-CDCCs. In particular, writing $e=ks+r$ with $0\le r<k$, every $e$-CDCC has redundancy at least $s\log_{Q_{q,k}}(n)+r\log_{Q_{q,k}}\ln(n)-O(1)$. For $t$-$(1,\ldots,1)$-CDCCs, the lower bound is $t\log_{Q_{q,k}}\ln(n)-O(1)$ when $t<k$, and $\log_{Q_{q,k}}(n)-O(1)$ when $t=k$.

\vspace{5pt}
\subsubsection{\textbf{Constructions}}
%%%%%%%%%%%%%%%%%%%%%%%%%%%%%%%%%%%%%%%%%%
For every $q,k\ge2$ and $1\le t\le k$, we give systematic constructions of $t$-$(1,\ldots,1)$-CECCs with redundancy $t\log_{Q_{q,k}}(n)+O(1)$. These constructions match the corresponding redundancy lower bounds up to an additive constant. The case $t=1$, equivalently the $1$-CECC model, is treated in \Cref{sec_1CECC}. The cases $2\le t\le k$ are treated in \Cref{sec_t1CECC}. For binary $1$-CECCs, we also give precise finite-length guarantees: such codes exist with redundancy at most $\log_{k+1}(2n+1)$ for every $k\ge2$, and our systematic encoder has redundancy $\ceilenv{\log_{k+1}(n)}+1$.

For every $q\ge2$ and $1\le t<k$, \Cref{sec_t<k} gives $t$-$(1,\ldots,1)$-CDCCs with explicit encoding and decoding algorithms and redundancy $t\log_{Q_{q,k}}\ln(n)+O(1)$, matching the lower bound up to an additive constant. Taking $t=1$ yields $1$-CDCCs, and hence also $(1,0,\ldots,0)$-CDCCs. When $t=k$, \Cref{sec_t=k} gives $(1,\ldots,1)$-CDCCs with explicit encoding and decoding algorithms and redundancy $\log_{Q_{q,k}}(n)+(k-1)\log_{Q_{q,k}}\ln(n)+O(1)$. The coefficient of $\log_{Q_{q,k}}(n)$ is optimal.

In addition to these constructions, \Cref{sec_10CECC} provides explicit message encoding and decoding algorithms for the binary $k$-resolution $(1,0,\ldots,0)$-CECC of Dollma \emph{et al} \cite[Theorem~9]{BesartDollma202509}.

\vspace{5pt}
The rest of this paper is organized as follows. \Cref{sec_pre} introduces the notation and error models. \Cref{sec_CECCbound,sec_conssub} present bounds and constructions for CECCs, respectively, and \Cref{sec_bounddeltion,sec_consdel} do the same for CDCCs. \Cref{sec_conclusion} concludes the paper. Appendix~\ref{sec_equiv} proves the code equivalences, and Appendix~\ref{sec_restrictedencoder} gives the auxiliary encoding and decoding algorithms used in the CDCC constructions.

\section{Preliminary}\label{sec_pre}
%%%%%%%%%%%%%%%%%%%%%%%%%%%%%%%%%%%%%%%%%%%%%%%%%%%
For integers $m$ and $n$ with $m\le n$, let the interval $[m,n]$ denote the set $\mathset{m,m+1,\ldots,n}$. We set $[m,n]=\emptyset$ when $m>n$. Throughout this paper, we abbreviate $[1,n]$ and $[0,n-1]$ as $[n]$ and $[[n]]$, respectively. For any alphabet $\Sigma$, its elements are referred to as \emph{letters} or \emph{symbols}. Denote by $\Sigma^n$ the set of all sequences (or vectors) of length $n$ over $\Sigma$, where $n$ is a nonnegative integer. In particular, $\Sigma^0$ contains only the empty sequence $\bm\epsilon$. The \emph{concatenation} of two sequences $\bmx$ and $\bmy$ over $\Sigma$ is denoted by $\bmx\bmy$. In particular, for any $a\in\Sigma$ and $n\ge0$, we use $a^n$ to denote the sequence consisting of $n$ copies of the symbol $a$.

Given a sequence (or vector) $\bms\in\Sigma^n$ and an index $i\in[n]$, we write $\bms[i]$ for the symbol at position $i$ in $\bms$. Thus, a sequence (resp. vector) $\bms\in\Sigma^n$ can also be written as $\bms[1]\bms[2]\cdots\bms[n]$ (resp. $\parenv{\bms[1],\bms[2],\ldots,\bms[n]}$). For a set $I=\mathset{i_1,\ldots,i_t}\subseteq[n]$ with $i_1<\cdots<i_t$ and a sequence $\bms\in\Sigma^n$, let $\bms\mid_I$ denote the \emph{subsequence} $\bms[i_1]\cdots\bms[i_t]$. In particular, we refer to $\bms\mid_I$ as a \emph{substring} if $I$ is an interval.

For a finite set $S$, we use $\abs{S}$ to denote the cardinality of $S$. For $\alpha\in\Sigma$ and a sequence (or vector) $\bms$ of length $n$ over $\Sigma$, define $\Wth{\alpha}{\bms}\triangleq\abs{\mathset{i\in[n]:\bms[i]=\alpha}}$, which counts the number of occurrences of $\alpha$ in $\bms$.

\subsection{Ordered Composite Letters and Sequences}%\label{sec_LetterSequence}
Let $k,q\ge2$ be two integers. Denote by $\Sigma_q$ the alphabet $\mathset{0,1,\ldots,q-1}$. A \emph{composite letter} over $\Sigma_q$ is a vector of probabilities $\parenv{p_0,p_1,\ldots,p_{q-1}}$ satisfying $0\le p_0,p_1,\ldots,p_{q-1}\le 1$ and $\sum_{i=0}^{q-1}p_i=1$. This means that the symbol $i\in\Sigma_q$ is observed with probability $p_i$. A composite letter of \emph{resolution} $k$ over $\Sigma_q$ is a vector of the form $\parenv{\frac{k_0}{k},\frac{k_1}{k},\ldots,\frac{k_{q-1}}{k}}$, where $k_0,k_1,\ldots,k_{q-1}$ are nonnegative integers summing to $k$.

Given a composite letter $\parenv{\frac{k_0}{k},\frac{k_1}{k},\ldots,\frac{k_{q-1}}{k}}$ of resolution $k$, we construct a column vector $\sigma=\sparenv{a_1,a_2,\ldots,a_{k}}^\T$ of length $k$ over $\Sigma_q$ with the property that $a_1\le a_2\le \cdots\le a_{k}$ and $\Wth{i}{\sigma}=k_i$ for every $i\in\Sigma_q$. In other words, $\sigma$ is a nondecreasing vector, and symbol $i\in\Sigma_q$ appears exactly $k_i$ times in $\sigma$. We call $\sigma$ an \emph{ordered composite letter/symbol} of \emph{resolution} $k$ over $\Sigma_q$. For brevity, we refer to ordered composite letters simply as symbols when no ambiguity arises. The set of all ordered composite letters of resolution $k$ over $\Sigma_q$ is denoted by $\Phi_{q,k}$. Define $Q_{q,k}=\abs{\Phi_{q,k}}$, i.e., the number of valid letters. Clearly, $Q_{q,k}=\binom{k+q-1}{q-1}$. Throughout this paper, an \emph{invalid} letter (or symbol) refers to any element in $\Sigma_q^k\setminus\Phi_{q,k}$, which violates the nondecreasing constraint. Elements in $\Phi_{q,k}$ and $\Phi_{q,k}^n$ are also called \emph{valid} letters and sequences, respectively.

Since each letter in $\Phi_{q,k}$ is a column vector of length $k$, a sequence $\bms$ in $\Phi_{q,k}^n$ can be regarded as a $k\times n$ matrix over $\Sigma_q$. For any $\bms\in\Phi_{q,k}^n$ and each $i\in[k]$, we always denote by $\bms_i$ the $i$-th row of $\bms$, unless stated otherwise.
\begin{example}
    Let $q=2$ and $k=n=3$. Then $\Phi_{2,3}=\mathset{\sparenv{0,0,0}^\T,\sparenv{0,0,1}^\T,\sparenv{0,1,1}^\T,\sparenv{1,1,1}^\T}$. The symbol $\sparenv{0,1,0}^\T$ is invalid, since it is not nondecreasing. Consider the sequence
    \begin{equation*}
        \bms=
        \begin{pmatrix}
            0&0&1\\
            0&1&1\\
            0&1&1
        \end{pmatrix}
    \end{equation*}
    in $\Phi_{2,3}^3$. By definition, we have $\bms_1=001$, $\bms_2=011$ and $\bms_3=011$.
\end{example}

\begin{definition}\label{def_bijection}
    Let $\cA_{q,k}=\mathset{\sum_{i=1}^{q-1}n_i(k+1)^{i-1}:n_1,\ldots,n_{q-1}\in\mathbb{Z}_{\ge0},\sum_{i=1}^{q-1}n_i\le k}$. Then $\abs{\cA_{q,k}}=Q_{q,k}$. Define a map
    \begin{equation*}
    \begin{array}{l}
    \Phi_{q,k}\rightarrow\cA_{q,k},\\
    \quad\sigma\mapsto\sum_{i=1}^{q-1}\Wth{i}{\sigma}(k+1)^{i-1}.
    \end{array}
    \end{equation*}
    It is straightforward to verify that this is a bijection. On the other hand, writing $\cA_{q,k}=\mathset{a_0,a_1,\ldots,a_{Q_{q,k}-1}}$ with $a_0<a_1<\cdots<a_{Q_{q,k}-1}$, we obtain a natural bijection from $\cA_{q,k}$ to $\Sigma_{Q_{q,k}}$ that sends $a_i$ to $i$.
    Thus, we obtain a bijection, denoted by $\iota_{q,k}$, from $\Phi_{q,k}$ to $\Sigma_{Q_{q,k}}$.

    For any integer $n\ge1$, the bijection $\iota_{q,k}$ extends to a bijection from $\Phi_{q,k}^n$ to $\Sigma_{Q_{q,k}}^n$ by $\iota_{q,k}(\bmx) \triangleq\iota_{q,k}(\bmx[1])\cdots \iota_{q,k}(\bmx[n])$. Its inverse $\iota_{q,k}^{-1}$ extends in the same way.
\end{definition}
Fixing this bijection, we may also interpret $\bms$ as a sequence in $\Sigma_{Q_{q,k}}^n$. In this paper, we alternately view sequences in $\Phi_{q,k}^n$ as $k\times n$ matrices over $\Sigma_q$ and as sequences in $\Sigma_{Q_{q,k}}^n$.
\begin{example}
\begin{itemize}
    \item
    Let $q=2$. Every letter $\sigma\in\Phi_{2,k}$ is uniquely
    determined by its number of ones: if $\Wth{1}{\sigma}=r$, then $\sigma$ consists of $k-r$ zeros followed by $r$ ones. In this case, $\cA_{2,k}=\mathset{0,1,\ldots,k}=\Sigma_{k+1}$,
    and the first bijection maps $\sigma$ to $\Wth{1}{\sigma}$. Since the elements of $\cA_{2,k}$ are already consecutive
    integers starting from $0$, the second bijection is the identity map. Thus, the resulting bijection from $\Phi_{2,k}$ to $\Sigma_{k+1}$ is simply
    $\sigma\mapsto\Wth{1}{\sigma}$.
    \item
    Let $q=4$ and $k=3$. Then
    $Q_{4,3}=\binom{6}{3}=20$, and the first bijection maps
    each $\sigma\in\Phi_{4,3}$ to $\Wth{1}{\sigma}
        +4\Wth{2}{\sigma}
        +16\Wth{3}{\sigma}$.
    The number of zeros in $\sigma$ is determined by the other
    three counts, since
    $\Wth{0}{\sigma}
    =3-\Wth{1}{\sigma}-\Wth{2}{\sigma}-\Wth{3}{\sigma}$.
    The resulting set of integers is $\cA_{4,3} =\mathset{0,1,2,3,4,5,6,8,9,12,16,17,18,20,21,24,32,33,36,48}$.

    Arranging these integers in increasing order, the second
    bijection maps the $(j+1)$-st integer to $j$, for
    $0\le j\le19$. The following table lists each ordered
    composite letter, its image in $\cA_{4,3}$, and its
    corresponding symbol in $\Sigma_{20}$:
    \begin{center}
    \small
    \setlength{\tabcolsep}{4pt}
    \renewcommand{\arraystretch}{1.12}
    \begin{tabular}{|ccc|ccc|ccc|ccc|}
        \hline
        $\sigma$ & $\cA_{4,3}$ & $\Sigma_{20}$ & $\sigma$ & $\cA_{4,3}$ & $\Sigma_{20}$ & $\sigma$ & $\cA_{4,3}$ & $\Sigma_{20}$ & $\sigma$ & $\cA_{4,3}$ & $\Sigma_{20}$ \\
        \hline
        $\sparenv{0,0,0}^{\T}$ & $0$ & $0$ & $\sparenv{0,1,2}^{\T}$ & $5$ & $5$ & $\sparenv{0,0,3}^{\T}$ & $16$ & $10$ & $\sparenv{2,2,3}^{\T}$ & $24$ & $15$ \\
        $\sparenv{0,0,1}^{\T}$ & $1$ & $1$ & $\sparenv{1,1,2}^{\T}$ & $6$ & $6$ & $\sparenv{0,1,3}^{\T}$ & $17$ & $11$ & $\sparenv{0,3,3}^{\T}$ & $32$ & $16$ \\
        $\sparenv{0,1,1}^{\T}$ & $2$ & $2$ & $\sparenv{0,2,2}^{\T}$ & $8$ & $7$ & $\sparenv{1,1,3}^{\T}$ & $18$ & $12$ & $\sparenv{1,3,3}^{\T}$ & $33$ & $17$ \\
        $\sparenv{1,1,1}^{\T}$ & $3$ & $3$ & $\sparenv{1,2,2}^{\T}$ & $9$ & $8$ & $\sparenv{0,2,3}^{\T}$ & $20$ & $13$ & $\sparenv{2,3,3}^{\T}$ & $36$ & $18$ \\
        $\sparenv{0,0,2}^{\T}$ & $4$ & $4$ & $\sparenv{2,2,2}^{\T}$ & $12$ & $9$ & $\sparenv{1,2,3}^{\T}$ & $21$ & $14$ & $\sparenv{3,3,3}^{\T}$ & $48$ & $19$ \\
        \hline
    \end{tabular}
    \end{center}
    For example, $\sigma=\sparenv{1,2,3}^{\T}$ contains one
    occurrence of each nonzero symbol. The first bijection
    therefore maps $\sigma$ to $1+4+16=21$.
    Since $21$ is the fifteenth smallest element of
    $\cA_{4,3}$, the second bijection maps it to $14$.
    Thus, the two steps are $\sparenv{1,2,3}^{\T}
    \longmapsto 21\longmapsto 14$.
\end{itemize}
\end{example}

\subsection{The Ordered Composite DNA Channel and Error Models}
%%%%%%%%%%%%%%%%%%%%%%%%%%%%%%%%%%%%%%%%%%%%
The authors of \cite{BesartDollma202509} initiated the study of the \emph{ordered composite DNA channel}. This channel takes a sequence $\bms\in\Phi_{q,k}^n$ as input. Each row $\bms_i$ of $\bms$ is sent through a separate channel $i$ which may introduce errors. Let the output of the $i$-th channel be $\bmy_i$. The aim is to recover $\bms$ from $\bmy_1,\bmy_2,\ldots,\bmy_{k}$.

The following four error patterns are introduced in \cite{BesartDollma202509}.
\begin{definition}
   \begin{itemize}
       \item An $(e_1, e_2, \ldots, e_{k})$-composite-error is an error pattern where, for each $1 \le i \le k$, the $i$-th channel introduces up to $e_i$ substitutions.
       \item A $k$-resolution $e$-composite-error is an error pattern where the $k$ channels together introduce up to $e$ substitutions.
       \item An $(e_1, e_2, \ldots, e_{k})$-composite-deletion is an error pattern where, for each $1\le i \le k$, the $i$-th channel introduces $e_i$ deletions.
       \item A $k$-resolution $e$-composite-deletion is an error pattern where the $k$ channels together introduce $e$ deletions.
   \end{itemize}
\end{definition}

Correspondingly, the following four families of codes were defined.
\begin{definition}
    \begin{itemize}
        \item An $(e_1, e_2, \ldots, e_{k})$-composite-error correcting code (abbreviated as $(e_1, e_2, \ldots, e_{k})$-CECC) is a code that can correct an $(e_1, e_2, \ldots, e_{k})$-composite-error.
        \item A $k$-resolution $e$-composite-error correcting code (abbreviated as $e$-CECC) is a code that can correct a $k$-resolution $e$-composite-error.
        \item An $(e_1, e_2, \ldots, e_{k})$-composite-deletion correcting code (abbreviated as $(e_1, e_2, \ldots, e_{k})$-CDCC) is a code that can correct an $(e_1, e_2, \ldots, e_{k})$-composite-deletion.
        \item A $k$-resolution $e$-composite-deletion correcting code (abbreviated as $e$-CDCC) is a code that can correct a $k$-resolution $e$-composite-deletion.
    \end{itemize}
\end{definition}
In Appendix~\ref{sec_equiv}, we show that it is not necessary to study these codes for all possible tuples $\parenv{e_1,e_2,\ldots,e_{k}}\in[0,n]^k$ (where $n$ is the code-length), as many are equivalent.

In this paper, we also study the following error model, which assumes the maximum number of channels that introduce errors and the maximum number of errors introduced by each channel. But it is not known a priori which of these channels will introduce errors.
\begin{definition}
Let $t\le k$, and let $e_1,\ldots,e_t$ be
nonnegative integers. A $k$-resolution $t$-$(e_1,\ldots,e_{t})$-composite-error (resp. $t$-$(e_1,\ldots,e_{t})$-composite-deletion) is an error pattern for which there are pairwise distinct channel indices $j_1,\ldots,j_t\in[k]$ with $j_1<\cdots<j_t$ such that channel $j_i$ introduces at most $e_i$ substitutions (resp. exactly $e_i$ deletions) for every $i\in[t]$, and all other channels are error-free. The indices $j_1,\ldots,j_t$ are not specified in advance.

 A $k$-resolution $t$-$(e_1,\ldots,e_{t})$-composite-error correcting code (resp. $t$-$(e_1,\ldots,e_{t})$-composite-deletion correcting code), abbreviated as $t$-$(e_1,\ldots,e_{t})$-CECC (resp. $t$-$(e_1,\ldots,e_{t})$-CDCC), is a code that can correct a $k$-resolution $t$-$(e_1,\ldots,e_{t})$-composite-error (resp. $t$-$(e_1,\ldots,e_{t})$-composite-deletion).
\end{definition}

\begin{remark}\label{rmk_relationship}
Let $1\le t\le k$, let $e_1,\ldots,e_t$ be nonnegative
integers, and put $E=\sum_{i=1}^t e_i$.
Every $E$-CECC is a $t$-$(e_1,\ldots,e_t)$-CECC,
and every $t$-$(e_1,\ldots,e_t)$-CECC is an
$(e_1,\ldots,e_t,0,\ldots,0)$-CECC.
Taking $t=k$ also shows that every
$\parenv{\sum_{i=1}^k e_i}$-CECC is an
$(e_1,\ldots,e_k)$-CECC.

The same implications hold for CDCCs.
Moreover, for fixed-length codes and admissible deletion
counts, replacing ``exactly'' by ``at most'' does not
change the deletion-correction requirements.
Indeed, any common output of two distinct codewords
produced by fewer deletions can be further shortened
to a common output with the prescribed number of deletions.

In particular, a $1$-CECC is equivalent to a
$1$-$(1)$-CECC, and a $1$-CDCC is equivalent to a
$1$-$(1)$-CDCC.
\end{remark}

For a $k$-resolution CECC or CDCC $\cC$ of length $n$, we say that $\cC$ is a $q$-ary code, if $\cC\subseteq\Phi_{q,k}^n$. The redundancy of $\cC$ is defined to be $\log_{Q_{q,k}}\parenv{\frac{Q_{q,k}^n}{\abs{\cC}}}$.

Throughout the paper, we assume that the alphabet size $q$, the resolution parameter $k$, and the total number of errors are constants independent of the code length $n$. Empty sums are zero. We also use $\binom{a}{b}=0$ for integers $0\le a<b$. We frequently assume the smallest prime $p$ greater than some integer $m\ge2$. Bertrand's postulate guarantees that such a $p$ does exist and satisfies $m<p<2m$. There is an efficient algorithm to test if a given integer is a prime \cite{Agrawal2004Annals,Agrawal2019Annals}.

\section{Upper Bounds on CECCs}\label{sec_CECCbound}
%%%%%%%%%%%%%%%%%%%%%%%%%%%%%%
Let $A_{q,k}^S(n;(e_1,\ldots,e_k))$ and
$A_{q,k}^S(n;e)$ denote the maximum sizes of
$q$-ary $k$-resolution $(e_1,\ldots,e_k)$-CECCs
and $e$-CECCs of length $n$, respectively. In this section, we derive upper bounds on these quantities. These bounds also yield upper bounds on the sizes of $t$-$(e_1,\ldots,e_t)$-CECCs, since every such code is an $(e_1,\ldots,e_t,0,\ldots,0)$-CECC by \Cref{rmk_relationship}. 

For a sequence $\bms\in\Phi_{q,k}^n$, a tuple $\parenv{e_1,e_2,\ldots,e_{k}}\in\mathbb{N}^k$ and an integer $e\ge 1$, we define two error balls centered at $\bms$:
\begin{equation*}
   \begin{array}{c}
         \cB_{\parenv{e_1,e_2,\ldots,e_{k}}}^{(q,k)}(\bms)\triangleq\mathset{\bmy\in\Phi_{q,k}^n:\dth{\bmy_i,\bms_i}\le e_i,\forall i\in[k]},\\
         \cB_{e}^{(q,k)}(\bms)\triangleq\mathset{\bmy\in\Phi_{q,k}^n:\sum_{i=1}^{k}\dth{\bmy_i,\bms_i}\le e},
   \end{array}
\end{equation*}
where $\dth{\cdot,\cdot}$ denotes the Hamming distance. In other words, $\cB_{\parenv{e_1,e_2,\ldots,e_{k}}}^{(q,k)}(\bms)$ consists of all \emph{valid} sequences obtained from $\bms$ by introducing at most $e_i$ substitutions in the $i$-th row, while $\cB_{e}^{(q,k)}(\bms)$ consists of all \emph{valid} sequences obtained from $\bms$ by introducing a total of at most $e$ substitutions.

\begin{remark}
    Applying an $\parenv{e_1,e_2,\ldots,e_{k}}$-composite-error (or an $e$-composite-error) to $\bms$ does not necessarily yield a valid sequence $\bmy$. Our definition includes only valid sequences in error balls. This enables estimating their sizes and consequently, deriving nontrivial upper bounds on code sizes.
\end{remark}

We begin with a fundamental lemma that quantifies the sizes of these error balls.
\begin{lemma}\label{lem_subballsize}
For any two integers  $q,k\ge2$, we have the following results.
    \begin{enumerate}[$(i)$]
        \item Let $\sigma:=\sparenv{
        a_1,\cdots,a_{k}}^\T\in\Phi_{q,k}$. Then
    \begin{align}
        \abs{\cB_{1}^{(q,k)}(\sigma)\setminus\{\sigma\}}=q-1+a_{k}-a_1,\label{eq_subsize1}\\
        \abs{\cB_{(1,1,\ldots,1)}^{(q,k)}(\sigma)\setminus\{\sigma\}}=\sum_{\ell=1}^{q-1}\binom{\ell+k-1}{\ell}.\label{eq_subsize2}
    \end{align}
    \item Let $e\ge 1$ and $0\le\ell\le q-1$ be two integers. Given an $\bms\in\Phi_{q,k}^n$ such that $\bms_{k}[i]-\bms_1[i]\ge\ell$ for at least $m$ indices $i\in[n]$, where $e\le m\le n$, we have
    \begin{equation}\label{eq_subsize3}
        \abs{\cB_{e}^{(q,k)}(\bms)}\ge\binom{m}{e}(q-1+\ell)^e+1.
    \end{equation}
    \item If $e_1\ge e_2\ge\cdots\ge e_{k}\ge1$, then for every $\bms\in\Phi_{q,k}^n$,
    \begin{equation}\label{eq_subsize4}
        \abs{\cB_{(e_1,\ldots,e_{k})}^{(q,k)}(\bms)}\ge\binom{n}{e_{k}}\sparenv{\sum_{\ell=1}^{q-1}\binom{\ell+k-1}{\ell}}^{e_{k}}+1.
    \end{equation}
    \end{enumerate}
\end{lemma}
\begin{IEEEproof}
    Consider a vector $\sigma^\prime=\begin{bmatrix}
        a_1^\prime,\cdots,a_{k}^\prime
    \end{bmatrix}^\T\in\cB_{1}^{(q,k)}(\sigma)\setminus\{\sigma\}$. It differs from $\sigma$ in exactly one component. Suppose that the single substitution occurs at position $i$, where $1\le i\le k$. Then $a_j^\prime=a_j$ for every $j\ne i$. For $\sigma^\prime$ to remain nondecreasing, the new value $a_i^\prime$ must satisfy $a_{i-1}\le a_i^\prime\le a_{i+1}$, where we set $a_{0}=0$ and $a_{k+1}=q-1$. Excluding the original value $a_i$, there are exactly $a_{i+1}-a_{i-1}$ possible choices for $a_i^\prime$. Summing over all
    $i$ gives $\abs{\cB_{1}^{(q,k)}(\sigma)\setminus\{\sigma\}}=\sum_{i=1}^{k}(a_{i+1}-a_{i-1})=q-1+a_{k}-a_1$, which establishes \eqref{eq_subsize1}.

    To prove \eqref{eq_subsize2}, consider $\sigma^\prime=\sparenv{a_1^\prime,\cdots,a_{k}^\prime}^\T\in\cB_{(1,1,\ldots,1)}^{(q,k)}(\sigma)\setminus\{\sigma\}$. For each fixed value $b\in\Sigma_q$ for $a_{k}^\prime$, for $\sigma^\prime$ to be nondecreasing, the tuple $\sparenv{a_1^\prime,\cdots,a_{k-1}^\prime}^\T$ must be a nondecreasing vector over $\Sigma_{b+1}$. The number of such tuples is $\binom{b+k-1}{b}$. Summing over all $b\in\Sigma_q$ and excluding the original symbol $\sigma$, we obtain that $\abs{\cB_{(1,1,\ldots,1)}^{(q,k)}(\sigma)\setminus\{\sigma\}}=\sum_{\ell=0}^{q-1}\binom{\ell+k-1}{\ell}-1$. Thus part (i) is proved.

    For part (ii), fix $m$ columns satisfying $\bms_k[i]-\bms_1[i]\ge\ell$ and choose $e$ of them. By \eqref{eq_subsize1}, each selected column can be changed into at least $q-1+\ell$ distinct valid columns by a single substitution. Leaving all other columns unchanged, we obtain at least $\binom{m}{e}(q-1+\ell)^e$ distinct sequences in $\cB_e^{(q,k)}(\bms)\setminus\{\bms\}$. Indeed, the resulting sequence uniquely determines the changed columns and their replacement values. Including $\bms$ itself gives \eqref{eq_subsize3}.

    For part (iii), choose any $e_k$ columns and leave the remaining columns unchanged. By \eqref{eq_subsize2}, each selected column can be replaced by any of $\sum_{\ell=1}^{q-1}\binom{\ell+k-1}{\ell}$ valid columns different from the original one. This introduces at most $e_k\le e_i$ substitutions in row $i$ for every $i\in[k]$, so every resulting sequence belongs to $\cB_{(e_1,\ldots,e_k)}^{(q,k)}(\bms)$. Since each selected column is changed, the resulting sequence uniquely determines the selected columns and their replacement values. This implies that there are at least $\binom{n}{e_k} \sparenv{\sum_{\ell=1}^{q-1}\binom{\ell+k-1}{\ell}}^{e_k}$ distinct sequences in $\cB_{(e_1,\ldots,e_k)}^{(q,k)}(\bms)\setminus\{\bms\}$. Including $\bms$ itself gives \eqref{eq_subsize4}.
\end{IEEEproof}

Upper bounds on the sizes of CECCs are presented in \Cref{thm_boundCECC1,thm_boundCECC2,thm_boundCECC3}. Except for \eqref{eq_subupperbound1} and \eqref{eq_subupperbound2}, all bounds are derived using the following common strategy. Choose a subset $\cA\subseteq\Phi_{q,k}$. Partition the code $\cC$ into $\cC_1$ and $\cC_2$, where $\cC_1$ consists of codewords that contain a ``large" number of symbols from $\cA$ and $\cC_2=\cC\setminus\cC_1$. This partition is designed so that error balls centered at codewords in $\cC_1$ are sufficiently large. A sphere-packing argument then yields an upper bound $\abs{\cC_1}\le U_n$, where $n$ is the code length. We further show that $\abs{\cC_2}=o\parenv{U_n}$. Consequently, $\abs{\cC}=\abs{\cC_1}+\abs{\cC_2}\le U_n(1+o(1))$. The key steps lie in constructing an appropriate subset $\cA$ and estimating lower bounds on the sizes of error balls centered at codewords in $\cC_1$. Details are provided in the proofs of \eqref{eq_subupperbound3}, \eqref{eq_subuppergeneral} and \eqref{eq_subupperbound4}--\eqref{eq_subupperbound6}.

For estimating the size of $\abs{\cC_2}$, we employ Hoeffding's inequality introduced below.
\begin{lemma}(\cite[Theorem 2.1]{LucGabor2001}, Hoeffding's inequality)\label{lem_hoeffding}
    Let $X_1,\ldots,X_n$ be $n$ mutually independent random variables with $a_i\le X_i\le b_i$ for all $i$. Let $X=\sum_{i=1}^nX_i$. Then for all $t>0$, we have
    \begin{equation*}
        \Pr[X-\E[X]\le -t]\le\e^{-\frac{2t^2}{\sum_{i=1}^{n}(b_i-a_i)^2}},
    \end{equation*}
    where $\e$ is the base of the natural logarithm.
\end{lemma}

The following fact is useful in our calculations.
\begin{fact}\label{fac_binom}
    It holds that $\binom{N}{r}=\frac{N^r(1+o(1))}{r!}$ for fixed $r\ge1$ and sufficiently large $N$.
\end{fact}

We now present the first main result of this section.
\begin{theorem}\label{thm_boundCECC1}
    Let $q,k\ge2$. Recall that $Q_{q,k}=\abs{\Phi_{q,k}}=\binom{k+q-1}{q-1}$.
    \begin{enumerate}[$(i)$]
        \item For integers $e_1,e_2,\ldots,e_{k}$ satisfying $e_1\ge\cdots\ge e_{k}\ge1$ and integer $e\ge 1$, it holds that
    \begin{equation}\label{eq_subupperbound1}
        A_{q,k}^S(n;(e_1,\ldots,e_{k}))\le\frac{Q_{q,k}^n}{\binom{n}{e_{k}}\sparenv{\sum_{\ell=1}^{q-1}\binom{\ell+k-1}{\ell}}^{e_{k}}+1}
    \end{equation}
    and
    \begin{equation}\label{eq_subupperbound2}
        A_{q,k}^S(n;e)\le\frac{Q_{q,k}^n}{\binom{n}{e}(q-1)^e+1}.
    \end{equation}
    \item For $e\ge 1$ and $1\le\ell\le q-1$, let $n_0=Q_{q,k}-(q-\ell)\binom{\ell+k-2}{k-1}-\binom{\ell+k-1}{k}$.\footnote{As will be shown in the proof, $n_0$ is a positive integer.} Then, as $n\to\infty$,
    \begin{equation}\label{eq_subupperbound3}
        A_{q,k}^S(n;e)\le\frac{Q_{q,k}^{n+e}e!}{\sparenv{n_0(q-1+\ell)}^en^e}(1+o(1)).
    \end{equation}
    \end{enumerate}
\end{theorem}
\begin{IEEEproof}
    The upper bounds in \eqref{eq_subupperbound1} and \eqref{eq_subupperbound2} follow from a standard sphere‑packing argument applied to the lower bounds on the ball sizes in \eqref{eq_subsize4} and \eqref{eq_subsize3} (setting $\ell=0$ and $m=n$), respectively. Notice that if $e_k>n$ or $e>n$, the bound in \eqref{eq_subupperbound1} or \eqref{eq_subupperbound2} reduces to the trivial bound $Q_{q,k}^n$. Next, we prove part (ii).

    Let $\Phi_{q,k}^\prime$ be the subset consisting of all symbols $\sigma=\sparenv{a_1,\cdots,a_{k}}^\T\in\Phi_{q,k}$ that satisfy $a_{k}-a_1\ge \ell$. We first show that $\abs{\Phi_{q,k}^{\prime}}=n_0$. Consider a symbol
    $\sigma=\sparenv{a_1,\ldots,a_k}^{\T}\in\Phi_{q,k}$ with $a_k-a_1<\ell$. If $a_k<\ell$, then all its entries belong to $\{0,\ldots,\ell-1\}$.
    There are $\binom{\ell+k-1}{k}$ such nondecreasing vectors. For each $b\in[\ell,q-1]$, the symbols $\sigma$ satisfying $a_k=b$ and $a_k-a_1<\ell$ are precisely those whose first $k-1$ entries form a nondecreasing vector over $\{b-\ell+1,\ldots,b\}$. There are $\binom{\ell+k-2}{k-1}$ such symbols. These cases are disjoint, giving $\abs{\Phi_{q,k}\setminus\Phi_{q,k}^\prime}=(q-\ell)\binom{\ell+k-2}{k-1}+\binom{\ell+k-1}{k}$ and consequently, $\abs{\Phi_{q,k}^\prime}=n_0$. Since $\Phi_{q,k}^\prime$ is not empty, $n_0$ is a positive integer.

    Define $m=\frac{n_0n}{Q_{q,k}}-\sqrt{en\ln{n}}$. For simplicity, we suppress integer rounding in the definition of $m$, which does not affect the asymptotic bound in \eqref{eq_subupperbound3}. Let $\cC\subseteq\Phi_{q,k}^n$ be a $k$-resolution $e$-CECC. Split $\cC$ into two parts:
    \begin{equation*}
        \cC_1=\mathset{\bmc\in\cC:\bmc\text{ contains at least }m\text{ symbols from }\Phi_{q,k}^\prime},\quad \cC_2=\cC\setminus\cC_1.
    \end{equation*}
    We bound from above the sizes of $\cC_1$ and $\cC_2$ separately.

    First, it follows from \eqref{eq_subsize3} that $\abs{\cB_{e}^{(q,k)}(\bmc)}\ge\binom{m}{e}(q-1+\ell)^e+1$ for all $\bmc\in\cC_1$. Notice that $\cC_1$ is also a $k$-resolution $e$-CECC, which implies $\cB_{e}^{(q,k)}(\bmc)\cap\cB_{e}^{(q,k)}(\bmc^\prime)=\emptyset$ for any distinct $\bmc,\bmc^\prime\in\cC_1$. Hence,
    \begin{equation}\label{eq_boundC11}
        \abs{\cC_1}\le\frac{Q_{q,k}^n}{\binom{m}{e}(q-1+\ell)^e+1}\le\frac{Q_{q,k}^{n+e}e!}{\sparenv{n_0(q-1+\ell)}^en^e}(1+o(1)).
    \end{equation}
    The second inequality in \eqref{eq_boundC11} follows from \Cref{fac_binom} and $m=\frac{n_0n}{Q_{q,k}}(1+o(1))$ for sufficiently large $n$ and fixed $q,k,e$ and $\ell$.

    Next, we bound the size of $\cC_2$. To this end, we estimate the number of sequences $\bms\in\Phi_{q,k}^n$ containing fewer than $m$ symbols from $\Phi_{q,k}^\prime$. Choose $\bms$ uniformly at random from $\Phi_{q,k}^n$. Define indicator random variables $X_1,\ldots,X_n$ by
    \begin{equation*}
        X_i=
        \begin{cases}
            1,\mbox{ if }\bms[i]\in\Phi_{q,k}^\prime,\\
            0,\mbox{ otherwise},
        \end{cases}
    \end{equation*}
    and set $X=X_1+\cdots+X_n$, which is the number of symbols in $\Phi_{q,k}^\prime$ contained in $\bms$.
    Because $\bms$ is chosen uniformly at random, it holds that each $\bms[i]$ is chosen randomly and uniformly from $\Phi_{q,k}$, and that $X_1,\ldots,X_n$ are mutually independent. Therefore, we have $\E[X]=\sum_{i=1}^{n}\E[X_i]=\sum_{i=1}^{n}\Pr[X_i=1]=\frac{n_0n}{Q_{q,k}}$. Applying \Cref{lem_hoeffding} gives
    \begin{align*}
        \Pr[X<m]=\Pr\sparenv{X-\E[X]<-\sqrt{en\ln{n}}}\le \e^{-\frac{2en\ln{n}}{n}}=\frac{1}{n^{2e}}.
    \end{align*}
    Consequently, we have $\abs{\cC_2}\le\abs{\mathset{\bms\in\Phi_{q,k}^n:\bms\text{ contains fewer} than }m\text{ symbols from }\Phi_{q,k}^\prime}\le\abs{\Phi_{q,k}^n}\cdot\Pr[X<m]\le\frac{Q_{q,k}^n}{n^{2e}}$. Combining this bound with \eqref{eq_boundC11} for $\abs{\cC_1}$ completes the proof.
\end{IEEEproof}

\begin{remark}
We explain how to optimize $\ell$ in \eqref{eq_subupperbound3}. For fixed $q,k$, define
\[
 \begin{aligned}
 n_0(\ell)&=Q_{q,k}-(q-\ell)\binom{\ell+k-2}{k-1}
                    -\binom{\ell+k-1}{k},\\
 F_{q,k}(\ell)&=n_0(\ell)(q-1+\ell),\qquad 1\le\ell\le q-1.
 \end{aligned}
\]
Choose any $\ell^*\in\operatorname*{arg\,max}_{1\le\ell\le q-1}F_{q,k}(\ell)$. Since the other factors in the bound are independent of $\ell$, the optimized form is
\[
 A_{q,k}^S(n;e)\le
 \frac{Q_{q,k}^{n+e}e!}
 {n^e\bigl[\max_{1\le\ell\le q-1}F_{q,k}(\ell)\bigr]^e}
 (1+o(1)).
\]
This requires evaluating only $q-1$ positive integers, independently of $n$ and $e$. For $q=2$, the only choice is $\ell=1$. As a nonbinary example, for $(q,k)=(4,6)$, the values of $F_{4,6}(1),F_{4,6}(2),F_{4,6}(3)$ are $320,325,210$, respectively, so $\ell^*=2$.
\end{remark}

The bound in \eqref{eq_subupperbound1} is limited as it does not depend on $e_1,\ldots,e_{k-1}$ and relies heavily on the condition $e_1\ge\cdots\ge e_{k}\ge1$. The following theorem provides a sharper bound for sufficiently large $n$ under more general assumptions.
\begin{theorem}\label{thm_boundCECC2}
     Suppose that $k,q\ge 2$. Let $e_1,\ldots,e_{k}$ be $k$ nonnegative integers. Denote the multiset of nonzero integers among them by $\multiset{e_{\ell_1},\ldots,e_{\ell_m}}$, where $1\le m\le q$ and $\ell_1<\cdots<\ell_m$. Define $R=\mathset{1<j<m:\ell_j-\ell_{j-1}}=1$ and $e=\sum_{j=1}^me_{\ell_j}$. Then
     \begin{equation}\label{eq_subuppergeneral}
        A_{q,k}^S(n;(e_1,\ldots,e_{k}))\le\frac{Q_{q,k}^{n+e}\prod_{j=1}^{m}e_{\ell_j}!}{2^{\abs{R}}(q-1)^{e_{\ell_m}}}\binom{q}{m}^{e-e_{\ell_m}}n^{e}\parenv{1+o(1)}.
    \end{equation}
\end{theorem}
\begin{IEEEproof}
Let $B_1,\ldots,B_{\binom{q}{m}}$ be all subsets of $\Sigma_q$ of size $m$. For each $i$, write $B_i=\mathset{b_{i,1},b_{i,2},\ldots,b_{i,m}}$ with $b_{i,1}<b_{i,2}<\cdots<b_{i,m}$. Construct $\binom{q}{m}$ column vectors $\sigma_1,\ldots,\sigma_{\binom{q}{m}}$ of length $k$ over $\Sigma_q$ as follows: for $1\le i\le\binom{q}{m}$, set $\sigma_i=\sparenv{a_{i,1},a_{i,2},\cdots,a_{i,k}}^\T$, where
    \begin{equation}\label{eq_sigma}
        \begin{array}{c}
            a_{i,\ell_{j-1}+1}=\cdots=a_{i,\ell_j}=b_{i,j}\;(1\le j\le m),\\
            a_{i,\ell_{m}+1}=\cdots=a_{i,k}=q-1.
        \end{array}
    \end{equation}
    Here, we set $\ell_0=0$. Also, define a special vector $\sigma=\sparenv{a_{1},\cdots,a_{k}}^\T$, where $a_1=\cdots=a_{\ell_m}=0$ and $a_{\ell_m+1}=\cdots=a_{k}=q-1$.
    It is easy to see that $\sigma_i$'s and $\sigma$ are nondecreasing vectors and thus, symbols in $\Phi_{q,k}$. Denote $\cA=\mathset{\sigma_1,\ldots,\sigma_{\binom{q}{m}}}$. Then $\sigma\notin\cA$ if $m\ge2$.

    Denote $n_1=\frac{n}{Q_{q,k}}\binom{q}{m}-\sqrt{en\ln{n}}$ and $n_2=\frac{n}{Q_{q,k}}-\sqrt{en\ln{n}}$. Let $\cC\subseteq\Phi_{q,k}^n$ be an $\parenv{e_1,\ldots,e_{k}}$-CECC. Partition it into two parts:
    \begin{equation*}
        \cC_1=\mathset{\bmc\in\cC:\bmc\text{ has }\ge n_1\text{ symbols from }\cA\text{ and }\ge n_2\text{ copies of }\sigma},\quad \cC_2=\cC\setminus\cC_1.
    \end{equation*}
    Next, we bound $\abs{\cC_1}$ and $\abs{\cC_2}$ from above.

    For a codeword $\bmc\in\cC_1$, let $I=\mathset{s\in[n]:\bmc[s]\in\cA}$ and $J=\mathset{s\in[n]:\bmc[s]=\sigma}$. Then $\abs{I}\ge n_1$ and $\abs{J}\ge n_2$ by definition of $\cC_1$. Since $n_1$ and $n_2$ grow with $n$ while $e_1,\ldots,e_k$ are fixed as assumed at the end of \Cref{sec_pre}, we have $n_1\ge e$ and $n_2\ge e_{\ell_m}$ when $n$ is sufficiently large. Select $m-1$ mutually disjoint subsets $I_1,\ldots,I_{m-1}$ of $I$ with $\abs{I_j}=e_{\ell_j}$ for each $1\le j< m$. And select one subset $J_1$ of $J$ with $\abs{J_1}=e_{\ell_m}$. There are at least $\binom{n_1}{e_{\ell_1},\ldots,e_{\ell_{m-1}},n_1-\sum_{j=1}^{m-1}e_{\ell_j}}\binom{n_2}{e_{\ell_m}}$ choices\footnote{Here, $\binom{n_1}{e_{\ell_1},\ldots,e_{\ell_{m-1}},n_1-\sum_{j=1}^{m-1}e_{\ell_j}}$ is a multinomial coefficient.} for tuple $\parenv{I_1,\ldots,I_{m-1},J_1}$. Since $\sigma\notin\cA$ when $m\ge2$, it holds that $J_1$ is disjoint from $\cup_{j=1}^{m-1}I_j$. For each choice of $\parenv{I_1,\ldots,I_{m-1},J_1}$, construct a sequence $\bmc^\prime$ from $\bmc$ as follows:
    \begin{itemize}
        \item For each $j\in[m-1]\setminus R$ and $s\in I_{j}$, change the $\ell_j$-th component of $\bmc[s]$ (which is $b_{i,j}$ for some $i$) to $b_{i,j}+1$.
        \item For each $j\in R$ and $s\in I_j$, change the $\ell_j$-th component of $\bmc[s]$ (which is $b_{i,j}$ for some $i$) to $b_{i,j}+1$ or $b_{i,j}-1$. There are two choices.
        \item For each $s\in J_1$, we have $\bmc[s]=\sigma$. Change the $\ell_m$-th component of $\bmc[s]$ from $0$ to any symbol in $\Sigma_q\setminus\{0\}$. There are $q-1$ choices.
    \end{itemize}
    Because $b_{i,1}<b_{i,2}<\cdots<b_{i,m-1}<b_{i,m}\le q-1$ and $\ell_j-\ell_{j-1}=1$ for each $j\in R$, each $\bmc^\prime[s]$ ($s\in[n]$) is still nondecreasing. Hence, $\bmc^\prime\in\cB_{(e_1,\ldots,e_{k})}^{(q,k)}(\bmc)$. As a result,
    \begin{equation*}
        \abs{\cB_{(e_1,\ldots,e_{k})}^{(q,k)}(\bmc)}\ge\binom{n_1}{e_{\ell_1},\ldots,e_{\ell_{m-1}}},n_1-\sum_{j=1}^{m-1}e_{\ell_j}\binom{n_2}{e_{\ell_m}}2^{\abs{R}}(q-1)^{e_{\ell_m}}.
    \end{equation*}
    Since $\cC_1$ is an $\parenv{e_1,\ldots,e_{k}}$-CECC, we obtain
    \begin{equation*}
        \abs{\cC_1}\le\frac{Q_{q,k}^n}{\binom{n_1}{e_{\ell_1},\ldots,e_{\ell_{m-1}}},n_1-\sum_{j=1}^{m-1}e_{\ell_j}}\binom{n_2}{e_{\ell_m}}2^{\abs{R}}(q-1)^{e_{\ell_m}}=\frac{Q_{q,k}^{n+e}\prod_{j=1}^{m}e_{\ell_j}!}{2^{\abs{R}}(q-1)^{e_{\ell_m}}}\binom{q}{m}^{e-e_{\ell_m}}n^{e}\parenv{1+o(1)}.
    \end{equation*}
    The equality follows from the fact that, with all parameters other than $n$ fixed, we have
    \begin{align*}
    &\binom{n_1}
    {e_{\ell_1},\ldots,e_{\ell_{m-1}},
    n_1-\sum_{j=1}^{m-1}e_{\ell_j}}
    \binom{n_2}{e_{\ell_m}}=
    \frac{n_1^{e-e_{\ell_m}}n_2^{e_{\ell_m}}}
    {\prod_{j=1}^{m}e_{\ell_j}!}(1+o(1))=
    \frac{\binom{q}{m}^{e-e_{\ell_m}}n^e}
    {Q_{q,k}^e\prod_{j=1}^{m}e_{\ell_j}!}(1+o(1)), \end{align*}
    where first equality follows from \Cref{fac_binom}, and the second follows from the fact of $n_1=\frac{n}{Q_{q,k}}\binom{q}{m}(1+o(1))$ and $n_2=\frac{n}{Q_{q,k}}(1+o(1))$.

    Next, we upper-bound $\abs{\cC_2}$. If $\bmc\in\cC_2$, it contains either fewer than $n_1$ symbols from $\cA$, or fewer than $n_2$ copies of $\sigma$. Define
    \begin{equation*}
        \begin{aligned}
            \cB_1=\mathset{\bms\in\Phi_{q,k}^n:\bms\text{ contains fewer than }n_1\text{ symbols from }\cA},\\
            \cB_2=\mathset{\bms\in\Phi_{q,k}^n:\bms\text{ contains fewer than }n_2\text{ copies of }\sigma}.
        \end{aligned}
    \end{equation*}
    Then $\cC_2\subseteq\cB_1\cup\cB_2$, and it is sufficient to bound $\abs{\cB_1}$ and $\abs{\cB_2}$.

    Choose a sequence $\bms$ uniformly at random from $\Phi_{q,k}^n$. Define indicator random variables $X_1,\ldots,X_n$ by
    \begin{equation*}
        X_i=
        \begin{cases}
            1,\mbox{ if }\bms[i]\in\cA,\\
            0,\mbox{ otherwise},
        \end{cases}
    \end{equation*}
    and set $X=X_1+\cdots+X_n$, which is the total number of symbols from $\cA$ in $\bms$.
    Since $\bms$ is chosen randomly and uniformly from $\Phi_{q,k}^n$, it holds that each $\bms[i]$ is chosen randomly and uniformly from $\Phi_{q,k}$, and that $X_1,\ldots,X_n$ are mutually independent. Therefore, we have $\E[X]=\sum_{i=1}^{n}\E[X_i]=\sum_{i=1}^{n}\Pr[X_i=1]=\frac{n}{Q_{q,k}}\binom{q}{m}$. Then it follows from \Cref{lem_hoeffding} that
    \begin{align*}
        \Pr[X<n_1]=\Pr\sparenv{X-\E[X]<-\sqrt{en\ln{n}}}\le \e^{-\frac{2en\ln{n}}{n}}=\frac{1}{n^{2e}}.
    \end{align*}
    Hence, we have $\abs{\cB_1}\le\abs{\Phi_{q,k}^n}\cdot\Pr[X<n_1]\le\frac{Q_{q,k}^n}{n^{2e}}$. Similarly, we can show that $\abs{\cB_2}\le\frac{Q_{q,k}^n}{n^{2e}}$. Consequently, $\abs{\cC_2}\le 2\frac{Q_{q,k}^n}{n^{2e}}$. Combining the bounds for $\abs{\cC_1}$ and $\abs{\cC_2}$ gives the overall asymptotic bound stated in the theorem.
\end{IEEEproof}

In the proof of \Cref{thm_boundCECC2}, we selected $m-1$ mutually disjoint subsets $I_1,\ldots,I_{m-1}$ of $I$ and then modified the $\ell_j$-th component of $\bmc[s]$ for $s\in I_j$. We cannot alter the $\ell_m$-th component by selecting another subset $I_m$ because it is possible that $b_{i,m}=q-1$ and $\ell_{m}-\ell_{m-1}>1$. To alter the $\ell_m$-th component, we introduced the additional symbol $\sigma$. In the next theorem, we present improvements upon the bound in \eqref{eq_subuppergeneral} when $\ell_{m}-\ell_{m-1}=1$ or $m=1$.
\begin{theorem}\label{thm_boundCECC3}
    Suppose that $k,q\ge 2$. Let $e_1,\ldots,e_{k}$ be $k$ nonnegative integers. Denote the multiset of nonzero integers among them by $\multiset{e_{\ell_1},\ldots,e_{\ell_m}}$, where $1\le m\le q$ and $\ell_1<\cdots<\ell_m$. Let $R=\mathset{1<j<m:\ell_j-\ell_{j-1}}=1$ and $e=\sum_{j=1}^me_{\ell_j}$.
    \begin{enumerate}[$(i)$]
        \item If $m\ge 2$ and $\ell_m-\ell_{m-1}=1$, the following bound holds:
    \begin{equation}\label{eq_subupperbound4}
        A_{q,k}^S(n;(e_1,\ldots,e_{k}))\le\frac{Q_{q,k}^{n+e}\prod_{j=1}^me_{\ell_j}!}{2^{\abs{R}}\parenv{\binom{q}{m}n}^{e}}\parenv{1+o(1)}.
    \end{equation}
    \item If $m=1$, then
    \begin{equation}\label{eq_subupperbound5}
        A_{q,k}^S(n;(e_1,\ldots,e_{k}))\le\frac{Q_{q,k}^{n+e}e!}{\sparenv{nq(q-1)}^{e}}\parenv{1+o(1)}.
    \end{equation}
    \item If $m=2$ and $\ell_m-\ell_{m-1}=1$, then
    \begin{equation}\label{eq_subupperbound6}
        A_{q,k}^S(n;(e_1,\ldots,e_{k}))\le\frac{Q_{q,k}^{n+e}(e_{\ell_1}}!)(e_{\ell_2}!){\sparenv{\binom{q}{2}+1}^en^{e}}\parenv{1+o(1)}.
    \end{equation}
    \end{enumerate}
\end{theorem}
\begin{IEEEproof}
   (i) Construct $\binom{q}{m}$ column vectors $\sigma_1,\ldots,\sigma_{\binom{q}{m}}$ of length $k$ over $\Sigma_q$ as in \eqref{eq_sigma} and let $\cA=\mathset{\sigma_1,\ldots,\sigma_{\binom{q}{m}}}$. Denote $n_0=\frac{n}{Q_{q,k}}\binom{q}{m}-\sqrt{en\ln{n}}$. Let $\cC\subseteq\Phi_{q,k}^n$ be an $\parenv{e_1,\ldots,e_{k}}$-CECC. Partition $\cC$ into two parts:
    \begin{equation*}
        \cC_1=\mathset{\bmc\in\cC:\bmc\text{ contains at least }n_0\text{ symbols from }\cA},\quad \cC_2=\cC\setminus\cC_1.
    \end{equation*}
    We will bound $\abs{\cC_1}$ and $\abs{\cC_2}$ from above separately.

    For a codeword $\bmc\in\cC_1$, let $I=\mathset{s\in[n]:\bmc[s]\in\cA}$. Then $\abs{I}\ge n_0$ by definition of $\cC_1$. Notice that $n_0$ grows linearly with $n$. This implies $n_0\ge e$ for sufficiently large $n$. Select $m$ mutually disjoint subsets $I_1,\ldots,I_m$ of $I$ with $\abs{I_j}=e_{\ell_j}$ for each $1\le j\le m$. Construct a sequence $\bmc^\prime$ as follows:
    \begin{itemize}
        \item For each $j\ne m$ and $s\in I_{j}$, change the $\ell_j$-th component of $\bmc[s]$ as in the proof of \Cref{thm_boundCECC2}.
        \item For each $s\in I_{m}$, change the $\ell_m$-th component of $\bmc[s]$ (which is $b_{i,m}$ for some $i$) to $b_{i,m}-1$.
    \end{itemize}
    Because $b_{i,1}<b_{i,2}<\cdots<b_{i,m}$ and $\ell_j-\ell_{j-1}=1$ for each $j\in R\cup\{m\}$, each $\bmc^\prime[s]$ is still nondecreasing. Hence, $\bmc^\prime\in\cB_{(e_1,\ldots,e_{k})}^{(q,k)}(\bmc)$. As a result, it holds that
    \begin{equation*}
        \abs{\cB_{(e_1,\ldots,e_{k})}^{(q,k)}(\bmc)}\ge\binom{n_0}{e_{\ell_1}},\ldots,e_{\ell_m},n_0-e2^{\abs{R}}.
    \end{equation*}
    Since $\cC_1$ is an $\parenv{e_1,\ldots,e_{k}}$-CECC, we obtain
    \begin{equation*}
        \abs{\cC_1}\le\frac{Q_{q,k}^n}{\binom{n_0}{e_{\ell_1}},\ldots,e_{\ell_m},n_0-e}2^{\abs{R}}=\frac{Q_{q,k}^{n+e}\prod_{j=1}^me_{\ell_j}!}{2^{\abs{R}}\parenv{\binom{q}{m}n}^{e}}\parenv{1+o(1)}.
    \end{equation*}

    Similar to the proofs of \Cref{thm_boundCECC1,thm_boundCECC2}, we can prove that $\abs{\cC_2}\le\frac{Q_{q,k}^n}{n^{2e}}$ by applying \Cref{lem_hoeffding}. Now the proof of part (i) is completed.

    \vspace{5pt}
    (ii) When $m=1$, according to \Cref{pro_subequal}, it suffices to assume that $e_1=e>0$ and $e_2=\cdots=e_{k}=0$. The proof is almost the same as that of part (i). The only difference lies in the lower bound of $\abs{\cB_{(e,0,\ldots,0)}^{(q,k)}(\bmc)}$, where $\bmc\in\cC_1$. In this case, for each $i\in I_1$, there are $q-1$ possible directions for changing the first component of $\bmc[i]$. This implies that $\abs{\cB_{(e,0,\ldots,0)}^{(q,k)}(\bmc)}\ge\binom{n_0}{e}(q-1)^e$.

    \vspace{5pt}
    (iii) Let $\cA_0=\mathset{\sigma_1,\ldots,\sigma_{\binom{q}{2}}}$, where $\binom{q}{2}$ symbols $\sigma_1,\ldots,\sigma_{\binom{q}{2}}\in\Phi_{q,k}$ are defined as in \eqref{eq_sigma}. We further define two symbols $\tau_j=\begin{bmatrix}
        a_1^{(j)},\ldots,a_{k}^{(j)}
    \end{bmatrix}^\T$ ($j=1,2$) as
    \begin{equation*}
        \begin{array}{l}
           a_{r}^{(j)}=0,\;(j=1,2,1\le r<\ell_1), \\
           a_{r}^{(j)}=q-1,\;(j=1,2,\ell_2<r\le k),\\
           a_{\ell_1}^{(1)}=a_{\ell_2}^{(1)}=q-1,\\
           a_{\ell_1}^{(2)}=a_{\ell_2}^{(2)}=0.
        \end{array}
    \end{equation*}
    For $j=1,2$, let $\cA_j=\mathset{\tau_j}$.  Denote $n_0=\frac{n}{Q_{q,k}}\binom{q}{2}-\sqrt{en\ln{n}}$ and $n_1=n_2=\frac{n}{Q_{q,k}}-\sqrt{en\ln{n}}$. Let $\cC\subseteq\Phi_{q,k}^n$ be an $\parenv{e_1,\ldots,e_{k}}$-CECC. Partition $\cC$ into two subsets $\cC_1=\mathset{\bmc\in\cC:\bmc\text{ contains at least }n_j\text{ symbols from }\cA_j,\forall 0\le j\le 2}$ and $\cC_2=\cC\setminus\cC_1$. An application of \Cref{lem_hoeffding} leads to $\abs{\cC_2}\le\frac{3Q_{q,k}^{n}}{n^{2e}}$. Next, we bound $\abs{\cC_1}$.

    For a codeword $\bmc\in\cC_1$, let $I_j=\mathset{s\in[n]:\bmc[s]\in\cA_j}$. Then $\abs{I_j}\ge n_j$ for all $0\le j\le 2$. Notice that $n_0+n_1\ge e_{\ell_1}$ and $n_0+n_2-e_{\ell_1}\ge e_{\ell_2}$ for sufficiently large $n$. Select one subset $J_1$ of $I_0\cup I_1$ with $\abs{J_1}=e_{\ell_1}$. Then select one subset $J_2$ of $\parenv{I_0\cup I_2}\setminus J_1$ with $\abs{J_2}=e_{\ell_2}$. There are at least $\binom{n_0+n_1}{e_{\ell_1}}\binom{n_0+n_2-e_{\ell_1}}{e_{\ell_2}}$ different choices for $\parenv{J_1,J_2}$. For each choice of $\parenv{J_1,J_2}$, construct a sequence $\bmc^\prime$ from $\bmc$ as follows:
    \begin{itemize}
        \item For each $s\in J_{1}$, if $s\in I_0$, change the $\ell_1$-th component of $\bmc[s]$ (which is $b_{i,1}$ for some $i$) to $b_{i,1}+1$; if $s\in I_1$, change the $\ell_1$-th component of $\bmc[s]$ (which is $q-1$) to $0$.
        \item For each $s\in J_{2}$, if $s\in I_0$, change the $\ell_2$-th component of $\bmc[s]$ (which is $b_{i,2}$ for some $i$) to $b_{i,2}-1$; if $s\in I_2$, change the $\ell_2$-th component of $\bmc[s]$ (which is $0$) to $q-1$.
    \end{itemize}
    Because $b_{i,1}<b_{i,2}$, $\ell_2-\ell_{1}=1$, $a_{\ell_1}^{(1)}=a_{\ell_2}^{(1)}=q-1$ and $a_{\ell_1}^{(2)}=a_{\ell_2}^{(2)}=0$, each $\bmc^\prime[s]$ is still nondecreasing. Hence, $\bmc^\prime\in\cB_{(e_1,\ldots,e_{k})}^{(q,k)}(\bmc)$. As a result, it holds that
    \begin{equation*}
        \abs{\cB_{(e_1,\ldots,e_{k})}^{(q,k)}(\bmc)}\ge\binom{n_0+n_1}{e_{\ell_1}}\binom{n_0+n_2-e_{\ell_1}}{e_{\ell_2}}.
    \end{equation*}
    Since $\cC_1$ is an $\parenv{e_1,\ldots,e_{k}}$-CECC, we obtain
    \begin{equation*}
        \abs{\cC_1}\le\frac{Q_{q,k}^n}{\binom{n_0+n_1}{e_{\ell_1}}}\binom{n_0+n_2-e_{\ell_1}}{e_{\ell_2}}=\frac{Q_{q,k}^{n+e}(e_{\ell_1}}!)(e_{\ell_2}!){\sparenv{\binom{q}{2}+1}^en^{e}}\parenv{1+o(1)}.
    \end{equation*}
    Now the proof is completed.
\end{IEEEproof}

Notice that when $e$ is even, a $k$-resolution $e$-CECC is also a $k$-resolution $(e/2,e/2,0,\ldots,0)$-CECC. Then \eqref{eq_subupperbound6} immediately yields the following result.
\begin{corollary}\label{cor_evene}
    If $e>0$ is an even integer, then
    \begin{equation}\label{eq_evene}
        A_{q,k}^S(n;e)\le\frac{Q_{q,k}^{n+e}2^e\parenv{\parenv{\frac{e}{2}}!}^2}{\sparenv{q^2-q+2}^en^e}(1+o(1)).
    \end{equation}
\end{corollary}
\begin{remark}
\begin{itemize}
    \item Bounds \eqref{eq_subupperbound3}, \eqref{eq_subuppergeneral}, \eqref{eq_subupperbound4}--\eqref{eq_subupperbound6} and \eqref{eq_evene} hold in the regime where $n$ is sufficiently large. When $n$ is small, the simpler sphere-packing bounds in \eqref{eq_subupperbound1} (or \eqref{eq_subupperbound2}) might be superior.
    \item The bound in \eqref{eq_subuppergeneral} is $\parenv{\binom{q}{m}/(q-1)}^{e_{\ell_m}}$ times the bound in \eqref{eq_subupperbound4}, implying that the latter bound is better when $m<q$.
    \item When $m=1$, the bound in \eqref{eq_subuppergeneral} is $q^e$ times the bound in \eqref{eq_subupperbound5}.
    \item The bound in \eqref{eq_subupperbound3} holds for all $e$ while the bound in \eqref{eq_evene} only holds for even $e$. When $e$ is even, it is possible that the latter bound outperforms the former one. For example, when $q=k=2$, the latter bound is $\frac{Q_{q,k}^{n+e}\parenv{\frac{e}{2}!}^2}{2^en^e}(1+o(1))$ while the former bound is $\frac{Q_{q,k}^{n+e}e!}{2^en^e}(1+o(1))$ with $\ell=1$. However, when $q=2$ and $k\ge3$, the bound in \eqref{eq_subupperbound3} with $\ell=1$ is better.
\end{itemize}
\end{remark}

\subsection{The case \texorpdfstring{$m>q$}{m>q}}
\label{subsec_mdyq}
%%%%%%%%%%%%%%%%%%%%%%%%%%%%%%%%%%%%%%%%%%%%%%%%%%%%
Although the previous bounds assume $m\le q$, they can be
extended to the case $m>q$ without much effort.
We briefly explain how to extend the bound
\eqref{eq_subuppergeneral} to the case $m>q$.
The bound in \eqref{eq_subupperbound4} can be extended
to the case $m>q$ in a similar way.
Details are omitted for brevity.

Fix an integer $m_0$ with $2\le m_0\le q$.
Starting from $\ell_1$, place every $m_0$ consecutive
indices in one group, with any remaining indices
forming the last group.
Denote the resulting groups by
$(\ell_{p,1},\ldots,\ell_{p,m_p})$, $p\in[h]$.
Then $m_p=m_0$ for $p<h$, and $1\le m_h\le m_0$. For each $p\in[h]$, define
\[
    E_p=\sum_{j=1}^{m_p}e_{\ell_{p,j}},\;
    f_p=e_{\ell_{p,m_p}},\;
    R_p=\mathset{
        1<j<m_p:
        \ell_{p,j}-\ell_{p,j-1}=1
    }.
\]
Denote $e=\sum_{p=1}^h E_p$.

For each group, construct $\cA_p$ using \eqref{eq_sigma}
with $(\ell_1,\ldots,\ell_m)$ replaced by
$(\ell_{p,1},\ldots,\ell_{p,m_p})$.
Thus, $\abs{\cA_p}=\binom{q}{m_p}$.
Also introduce an auxiliary symbol $\sigma^{(p)}$
whose first $\ell_{p,m_p}$ components are $0$
and whose remaining components are $q-1$.
Set
\[
    N_p=\frac{n}{Q_{q,k}}\binom{q}{m_p}
        -\sqrt{en\ln n},\;
    N_*=\frac{n}{Q_{q,k}}-\sqrt{en\ln n}.
\]
Let $\cC\subseteq\Phi_{q,k}^n$ be an
$(e_1,\ldots,e_k)$-CECC.
Let $\cC_1$ consist of the codewords containing at least
$N_p$ symbols from $\cA_p$ and at least $N_*$
copies of $\sigma^{(p)}$, for every $p\in[h]$,
and let $\cC_2=\cC\setminus\cC_1$.
By \Cref{lem_hoeffding} and the union bound,
\[
    \abs{\cC_2}
    \le \frac{2hQ_{q,k}^n}{n^{2e}}
    =o\left(\frac{Q_{q,k}^n}{n^e}\right).
\]

For $\bmc\in\cC_1$, apply the modifications in the proof
of \Cref{thm_boundCECC2} within each group.
Specifically, for $j<m_p$, select $e_{\ell_{p,j}}$
positions containing symbols from $\cA_p$;
for $j=m_p$, select $f_p$ positions containing
$\sigma^{(p)}$.
Require all selected sets of positions to be mutually
disjoint, including those belonging to different groups.
When choosing these sets successively, at most $e$
previously selected positions are excluded at each step.
Since $N_p$ and $N_*$ grow linearly with $n$,
these exclusions do not change the leading term
of the count.

The counting argument in the proof of
\Cref{thm_boundCECC2} therefore gives
\[
    \abs{\cB_{(e_1,\ldots,e_k)}^{(q,k)}(\bmc)}
    \ge
    \left(\frac{n}{Q_{q,k}}\right)^e
    \prod_{p=1}^h
    \frac{
        2^{\abs{R_p}}(q-1)^{f_p}
        \binom{q}{m_p}^{E_p-f_p}
    }{
        \prod_{j=1}^{m_p}e_{\ell_{p,j}}!
    }
    (1-o(1)).
\]
Here, each selected column is changed in only one row,
so different choices give distinct received sequences.
Using the disjointness of the error balls and the bound
on $\abs{\cC_2}$, we obtain
\[
    A_{q,k}^S(n;(e_1,\ldots,e_k))
    \le
    \frac{
        Q_{q,k}^{n+e}\prod_{j=1}^m e_{\ell_j}!
    }{
        n^e
        \prod_{p=1}^h
        \left[
            2^{\abs{R_p}}(q-1)^{f_p}
            \binom{q}{m_p}^{E_p-f_p}
        \right]
    }
    (1+o(1)).
\]
We conclude this section with the following corollary.
\begin{corollary}
    Let $q,k\ge2$ be fixed. As $n\to\infty$, the following redundancy lower bounds hold for codes in $\Phi_{q,k}^n$.
    \begin{enumerate}[$(i)$]
        \item For every fixed nonnegative integer $e$, the redundancy of an $e$-CECC is at least $e\log_{Q_{q,k}}n-O(1)$.
        \item For every fixed tuple of nonnegative integers $(e_1,\ldots,e_k)$, the redundancy of an $(e_1,\ldots,e_k)$-CECC is at least
        $$
        \parenv{\sum_{i=1}^k e_i}\log_{Q_{q,k}}n-O(1).
        $$
        \item For every fixed $1\le t\le k$ and every fixed tuple of nonnegative integers $(e_1,\ldots,e_t)$, the redundancy of a $t$-$(e_1,\ldots,e_t)$-CECC is at least $\parenv{\sum_{i=1}^t e_i}\log_{Q_{q,k}}n-O(1)$.
    \end{enumerate}
\end{corollary}
\begin{IEEEproof}
    Parts (i) and (ii) follow from \eqref{eq_subupperbound2}, \Cref{thm_boundCECC2,thm_boundCECC3}, \Cref{subsec_mdyq}, and the definition of redundancy. Part (iii) follows from part (ii) and the fact that a $t$-$(e_1,\ldots,e_t)$-CECC is also an $(e_1,\ldots,e_k)$-CECC with $e_i=0$ for every $i>t$.
\end{IEEEproof}

\section{Constructions of CECCs}\label{sec_conssub}
%%%%%%%%%%%%%%%%%%%%%%%%%%%%%%%%%%%%%%%%%%%%%%%%%%%%%%%%%
In this section, we construct composite-error correcting codes (CECCs). In \Cref{sec_10CECC}, we establish encoding and decoding algorithms for the binary $k$-resolution $(1,0,\ldots,0)$-CECC constructed in \cite{BesartDollma202509}. In \Cref{sec_1CECC,sec_t1CECC}, we present constructions of $1$-CECCs and $t$-$(1,\ldots,1)$-CECCs, respectively, for any $q\ge2$.

\subsection{\texorpdfstring{$(1,0,\ldots,0)$}{(1,0,\ldots,0)}-CECCs}\label{sec_10CECC}
%%%%%%%%%%%%%%%%%%%%%%%%%%%%%%%%%%%%%%%%%%%%%%%%%%%%%%
Recall that mapping every letter $\sigma\in\Phi_{2,k}$ to $\Wth{1}{\sigma}$ allows us to treat each sequence $\bms\in\Phi_{2,k}^n$ as a sequence in $\Sigma_{k+1}^n$. As discussed in \cite[Section IV]{BesartDollma202509}, a substitution in $\bms_1$ (the first row of $\bms$) induces only one of the following three types of errors in $\bms$: a symbol $k-1$ is replaced by symbol $k$; a symbol $k$ is replaced by symbol $k-1$; a symbol $\sigma\in\mathset{0,1,\ldots,k-2}$ is replaced by an invalid symbol. In the third case, the received column starts with $1,0$, whereas the transmitted column starts with $0,0$.
Replacing the first entry of the received column with $0$ therefore recovers the transmitted column.
Thus, it suffices to correct exchanges between $k-1$ and $k$.

Based on this observation, Dollma \emph{et al} \cite{BesartDollma202509} constructed a binary $k$-resolution $(1,0,\ldots,0)$-CECC as follows. For a sequence $\bms\in\Sigma_{k+1}^n$, delete all symbols $\sigma\in\Sigma_{k+1}\setminus\{k-1,k\}$ from $\bms$. Then replace $k-1$ with $0$, and replace $k$ with $1$. The resulting binary sequence, denoted by $\cF(\bms)$, has length $\ell:=\Wth{k-1}{\bms}+\Wth{k}{\bms}$. Let $\cC(0)=\mathset{\bm{\epsilon}}$ (here, $\bm{\epsilon}$ denotes the empty sequence), $\cC(1)=\{1\}$, $\cC(2)=\{11\}$. For every $\ell\ge 3$, let $\cC(\ell)$ be a binary single-substitution correcting code of length $\ell$. Denote
\begin{equation*}
    \cC_{\mathrm{Doll}}\triangleq\cup_{\ell=0}^{n}\mathset{\bmc\in\Sigma_{k+1}^n:\Wth{k-1}{\bmc}+\Wth{k}{\bmc}=\ell,\cF(\bmc)\in\cC(\ell)}.
\end{equation*}
Then we have the following theorem.
\begin{theorem}(\cite[Theorem 9 and Corollary 3]{BesartDollma202509})
  The code $\cC_{\mathrm{Doll}}$  is a $(1,0,\ldots,0)$-CECC. Moreover, it holds that
  \begin{equation}\label{eq_sizeDoll}
    \abs{\cC_{\mathrm{Doll}}}=\sum_{\ell=0}^{n}\binom{n}{\ell}(k-1)^{n-\ell}\abs{\cC(\ell)}.
\end{equation}
\end{theorem}

We now specify the codes $\cC(\ell)$ used in our algorithms. For $\ell\ge3$, let $H_{\ell}$ be the matrix whose columns are all nonzero binary vectors of length $\ceilenv{\log_2(\ell+1)}$. Arrange its columns such that the last $\ceilenv{\log_2(\ell+1)}$ columns form the identity matrix $I_{\ceilenv{\log_2(\ell+1)}}$. Clearly, $H_{\ell}$ has $2^{\ceilenv{\log_2(\ell+1)}}-1$ columns. Let $\hat{H_{\ell}}$ be the matrix obtained from $H_{\ell}$ by deleting the first $2^{\ceilenv{\log_2(\ell+1)}}-1-\ell$ columns. Since $\ell\ge\ceilenv{\log_2(\ell+1)}$, we have $\hat{H_{\ell}}=\parenv{A_{\ell},I_{\ceilenv{\log_2(\ell+1)}}}$, where $A_{\ell}$ is a $\ceilenv{\log_2(\ell+1)}\times \parenv{\ell-\ceilenv{\log_2(\ell+1)}}$ matrix. In this subsection, for $\ell\ge3$, we take $\cC(\ell)$ to be the binary linear code with parity-check matrix $\hat{H}_{\ell}$. Then $\cC(\ell)$ is a single-substitution correcting code with redundancy $\ceilenv{\log_2(\ell+1)}$. Therefore, we have $\abs{\cC(\ell)}=2^{\ell-\ceilenv{\log_2(\ell+1)}}$ for every $\ell\ge0$. For $\ell\ge3$, $G_\ell:=(I_{\ell-\ceilenv{\log_2(\ell+1)}},A_\ell^\T)$ is a systematic generator matrix of $\cC(\ell)$.

Since $|\cC(\ell)|=2^{\ell-\lceil\log_2(\ell+1)\rceil}
\ge2^{\ell-1}/(\ell+1)$, including $\ell=0,1,2$, it follows that
\begin{equation}\label{eq_lowerboundDoll}
 \begin{aligned}
 \abs{\cC_{\mathrm{Doll}}}\ge\frac12\sum_{\ell=0}^{n}
       \binom n\ell\frac{2^\ell(k-1)^{n-\ell}}{\ell+1}=\frac{(k+1)^{n+1}-(k-1)^{n+1}}{4(n+1)}:=M_{n,k}.
 \end{aligned}
\end{equation}
The equality follows from $\binom n\ell/(\ell+1)=\binom{n+1}{\ell+1}/(n+1)$ and the binomial theorem.
Let $m=\floorenv{\log_{k+1} M_{n,k}}$. Then \eqref{eq_lowerboundDoll} implies that the code $\cC_{\mathrm{Doll}}$ can encode messages in $\Sigma_{k+1}^m$. We give an injective encoder from $\Sigma_{k+1}^m$ to $\cC_{\mathrm{Doll}}$ and a decoder that recovers the message from the corrupted codeword. Assume $n\ge3$ below.

The encoder first converts a message into an integer and then uses this integer to select a codeword. We begin by introducing the notation and results used in the algorithm. For two distinct sequences of the same length, co-lexicographic order compares their symbols at the largest position where they differ: the sequence with
the smaller symbol at that position comes first.

For integers $n\ge1$ and $0\le w\le n$, let $\Sigma_2^n(w)$ denote the set of binary sequences of length $n$ and Hamming weight $w$. For $\bma\in\Sigma_2^n(w)$, define $I_{\bma}=\parenv{i_1,\ldots,i_w}$, where $i_j$ is the position of the $j$-th $1$ in $\bma$. Define the map
\begin{equation*}
\begin{array}{rl}
        f_{n,w}: &\Sigma_2^n(w)\rightarrow \sparenv{\binom{n}{w}}\\
        &\quad\quad\bma\mapsto 1+\sum_{j=1}^{w}\binom{i_j-1}{j}
\end{array}.
\end{equation*}
We have the following lemma.
\begin{lemma}(\cite{Schalkwijk1972IT,kabal2018})
    The value of $f_{n,w}(\bma)$ gives the index of $\bma$ in $\Sigma_2^n(w)$ under the co-lexicographic order.
\end{lemma}

We present some examples to illustrate the above lemma.
\begin{example}
\begin{itemize}
    \item In co-lexicographic order, the three sequences in $\Sigma_2^3(1)$ are: $100$, $010$, and $001$. By definition, $f_{3,1}(100)=1+\binom{1-1}{1}=1$, $f_{3,1}(010)=1+\binom{2-1}{1}=2$, and $f_{3,1}(001)=1+\binom{3-1}{1}=3$.
    \item In co-lexicographic order, the three sequences in $\Sigma_2^3(2)$ are: $110$, $101$, and $011$. By definition, $f_{3,2}(110)=1+\binom{1-1}{1}+\binom{2-1}{2}=1$, $f_{3,2}(101)=1+\binom{1-1}{1}+\binom{3-1}{2}=2$, and $f_{3,2}(011)=1+\binom{2-1}{1}+\binom{3-1}{2}=3$.
\end{itemize}
\end{example}

For any $i\in\sparenv{\binom{n}{w}}$, there is a unique $\bma\in\Sigma_2^n(w)$ such that $f_{n,w}(\bma)=i$. To recover the nonzero positions of $\bma$, notice that
\begin{equation*}
    \binom{i_j-1}{j}
    \le
    \sum_{s=1}^{j}\binom{i_s-1}{s}
    <
    \binom{i_j}{j}.
\end{equation*}
for $1\le j\le w$.
Thus, starting with $\temp=i-1$ and $j=w$,
the integer $i_j-1$ is the largest $u$ satisfying
$\binom{u}{j}\le\temp$.
Subtracting $\binom{i_j-1}{j}$ from $\temp$ leaves
$\sum_{s=1}^{j-1}\binom{i_s-1}{s}$.
Repeating this step recovers all nonzero positions.
This is the procedure in \Cref{alg_inversef}.
When $w=0$, the only sequence is $0^n$,
its index is $1$, and the loop is skipped.
\begin{lemma}(\cite[Section 2.3]{kabal2018})
    For $\bma\in\Sigma_2^n(w)$ and $i\in\sparenv{\binom{n}{w}}$, let $g_{n,w}(i)$ be the output of \Cref{alg_inversef}. Then $f_{n,w}\parenv{g_{n,w}(i)}=i$ and $g_{n,w}(f_{n,w}(\bma))=\bma$.
\end{lemma}
\begin{algorithm*}[!t]
\small{
\DontPrintSemicolon
\SetAlgoLined
\KwIn {$i\in\sparenv{\binom{n}{w}}$}
\KwOut {$\bma\in\Sigma_2^n(w)$}
\textbf{Initialization:}\;
$\temp\gets i-1$\;
$n_0\gets n$\;
$w_0\gets w$\;
$\bma\gets0^n$\;
\While{$w_0\ge1$}
{
\If{$\temp\ge\binom{n_0-1}{w_0}$}
{
$\bma[n_0]\gets1$\;
$\temp\gets\temp-\binom{n_0-1}{w_0}$\;
$w_0\gets w_0-1$
}
$n_0\gets n_0-1$
}
\Return{$\bma$}
\caption{Function $g_{n,w}$}
\label{alg_inversef}}
\end{algorithm*}
\begin{algorithm*}[!t]
\small{
\DontPrintSemicolon
\SetAlgoLined
\KwIn {$\bmx\in\Sigma_{k+1}^{m}$, where $m=\floorenv{\log_{k+1}\parenv{M_{n,k}}}$}
\KwOut {$\bmc\in\cC_{\mathrm{Doll}}$}
\textbf{Initialization:}\;
$N_0\gets\sum_{i=1}^{m}\bmx[i](k+1)^{i-1}+1$\tcp*{Step 1}
$S\gets (k-1)^n$\tcp*{Initialize $S=\binom{n}{0}(k-1)^n\abs{\cC(0)}$}
$\ell\gets0$\;
\While{$S<N_0$ \tcp*{Step 2}}
{
$\ell\gets \ell+1$\;
$S\gets S+\binom{n}{\ell}(k-1)^{n-\ell}2^{\ell-\ceilenv{\log_2(\ell+1)}}$
}
$S\gets S-\binom{n}{\ell}(k-1)^{n-\ell}2^{\ell-\ceilenv{\log_2(\ell+1)}}$\tcp*{compute $\sum_{i=0}^{\ell-1}\binom{n}{i}(k-1)^{n-i}\abs{\cC(i)}$}
$N_1\gets N_0-S$\;
$\temp_1\gets\ceilenv{\frac{N_1}{\binom{n}{\ell}(k-1)^{n-\ell}}}$\;
\eIf{$\ell\in\mathset{0,1,2}$\tcp*{Step 3}}
{$\bms\gets \text{the unique codeword in }\cC(\ell)$}
{
write $\temp_1-1=\sum_{i=0}^{\ell-\ceilenv{\log_2(\ell+1)}-1}y_i2^{i}$, where $y_i\in\{0,1\}$\;
$\bms\gets\parenv{y_0,\ldots,y_{\ell-\ceilenv{\log_2(\ell+1)}-1}}\cdot G_{\ell}$;
}
$\bms^\prime\gets\bms+k-1$\tcp*{Step 4}
$N_2\gets N_1-(\temp_1-1)\binom{n}{\ell}(k-1)^{n-\ell}$\tcp*{Step 5}
$\temp_2\gets\ceilenv{\frac{N_2}{(k-1)^{n-\ell}}}$\;
$\bmc\gets g_{n,\ell}(\temp_2)$\tcp*{See \Cref{alg_inversef}}
$I\gets\mathset{i\in[n]:\bmc[i]=1}$\;
$\bmc\mid_{I}\gets\bms^\prime$\;
$N_3\gets N_2-(\temp_2-1)(k-1)^{n-\ell}$\tcp*{Step 6}
\eIf{$k=2$ or $\ell=n$}{
    $\bmc\mid_{[n]\setminus I}\gets0^{n-\ell}$\;
}{
    Write $N_3-1=\sum_{i=0}^{n-\ell-1}z_i(k-1)^i$,
    where $z_i\in\{0,\ldots,k-2\}$\;
    $\bmc\mid_{[n]\setminus I}
    \gets(z_0,\ldots,z_{n-\ell-1})$\;
}
\Return{$\bmc$}
\caption{$\Enc_{\mathrm{Doll}}$: Encoder for $\cC_{\mathrm{Doll}}$}
\label{alg_encoder10CECC}}
\end{algorithm*}

Now we are ready to present the encoding algorithm, which is given in \Cref{alg_encoder10CECC}.
\begin{theorem}
    \Cref{alg_encoder10CECC} defines an injective map $\Enc_{\mathrm{Doll}}:\Sigma_{k+1}^m\to\cC_{\rm Doll}$.
\end{theorem}
\begin{IEEEproof}
For each $0\le\ell\le n$, a codeword having exactly $\ell$ symbols in $\{k-1,k\}$ is uniquely determined by three choices: a codeword $\bms\in\cC(\ell)$, the $\ell$ positions occupied by $\bms+(k-1)$, and the symbols in the remaining positions. The number of such codewords is
$\binom{n}{\ell}(k-1)^{n-\ell}|\cC(\ell)|$. The encoder makes these choices through the following six steps.
\begin{enumerate}[\textbf{Step} 1]
    \item Convert $\bmx$ into the integer $N_0=\sum_{i=1}^{m}\bmx[i](k+1)^{i-1}+1$. By \eqref{eq_lowerboundDoll}, we have $1\le N_0\le(k+1)^m\le\abs{\cC_{\mathrm{Doll}}}$.
    \item (determining $\ell$) Find the smallest $0\le \ell\le n$, such that $\sum_{i=0}^{\ell}\binom{n}{i}(k-1)^{n-i}\abs{\cC(i)}\ge N_0$. Such an $\ell$ does exist by \eqref{eq_sizeDoll}.
    \item (\textbf{identifying a codeword in} $\cC(\ell)$) Let $N_1=N_0-\sum_{i=0}^{\ell-1}\binom{n}{i}(k-1)^{n-i}\abs{\cC(i)}$. Clearly, $N_1\ge1$. Let $\temp_1=\ceilenv{\frac{N_1}{\binom{n}{\ell}(k-1)^{n-\ell}}}$. If $\ell\in\{0,1,2\}$, we have $\temp_1=1$. In this case, let $\bms$ be the unique codeword in $\cC(\ell)$. If $\ell\ge3$, then $\temp_1$ is the smallest integer in $\sparenv{1,\abs{\cC(\ell)}}$ such that $\temp_1\cdot\binom{n}{\ell}(k-1)^{n-\ell}\ge N_1$. Write $\temp_1-1=\sum_{i=0}^{\ell-\ceilenv{\log_2(\ell+1)}-1}y_i2^{i}$, where $y_i\in\{0,1\}$. This is feasible since $0\le\temp_1-1<\abs{\cC(\ell)}=2^{\ell-\ceilenv{\log_2(\ell+1)}}$. Let $\bms=\parenv{y_0,\ldots,y_{\ell-\ceilenv{\log_2(\ell+1)}-1}}\cdot G_{\ell}$. Then $\bms$ is a codeword in $\cC(\ell)$. Notice that distinct values of $\temp_1$ give distinct $\bms$.
    \item Let $\bms^\prime=\bms+k-1$ (i.e., $\bms^\prime[i]=\bms[i]+k-1$ for all $i$). When $\ell=0$, both $\bms$ and $\bms^\prime$ are empty.
    \item (\textbf{determining locations to place} $\bms^\prime$) Let $N_2=N_1-(\temp_1-1)\binom{n}{\ell}(k-1)^{n-\ell}$. Clearly, $N_2\ge1$. Let $\temp_2=\ceilenv{\frac{N_2}{(k-1)^{n-\ell}}}$. Then $\temp_2$ is the smallest integer in $\sparenv{1,\binom{n}{\ell}}$ such that $\temp_2(k-1)^{n-\ell}\ge N_2$. Let $\bmc=g_{n,\ell}\parenv{\temp_2}$. In other words, $\bmc$ is the sequence in $\Sigma_2^n(\ell)$ with index $\temp_2$ under the co-lexicographic order. For $1\le i\le \ell$, replace the $i$-th $1$ of $\bmc$ with the $i$-th symbol of $\bms^\prime$.
    \item (\textbf{determining symbols in the remaining} $n-\ell$ \textbf{positions of} $\bmc$) Let $N_3=N_2-(\temp_2-1)(k-1)^{n-\ell}$. Then $1\le N_3\le (k-1)^{n-\ell}$. If $k\ge3$ and $\ell<n$, write
    $N_3-1=\sum_{i=0}^{n-\ell-1}z_i(k-1)^i$,
    with $z_i\in\{0,\ldots,k-2\}$, and place these symbols in the remaining positions in increasing order. If $k=2$, then $N_3=1$,
    and all remaining symbols are zero.
    If $\ell=n$, then $N_3=1$, and there are no remaining positions.
\end{enumerate}

These steps produce a codeword
$\bmc\in\cC_{\mathrm{Doll}}$. Conversely, each codeword uniquely determines the tuple $(\ell,\temp_1,\temp_2,N_3)$. Define the index of the codeword in $\cC_{\mathrm{Doll}}$ with tuple $\parenv{\ell,\temp_1,\temp_2,N_3}$ as
\begin{equation}\label{eq_Dollindex}
    \sum_{i=0}^{\ell-1}\binom{n}{i}(k-1)^{n-i}\abs{\cC(i)}+(\temp_1-1)\binom{n}{\ell}(k-1)^{n-\ell}+(\temp_2-1)(k-1)^{n-\ell}+N_3.
\end{equation}
For the output of \Cref{alg_encoder10CECC},
this index equals $N_0$. Since distinct messages give distinct values of $N_0$, the encoder is injective. This completes the proof.
\end{IEEEproof}
\begin{algorithm*}[t]
\small
\DontPrintSemicolon
\SetAlgoLined
\KwIn{
A matrix $\bmc^\prime\in\Sigma_2^{k\times n}$
obtained from $\iota_{2,k}^{-1}\parenv{\Enc_{\mathrm{Doll}}(\bmx)}$ by at most one substitution in the first row
}
\KwOut{$\bmx\in\Sigma_{k+1}^{m}$}
\textbf{Initialization:}\;
$I\gets\mathset{i\in[n]:\bmc_2^\prime[i]=1}$\;
$\ell\gets|I|$\;
\eIf{there is an invalid column in $\bmc^\prime$}{
    Let $j$ be the position of this column\;
    $\bmc_1^\prime[j]\gets0$\;
    $\bmc\gets\iota_{2,k}(\bmc^\prime)$\;
    $\bms\gets\cF(\bmc)$\;
}{
    \eIf{$\ell\in\{0,1,2\}$}{
        $\bms\gets
        \text{the unique codeword in }\cC(\ell)$\;
    }{
    $\bmd\gets\bmc_1^\prime\mid_I$\;
        $\bms\gets
        \text{the output of the decoder of }\cC(\ell)
        \text{ on input }\bmd$\;
    }
    $\bmc\gets\iota_{2,k}(\bmc^\prime)$\;
    $\bmc\mid_I\gets\bms+(k-1)$\;
}
$(z_0,\ldots,z_{n-\ell-1})
\gets\bmc\mid_{[n]\setminus I}$\;
$N_3\gets1+\sum_{j=0}^{n-\ell-1}z_j(k-1)^j$\;
Write $I=\{i_1,\ldots,i_\ell\}$
with $i_1<\cdots<i_\ell$\;
$\temp_2\gets
1+\sum_{j=1}^{\ell}\binom{i_j-1}{j}$\;
$N_2\gets N_3+(\temp_2-1)(k-1)^{n-\ell}$\;
\eIf{$\ell\in\{0,1,2\}$}{
    $\temp_1\gets1$\;
}{
    $(y_0,\ldots,y_{\ell-\ceilenv{\log_2(\ell+1)}-1})
    \gets\text{the first }\ell-\ceilenv{\log_2(\ell+1)}\text{ bits of }\bms$\;
    $\temp_1\gets1+\sum_{i=0}^{\ell-\ceilenv{\log_2(\ell+1)}-1}y_i2^i$\;
}
$N_1\gets
N_2+(\temp_1-1)\binom{n}{\ell}(k-1)^{n-\ell}$\;
$N_0\gets
N_1+\sum_{i=0}^{\ell-1}\binom{n}{i}
(k-1)^{n-i}|\cC(i)|$\;
Write $N_0-1=\sum_{i=1}^{m}\bmx[i](k+1)^{i-1}$\;
$\bmx\gets\bmx[1]\cdots\bmx[m]$\;
\Return{$\bmx$}
\caption{$\Dec_{\mathrm{Doll}}$: Decoder corresponding to $\Enc_{\mathrm{Doll}}$}
\label{alg_decoder10CECC}
\end{algorithm*}

Let $\iota_{2,k}:\Phi_{2,k}\to\Sigma_{k+1}$ be the bijection defined in \Cref{def_bijection}. \Cref{alg_decoder10CECC} presents the decoder corresponding to $\Enc_{\mathrm{Doll}}$.
\begin{theorem}
Let $k\ge2$, $n\ge3$, and
$m=\floorenv{\log_{k+1}M_{n,k}}$.
For any $\bmx\in\Sigma_{k+1}^m$, let
$\bmc=\Enc_{\mathrm{Doll}}(\bmx)$.
If $\bmc^\prime\in\Sigma_2^{k\times n}$ is obtained
from $\iota_{2,k}^{-1}(\bmc)$ by at most one
substitution in the first row, then \Cref{alg_decoder10CECC} recovers $\bmx$ from
$\bmc^\prime$; that is,
$\Dec_{\mathrm{Doll}}(\bmc^\prime)=\bmx$.
\end{theorem}
\begin{IEEEproof}
We first show that the algorithm recovers the
codeword $\bmc$, following the argument in the proof of \cite[Theorem 9]{BesartDollma202509}.
Since rows $2,\ldots,k$ are unchanged and a valid binary column represents $k-1$ or $k$ if and only if its second entry is $1$, we have
\begin{equation*}
    I:=
    \mathset{i\in[n]:\bmc_2^\prime[i]=1}
    =
    \mathset{i\in[n]:\bmc[i]\in\{k-1,k\}}.
\end{equation*}
Thus, the decoder knows the positions of these symbols
and their number $\ell=|I|$.
Moreover, every received column at a position in $I$
is valid, because all its entries in rows $2,\ldots,k$
are $1$.

Suppose that the received column at position $j$ is invalid. The entries in rows $2,\ldots,k$ are unchanged and remain nondecreasing, so the violation must occur between the first two entries.
Hence $\bmc_1^\prime[j]=1$ and $\bmc_2^\prime[j]=0$.
The transmitted column has second entry $0$,
and therefore its first entry must also be $0$.
Replacing $\bmc_1^\prime[j]$ with $0$ recovers this column. Since at most one substitution occurred,
all other entries are correct. The corrected matrix is therefore $\iota_{2,k}^{-1}(\bmc)$, and applying $\iota_{2,k}$ recovers $\bmc$.

Now suppose that all received columns are valid. Any error must exchange $k-1$ and $k$ at a position in $I$, while the columns outside $I$ are unchanged. Let $\bms=\cF(\bmc)\in\cC(\ell)$.
The binary sequence $\bmd:=\bmc_1^\prime\mid_I$
then satisfies $\dth{\bmd,\bms}\le1$. For $\ell\in\{0,1,2\}$, the decoder returns the unique codeword in $\cC(\ell)$, which is $\bms$. For $\ell\ge3$, we have $\hat H_\ell\bms^\T=0$. If $\bmd=\bms$, then $\hat H_\ell\bmd^\T=0$. Otherwise, if $\bmd$ and $\bms$ differ at position $h$,
then $\hat H_\ell\bmd^\T$ equals the $h$-th column
of $\hat H_\ell$.
Since the columns of $\hat H_\ell$ are nonzero and
pairwise distinct, the syndrome determines whether
an error occurred and, if so, its position.
Thus, the decoder recovers $\bms$.
Replacing the symbols at positions in $I$ with
$\bms+(k-1)$, while retaining the unchanged symbols
outside $I$, recovers $\bmc$.

It remains to show that the algorithm recovers the original message from $\bmc$. Write $I=\{i_1,\ldots,i_\ell\}$ with $i_1<\cdots<i_\ell$, and write $\bmc\mid_{[n]\setminus I}=(z_0,\ldots,z_{n-\ell-1})$.
By the definition of $f_{n,\ell}$ and the last step
of the encoder, the original values of $\temp_2$
and $N_3$ are
\begin{equation*}
    \temp_2
    =1+\sum_{j=1}^{\ell}\binom{i_j-1}{j},
    \;
    N_3
    =1+\sum_{j=0}^{n-\ell-1}z_j(k-1)^j.
\end{equation*}
These formulas also apply when $\ell=0$ or $\ell=n$,
with empty sums interpreted as zero.
When $k=2$, all symbols outside $I$ are zero,
so the second formula gives $N_3=1$, as required.

If $\ell\in\{0,1,2\}$, then $|\cC(\ell)|=1$,
and the encoder used $\temp_1=1$. Otherwise, since $G_\ell$ is systematic, the first $\ell-\ceilenv{\log_2(\ell+1)}$ bits of $\bms=\cF(\bmc)$
are exactly the information bits
$(y_0,\ldots,y_{\ell-\ceilenv{\log_2(\ell+1)}-1})$ used by the encoder. Consequently, the original value of $\temp_1$ is
$\temp_1=1+\sum_{i=0}^{\ell-\ceilenv{\log_2(\ell+1)}-1}y_i2^i$.
Thus, the algorithm recovers the same tuple
$(\ell,\temp_1,\temp_2,N_3)$ as in the encoder.
Its successive computations of $N_2$, $N_1$, and $N_0$
therefore recover the codeword index in
\eqref{eq_Dollindex}, namely
$N_0=1+\sum_{i=1}^{m}\bmx[i](k+1)^{i-1}$.
Since $0\le N_0-1<(k+1)^m$,
the unique base-$(k+1)$ expansion of $N_0-1$
using exactly $m$ digits recovers $\bmx$.
Hence $\Dec_{\mathrm{Doll}}(\bmc^\prime)=\bmx$.
\end{IEEEproof}

\subsection{\texorpdfstring{$1$}{1}-CECCs}\label{sec_1CECC}
%%%%%%%%%%%%%%%%%%%%%%%%%%%%%%%%%%%%%%%%%%%%%%%%%%%%
We now turn to the $k$-resolution $1$-composite-error model, starting with the case $q=2$. In this case, we identify $\Phi_{2,k}$ with $\Sigma_{k+1}$ by the bijection $\iota_{2,k}$ in \Cref{def_bijection}.

A limited-magnitude-error of magnitude $1$ ($1$-LME) without wrap-around is a substitution error where a symbol $\sigma\in\Sigma_Q$ may be altered to $\sigma+ 1$ or $\sigma-1$, with the constraint that symbol $0$ can only be changed to $1$, and symbol $Q-1$ can only be changed to $Q-2$. The proof of \cite[Theorem 11]{BesartDollma202509} implies the following lemma.
\begin{lemma}(\cite[Theorem 11]{BesartDollma202509})\label{lem_reduce1}
    If a code $\cC\subseteq\Sigma_{k+1}^n$ can correct a $1$-LME without wrap-around, it is a binary $k$-resolution $1$-CECC.
\end{lemma}

For any integer $Q\ge 2$ and $\bmx\in\Sigma_Q^{n}$, define $\vt\parenv{\bmx}\triangleq\sum_{i=1}^{n}i\bmx[i]$. Next, we present a $1$-LME correcting code.
\begin{theorem}\label{thm_leecoder1}
  Let $Q\ge 3$ and $n\ge2$ be integers. For any $0\le a\le 2n$, the code
  \begin{equation}\label{eq_binary1CECC}
      \cC(n;Q,a):=\mathset{\bmc\in\Sigma_Q^n:\vt(\bmc)\equiv a\pmod{2n+1}}
  \end{equation}
  can correct a $1$-LME without wrap-around. There exists some $a_0\in[0,2n]$ such that the redundancy of $\cC(n;Q,a_0)$ is at most $\log_Q(2n+1)$. For every $a\in[0,2n]$, \Cref{alg_leecoder1} systematically encodes messages into $\cC(n;Q,a)$ with redundancy $\ceilenv{\log_Q(n)}+1$.
\end{theorem}
\begin{IEEEproof}
    Let $\bmy$ be obtained from $\bmc\in\cC(n;Q,a)$ by at most one $1$-LME without wrap-around. Suppose that $\bmc[i]$ is altered to $\bmy[i]=\bmc[i]+\delta$, where $\delta\in\mathset{-1,0,1}$. To decode $\bmc$ from $\bmy$, it suffices to determine the values of $i$ and $\delta$. Let $\Delta=\parenv{\vt(\bmy)-a}\pmod{2n+1}$. It is easy to verify that
    \begin{equation*}
        \Delta=
        \begin{cases}
            0,\mbox{ if }\delta=0,\\
            i,\mbox{ if }\delta=1,\\
            2n+1-i,\mbox{ if }\delta=-1.
        \end{cases}
    \end{equation*}
    Therefore, we can conclude that $\delta=0$ when $\Delta=0$, and
    \begin{equation*}
        (i,\delta)=
        \begin{cases}
            \parenv{\Delta,1},\mbox{ if }\Delta\in[n],\\
            \parenv{2n+1-\Delta,-1},\mbox{ if }\Delta\in[n+1,2n].
        \end{cases}
    \end{equation*}
    This proves the first statement. The family $\mathset{\cC(n;Q,a):0\le a\le 2n}$ forms a partition of $\Sigma_Q^n$. By the pigeonhole principle, there is some $a_0$ such that $\abs{\cC(n;Q,a_0)}\ge\frac{Q^n}{2n+1}$. Therefore, its redundancy is $\log_Q\parenv{Q^n/\abs{\cC(n;Q,a_0)}}\le\log_Q(2n+1)$, proving the second statement. The redundancy of the systematic encoder is clear.

    Next, we prove the last statement. Since the length of the output codeword is $n$ and the length of the input message is $n-m-1$, the redundancy is $m+1=\ceilenv{\log_Q(n)}+1$. It is easy to verify that $Q^{m-1}<n$. Therefore, $n\notin S$ and $T$ is a subset of $[n]$ of size $m+1$. In Lines 4--18, the encoder first assigns symbols in $\bmx$, in their original order, to coordinates of $\bmc$ indexed by $[n]\setminus T$, and assigns symbols $0$ to coordinates indexed by $T$. Then it continues to calculate $d=\parenv{a-\vt(\bmc)}\pmod{2n+1}$. If $d=0$, then $\bmc\in\cC(n;Q,a)$ and the encoder outputs $\bmc$. Otherwise, the encoder updates $\bmc$ according to the value of $d$.

    If $1\le d<n$, the encoder expands $d$ into a $Q$-ary vector $\parenv{d_0,\ldots,d_{m-1}}$ (this is possible since $1\le d<n\le Q^m$), and then updates $\bmc\mid_S$ with $\parenv{d_0,\ldots,d_{m-1}}$.
    If $n\le d< 2n$, the encoder expands $d-n$ into a $Q$-ary vector $\parenv{d_0^\prime,\ldots,d_{m-1}^\prime}$, and then updates $\bmc\mid_T$ with $\parenv{d_0^\prime,\ldots,d_{m-1}^\prime,1}$.
    If $d=2n$, the encoder updates $\bmc[n]$ with $2$.

    By the definition of $d$, it is easy to verify that the updated $\bmc$ satisfies $\vt(\bmc)\equiv a\pmod{2n+1}$, and hence $\bmc\in\cC(n;Q,a)$.
\end{IEEEproof}
\begin{algorithm*}[!t]
\small{
\DontPrintSemicolon
\SetAlgoLined
\KwIn {$\bmx\in\Sigma_Q^{n-m-1}$, where $m=\ceilenv{\log_{Q}(n)}$}
\KwOut {$\bmc\in\cC(n;Q,a)$}
\textbf{Initialization:}\;
$S\gets\mathset{Q^j:0\le j\le m-1}$\;
$T\gets S\cup\{n\}$\;
$\bmc\mid_{[n]\setminus T}\gets\bmx$\; $\bmc\mid_T\gets0^{m+1}$\;
$d\gets\parenv{a-\vt\parenv{\bmc}}\pmod{2n+1}$\;
\If{$d\ne0$}
{
\eIf{$1\le d< n$}
{
write $d$ as $d=\sum_{j=0}^{m-1}d_jQ^{j}$, where each $0\le d_j<Q$\;
$\bmc\mid_S\gets\parenv{d_0,\ldots,d_{m-1}}$
}
{
\eIf{$n\le d<2n$}
{
$d^\prime\gets d-n$\;
write $d^\prime$ as $d^\prime=\sum_{j=0}^{m-1}d_j^\prime Q^{j}$, where each $0\le d_j^\prime<Q$\;
$\bmc\mid_S\gets\parenv{d_0^\prime,\ldots,d_{m-1}^\prime}$\;
$\bmc[n]\gets 1$
}
{
$\bmc[n]\gets 2$
}
}
}
}
\Return{$\bmc$}
\caption{Encoder for $\cC(n;Q,a)$}
\label{alg_leecoder1}
\end{algorithm*}

The following corollary follows immediately from \Cref{lem_reduce1} and \Cref{thm_leecoder1} by setting $Q=k+1$.
\begin{corollary}
    For all $k,n\ge2$, there exists a binary $k$-resolution $1$-CECC with redundancy at most $\log_{k+1}(2n+1)$, or equivalently, with size at least $\frac{(k+1)^n}{2n+1}$. Moreover, \Cref{alg_leecoder1} constructs a systematic $1$-CECC of length $n$ with redundancy $\ceilenv{\log_{k+1}(n)}+1$.
\end{corollary}

\begin{remark}
\begin{itemize}
    \item Setting $e=1$ in the proof of \cite[Corollary 1]{BesartDollma202509} yields an explicit construction of binary length-$n$ $k$-resolution $1$-CECCs with redundancy $\ceilenv{\log_{k+1}(n+1)}+1$ when $k+1$ is a prime power. In comparison, our explicit construction (\Cref{alg_leecoder1}) holds for all $k$ and has redundancy $\ceilenv{\log_{k+1}(n)}+1$.
    \item In \cite[Corollary 5]{BesartDollma202509}, Dollma \emph{et al} constructed a binary length-$n$ $k$-resolution $1$-CECC with redundancy $\ceilenv{\log_{k+1}(2n+1)}$. This construction holds only when $k$ is even. The above corollary ensures that there is a code with redundancy at most $\log_{k+1}(2n+1)$ for \emph{all} $k$. \Cref{alg_leecoder1} provides an explicit construction with redundancy $\ceilenv{\log_{k+1}(n)}+1$.
\end{itemize}
\end{remark}

Next, we proceed to the case $q>2$. Our construction relies on the following lemma.
\begin{lemma}\label{lem_q1CECC}
    For integers $q>2,k\ge2$ and $n\ge q$, let $p_1\ge n$ and $p_2\ge q$ be two primes. Let $\bmc\in\Phi_{q,k}^n$ and $\bmc^\prime$ be obtained from $\bmc$ by a $1$-composite-error. One can recover $\bmc$ from $\bmc^\prime$ if given the following three values:
    \begin{itemize}
        \item $a_1:=\sum_{i=1}^{k}\sum_{j=1}^{n}\bmc_i[j]\pmod {2q-1}$;
        \item $a_2:=\sum_{i=1}^{k}\vt\parenv{\bmc_i}\pmod{p_1}$;
        \item $a_3:=\sum_{i=1}^n\sum_{j=0}^{q-1}j^2\cdot \Wth{j}{\bmc[i]}\pmod{p_2}$.
    \end{itemize}
\end{lemma}
\begin{IEEEproof}
    Suppose that $\bmc^\prime$ is obtained from $\bmc$ by substituting $\bmc_i[j]$ with $\bmc_i[j]+\delta$, where $-(q-1)\le \delta\le q-1$. Let $\Delta_1= (\sum_{r=1}^{k}\sum_{s=1}^{n}\bmc_r^\prime[s]-a_1)\pmod{2q-1}$. Then $\Delta_1=\delta \pmod{2q-1}$, from which the value of $\delta$ can be recovered as
    \begin{equation*}
        \delta=
        \begin{cases}
            \Delta_1,&\mbox{ if }0\le\Delta_1<q,\\
            \Delta_1-(2q-1),&\mbox{ if }\Delta_1\ge q.
        \end{cases}
    \end{equation*}
    If $\delta=0$, no error occurred. Next, assume $\delta\ne0$.

    If $\bmc^\prime[j]$ is an invalid symbol, index $j$ can be directly identified. Moreover, there must exist an index $r\in[k-1]$ such that $\bmc_r^\prime[j]>\bmc_{r+1}^\prime[j]$, implying that the substitution either increased $\bmc_r[j]$ or decreased $\bmc_{r+1}[j]$. Since $\delta$ is known, the error can be corrected immediately. Now suppose that $\bmc^\prime[j]$ is a valid symbol. Let $\alpha=\bmc_i[j]$. To correct the error, it suffices to determine the values of $j$ and $\alpha$, as $\bmc[j]$ is a nondecreasing column vector over $\Sigma_q$ and $\delta$ is known.

    We first determine the value of $j$. Let $\Delta_2=\parenv{\sum_{r=1}^{k}\vt\parenv{\bmc_r^\prime}-a_2}\pmod{p_1}$. Then $\Delta_2=j\delta\pmod{p_1}$.
    Since $0<\abs{\delta}<p_1$ and $p_1$ is a prime, we have $j\equiv\delta^{-1}\Delta_2\pmod{p_1}$. Then since $1\le j\le n\le p_1$, integer $j$ can be recovered from $j\pmod{p_1}$.

    Next, we determine the value of $\alpha$. Let $\Delta_3=\parenv{\sum_{s=1}^n\sum_{\ell=0}^{q-1}\ell^2\cdot \Wth{\ell}{\bmc^\prime[s]}-a_3}\pmod{p_2}$. Then $\Delta_3=\parenv{2\alpha+\delta}\delta\pmod{p_2}$. Since $0<\abs{\delta}<p_2$, $0\le\alpha\le q-1<p_2$ and $p_2$ is an odd prime, we have $\alpha=\alpha\pmod{p_2}= 2^{-1}\parenv{\Delta_3\delta^{-1}-\delta}\pmod{p_2}$. This completes the proof.
\end{IEEEproof}

For integers $Q\ge 2$, $L\ge1$ and $0\le m<Q^L$, we can always express $m$ as $\sum_{i=1}^{L}m_iQ^{i-1}$ where $0\le m_i<Q$. Define $\expan{Q,L}{m}$ to be the $Q$-ary sequence $m_1m_2\cdots m_{L}$. For example, $\expan{2,2}{3}=11$, $\expan{2,3}{3}=110$, $\expan{2,3}{4}=001$ and $\expan{3,1}{2}=2$.
\begin{construction}\label{construction_q1CECC}
    For integers $q>2,k\ge2$ and $m\ge q$, let $p_1$ be the smallest prime no less than $m$, and let $p_2$ be the smallest prime no less than $q$. Let $\Delta=\ceilenv{\log_{Q_{q,k}}\parenv{p_1p_2}}$ and $n=m+2+\Delta$. Recall the bijection $\iota_{q,k}:\Phi_{q,k}\to\Sigma_{Q_{q,k}}$ from \Cref{def_bijection}. Define the map
    $$
    \begin{array}{rl}
         \varphi_1^S:&\Phi_{q,k}^m\rightarrow\Phi_{q,k}^n \\
         &\quad\bmx\mapsto\bmc
    \end{array}
    $$ as follows:
    \begin{itemize}
        \item Set $\bmc\mid_{[m]}=\bmx$.
        \item Write $\sum_{i=1}^{k}\sum_{j=1}^{m}\bmx_i[j]\pmod {2q-1}=aq+b$, where $a\in\{0,1\}$ and $0\le b\le q-1$. Set $\bmc[m+1]=\sparenv{a,\ldots,a}^\T$ and $\bmc[m+2]=\sparenv{b,\ldots,b}^\T$.
        \item Let $u(\bmx)=\sum_{i=1}^{k}\vt(\bmx_i)\pmod{p_1}$ and $v(\bmx)=\sum_{i=1}^{m}\sum_{j=0}^{q-1}j^2\Wth{j}{\bmx[i]}\pmod{p_2}$. Then we have $u(\bmx)+p_1v(\bmx)\in\sparenv{0,p_1p_2-1}$. Set $\bmc\mid_{[m+3,m+2+\Delta]}=\iota_{q,k}^{-1}\parenv{\expan{Q_{q,k},\Delta}{u(\bmx)+p_1v(\bmx)}}$.
    \end{itemize}
    Define $\cC_1^S=\mathset{\varphi_1^S(\bmx):\bmx\in\Phi_{q,k}^m}$.
\end{construction}

\begin{example}
In \Cref{construction_q1CECC}, take $q=3$, $k=2$, $m=3$, $p_1=p_2=3$. Then $Q_{3,2}=6$, $\Delta=2$, and $n=7$. Let
$$
    \bmx=
    \begin{pmatrix}
    0&1&2\\
    2&1&2
    \end{pmatrix}.
$$
Then $\sum_{i=1}^{k}\sum_{j=1}^{m}\bmx_i[j]\pmod {2q-1}=8\pmod{5}=1\cdot3+0$. Hence, $a=1$ and $b=0$. Straightforward computation gives $u(\bmx)=(8+10)\pmod{3}=0$ and $v(\bmx)=(0^2+2^2+1^2\times 2+2^2\times 2)\pmod3=2$.
Thus,
$\expan{Q_{q,k},\Delta}{u(\bmx)+p_1v(\bmx)}=\expan{6,2}{6}=01$.
Applying the inverse bijection $\iota_{3,2}^{-1}$ component-wise gives
\begin{equation*}
 \iota_{3,2}^{-1}\parenv{\expan{6,2}{6}}
 =\begin{pmatrix}0&0\\0&1\end{pmatrix}.
\end{equation*}
Therefore,
$$
    \varphi_1^S(\bmx)=
    \parenv{\begin{array}{ccc|cc|cc}
        0&1&2&1&0&0&0\\
        2&1&2&1&0&0&1
    \end{array}}.
$$
\end{example}

\begin{theorem}%\label{theo:1-CECC}
   The code $\cC_1^S$ is a $q$-ary $k$-resolution $1$-CECC with redundancy $2+\ceilenv{\log_{Q_{q,k}}\parenv{p_1p_2}}=\log_{Q_{q,k}}(n)+O(1)$.
\end{theorem}
\begin{IEEEproof}
Let $\bmc=\varphi_1^S(\bmx)$ be the transmitted codeword, and let $\bmy_i\in\Sigma_q^n$ be the received sequence in row $i\in[k]$. The received matrix $\bmy$ differs from $\bmc$ in at most one entry. We show how to recover $\bmx$ from $\bmy$.

By the construction, each of the columns $\bmc[m+1]$ and $\bmc[m+2]$ consists of $k$ identical symbols. Since $k\ge2$, a substitution in either column makes its entries unequal. Thus, if the entries of $\bmy[m+1]$ or $\bmy[m+2]$ are not
all equal, the only substitution occurred outside the first $m$ columns, and $\bmx=\bmy\mid_{[m]}$.

Suppose now that the entries within each of these two columns are all equal. They are therefore error-free, and the integers $a=\bmy_1[m+1]$ and $b=\bmy_1[m+2]$ are known. Let $S=\sum_{i=1}^k\sum_{j=1}^m\bmy_i[j]\pmod{2q-1}$.
If no substitution occurred in the first $m$ columns, then $S=aq+b$. Otherwise, the substitution changes the sum by a nonzero integer $\delta$ with $|\delta|\le q-1$, so $S\not\equiv aq+b\pmod{2q-1}$. Consequently, if $S=aq+b$, we again have $\bmx=\bmy\mid_{[m]}$.

It remains to consider the case $S\ne aq+b$. In this case, the only substitution occurred in the first $m$ columns. Hence, the last $\Delta$ columns are error-free, and the integer
\begin{equation*}
 z=\sum_{h=1}^{\Delta}
       \iota_{q,k}\parenv{\bmy[m+2+h]}Q_{q,k}^{h-1}
   =u(\bmx)+p_1v(\bmx)
\end{equation*}
is known. Since $0\le u(\bmx)<p_1$ and $0\le v(\bmx)<p_2$, we recover $u(\bmx)=z\pmod{p_1}$ and $v(\bmx)=\floorenv{z/p_1}$. The three values $aq+b$, $u(\bmx)$, and $v(\bmx)$ are exactly
the checks required by \Cref{lem_q1CECC}, with the sequence length set to $m$. That lemma therefore recovers $\bmx$ from $\bmy\mid_{[m]}$.

In every case, $\bmx$ is recovered. This proves the required error-correction property.

Finally, the first $m$ columns of $\varphi_1^S(\bmx)$ equal
$\bmx$. Hence, $\varphi_1^S$ is injective and
$|\cC_1^S|=Q_{q,k}^m$. Its redundancy is therefore
\begin{equation*}
 \log_{Q_{q,k}}\frac{Q_{q,k}^n}{|\cC_1^S|}
 =n-m
 =2+\ceilenv{\log_{Q_{q,k}}\parenv{p_1p_2}}.
\end{equation*}
By Bertrand's postulate, we have $p_1=\Theta(m)$ and $p_2=O(1)$, implying $p_1p_2=\Theta(m)$. By definition, we have $n=m+O(\log m)$. Thus,
the redundancy is $\log_{Q_{q,k}}n+O(1)$.
\end{IEEEproof}

\subsection{\texorpdfstring{$t$}{t}-\texorpdfstring{$(1,\ldots,1)$}{(1,\ldots,1)}-CECCs}\label{sec_t1CECC}
%%%%%%%%%%%%%%%%%%%%%%%%%%%%%%%%%%%%%%%%%%%%%%%%%%%%
In this subsection, we construct $k$-resolution $t$-$(1,\ldots,1)$-CECCs. For a sequence $\bmx\in\Sigma_q^n$, write $\Sum{\bmx}=\sum_{i=1}^{n}\bmx[i]$. We first give a lemma for correcting one substitution in $q$-ary sequences.
\begin{lemma}\label{lem_qary1sub}
    Let $q\ge2$ and $n\ge1$ be integers. For any $\bmx\in\Sigma_q^n$, if $\vt\parenv{\bmx}\pmod{2n(q-1)}$ and $\Sum{\bmx}\pmod{q}$ are known, we can decode $\bmx$ from any $\bmx^\prime$ that is obtained from $\bmx$ by at most one substitution.
\end{lemma}
\begin{IEEEproof}
    Suppose that $\bmx^\prime$ is obtained from $\bmx$ by substituting $\bmx[i]$ with $\bmx[i]+\delta$, where $-(q-1)\le\delta\le q-1$. Let $\Delta_1=\parenv{\Sum{\bmx^\prime}-\Sum{\bmx}}\pmod{q}$. Then $\Delta_1=\delta\pmod{q}$. If $\Delta_1=0$, no substitution occurred, and hence $\bmx=\bmx^\prime$.

    Now consider the case when $\Delta_1\ne 0$. Let $\Delta_2=\parenv{\vt(\bmx^\prime)-\vt(\bmx)}\pmod{2n(q-1)}$. Then $\Delta_2=i\delta\pmod{2n(q-1)}$. Since $-n(q-1)\le i\delta\le n(q-1)$ and $\delta\ne 0$, we conclude that $1\le\Delta_2\le 2n(q-1)-1$. To decode $\bmx$ from $\bmx^\prime$, it suffices to determine the values of $i$ and $\delta$. We proceed according to the value of $\Delta_2$.

    If $\Delta_2<n(q-1)$, then $\delta>0$. Therefore, $\delta=\Delta_1$ and $i\delta=\Delta_2$. Then it follows that $i=\Delta_2/\Delta_1$.

    If $\Delta_2>n(q-1)$, then $\delta<0$. Therefore, $\delta=\Delta_1-q$ and $i\delta=\Delta_2-2n(q-1)$. Then it follows that $i=\frac{\Delta_2-2n(q-1)}{\Delta_1-q}$.

    If $\Delta_2=n(q-1)$, it must be that $i=n$ and $\delta\in\mathset{q-1,-(q-1)}$. In this case, we have that $\delta=q-1$ if $\bmx^\prime[n]=q-1$, or $\delta=-(q-1)$ if $\bmx^\prime[n]=0$. Now the proof is completed.
\end{IEEEproof}

\begin{definition}\label{def_cauchymatrix}
  For two integers $k,t$ with $k\ge t\ge1$ and a prime $p>k+t$, define a matrix $\bbM=\parenv{m_{j,i}}_{t\times k}$ over the finite field $\mathbb{F}_p$ as $m_{j,i}=\parenv{j-(t+i)}^{-1}$ for every $j\in[t]$ and $i\in[k]$.
\end{definition}

The matrix $\bbM$ is a Cauchy matrix (\cite[Section~3.6]{Pan2001}). Our construction relies on the following property of $\bbM$.
\begin{lemma}(\cfcite{Pan2001}{Corollary~3.6.5})\label{lem_cauchymatrix}
  For every $1\le s\le t$, every $s\times s$ submatrix of $\bbM$ is invertible over $\mathbb{F}_p$.
\end{lemma}

Given $\bmx\in\Sigma_q^m$, define $\overline{\vt\parenv{\bmx}}\triangleq\vt\parenv{\bmx}\pmod{2m(q-1)}$. Then $0\le\overline{\vt\parenv{\bmx}}<2m(q-1)$. Our construction is given below.
\begin{construction}\label{construction_t1CECC}
  For integers $q\ge2$, $2\le t\le k$ and $m\ge1$, let $p$ be the smallest prime greater than $\max\mathset{2m(q-1),k+t}$. Let $\bbM=(m_{j,i})_{t\times k}$ be the matrix in \Cref{def_cauchymatrix}. Denote $\Delta=\ceilenv{\log_{Q_{q,k}}(p)}$. Let $n=m+2k+t(\Delta+k)$. Recall the bijection $\iota_{q,k}:\Phi_{q,k}\to\Sigma_{Q_{q,k}}$ from \Cref{def_bijection}. Define the map
    $$
    \begin{array}{rl}
         \varphi_2^S:&\Phi_{q,k}^m\rightarrow\Phi_{q,k}^n \\
         &\quad\bmx\mapsto\bmc
    \end{array}
    $$ as follows:
    \begin{itemize}
        \item Set $\bmc\mid_{[m]}=\bmx$.
        \item For $1\le i\le k$, set $\bmc[m+2i-1]=\bmc[m+2i]=\sparenv{a_i,\ldots,a_i}^\T$, where $a_i=\Sum{\bmx_i}\pmod{q}$.
        \item For each $j\in[t]$, put $r_j=m+2k+(j-1)(\Delta+k)$ and $I_j=[r_j+1,r_j+\Delta]$. Denote $h_j=\sum_{i=1}^k m_{j,i}\overline{\vt(\bmx_i)}\pmod{p}$, which is an integer in $[0,p-1]$. Set
        $$
        \bmc\mid_{I_j}=\iota_{q,k}^{-1}\parenv{\expan{Q_{q,k},\Delta}{h_j}}.
        $$
        For every $i\in[k]$, put $b_{i,j}=\Sum{\bmc_i\mid_{I_j}}\pmod q$ and set $\bmc[r_j+\Delta+i] =\sparenv{b_{i,j},\ldots,b_{i,j}}^\T$.
    \end{itemize}
    Define $\cC_2^S=\mathset{\varphi_2^S(\bmx):\bmx\in\Phi_{q,k}^m}$.
\end{construction}
\Cref{fig_t1CECC_encoding} shows the structure of the codeword for general parameters.
\begin{figure}[!t]
    \centering
    {\color{black}
\begin{tikzpicture}[
    x=1cm,y=1cm,font=\small,text=black,draw=black,
    every node/.style={align=center,inner sep=2pt},
    block/.style={draw=black!65,line width=.65pt},
    note/.style={draw=black!50,line width=.55pt,rounded corners=2pt}
]
\node at (7.15,1.48) {$\bmc\in\Phi_{q,k}^{n}$};
\node[font=\footnotesize,anchor=east] at (-.25,1.07) {Columns:};
\foreach \x/\lab in {1.25/m,4.25/{2k},7/\Delta,8.85/k,11.6/\Delta,13.45/k}
    \node at (\x,1.07) {$\lab$};
\draw[block,fill=blue!5] (0,0) rectangle (2.5,.78);
\draw[block,fill=black!5] (2.5,0) rectangle (6,.78);
\draw[block,fill=blue!5] (6,0) rectangle (8,.78);
\draw[block,fill=black!5] (8,0) rectangle (9.7,.78);
\draw[block,fill=blue!5] (10.6,0) rectangle (12.6,.78);
\draw[block,fill=black!5] (12.6,0) rectangle (14.3,.78);
\foreach \x/\lab in {1.25/\bmx,4.25/\bmA,7/\bmH_1,8.85/\bmB_1,10.15/\cdots,11.6/\bmH_t,13.45/\bmB_t}
    \node at (\x,.39) {$\lab$};
\draw[<->,>=Stealth,line width=.55pt] (-.32,0)--(-.32,.78);
\node[anchor=east] at (-.40,.39) {$k$ rows};
\node[font=\footnotesize] at (1.25,-.23) {Message};
\node[font=\footnotesize] at (4.25,-.23) {Two copies of each row sum};
\node[font=\footnotesize] at (7,-.23) {$I_1$};
\node[font=\footnotesize] at (8.85,-.23) {Row checks};
\node[font=\footnotesize] at (11.6,-.23) {$I_t$};
\node[font=\footnotesize] at (13.45,-.23) {Row checks};
\node[note,text width=6.65cm,minimum height=1.75cm,fill=white]
    at (3.35,-1.48) {
    $\begin{aligned}
        a_i&=\Sum{\bmx_i}\pmod q,\qquad i\in[k],\\
        v_i&=\overline{\vt(\bmx_i)},\\
        h_j&=\textstyle\sum_{i=1}^k m_{j,i}v_i\pmod p,\qquad j\in[t].
    \end{aligned}$
};
\node[note,text width=6.65cm,minimum height=1.75cm,fill=white]
    at (10.95,-1.48) {
    $\begin{aligned}
        \bmH_j&=\bmc\mid_{I_j}=\iota_{q,k}^{-1}\parenv{\expan{Q,\Delta}{h_j}},\\
        b_{i,j}&=\Sum{\bmc_i\mid_{I_j}}\pmod q,\\
        \bmB_j&=[b_{1,j}\mathbf1_k,\ldots,b_{k,j}\mathbf1_k].
    \end{aligned}$
};
\node[font=\footnotesize] at (7.15,-2.68) {
    $\mathbf1_k=[1,\ldots,1]^{\T},\quad Q=Q_{q,k},\quad
     \Delta=\ceilenv{\log_Q p},\quad n=m+2k+t(\Delta+k)$
};
\end{tikzpicture}
}
    \caption{The codeword structure in \Cref{construction_t1CECC}. The numbers above the blocks give their column counts. The first $m$ columns contain the message. The block $\bmA$ stores two constant columns $a_i\mathbf1_k$ for each $i\in[k]$. For each $j\in[t]$, $\bmH_j$ stores $h_j$ in $\Delta$ columns, and $\bmB_j$ contains the $k$ constant columns that check the row sums of $\bmH_j$.}
    \label{fig_t1CECC_encoding}
\end{figure}

\begin{example}
In \Cref{construction_t1CECC}, take $q=2$, $k=t=3$, $m=4$, and $p=11$. Then $Q_{q,k}=4$, $\Delta=2$, and $n=25$. Let
$$
    \bmx=
    \begin{pmatrix}
        0&0&0&1\\
        0&1&1&1\\
        1&1&1&1
    \end{pmatrix}.
$$
Then $(a_1,a_2,a_3)=(1,1,0)$. Straightforward computation gives $\overline{\vt\parenv{\bmx_1}}=4$, $\overline{\vt\parenv{\bmx_2}}=1$, and $\overline{\vt\parenv{\bmx_3}}=2$. Over $\mathbb F_{11}$,
$$
    \bbM=
    \begin{pmatrix}
        7&8&2\\
        5&7&8\\
        10&5&7
    \end{pmatrix},
    \;
    \bbM
    \begin{pmatrix}4\\1\\2\end{pmatrix}
    =
    \begin{pmatrix}7\\10\\4\end{pmatrix}.
$$
Thus, $(h_1,h_2,h_3)=(7,10,4)$. The bijection $\iota_{2,3}^{-1}$ given in \Cref{def_bijection} sends $0,1,2,3$ to $[0,0,0]^\T$, $[0,0,1]^\T$, $[0,1,1]^\T$, and $[1,1,1]^\T$, respectively. Thus,
\begin{align*}
    \expan{4,2}{h_1}=31\leftrightarrow[1,1,1]^\T,[0,0,1]^\T,\\
    \expan{4,2}{h_2}=22\leftrightarrow[0,1,1]^\T,[0,1,1]^\T,\\
    \expan{4,2}{h_3}=01\leftrightarrow[0,0,0]^\T,[0,0,1]^\T.
\end{align*}
Therefore, $b_{1,1}=b_{2,1}=1$, $b_{3,1}=0$, $b_{1,2}=b_{2,2}=b_{3,2}=0$, $b_{1,3}=b_{2,3}=0$, and $b_{3,3}=1$.
By definition, we have $r_1=10$, $r_2=15$, $r_3=20$, $I_1=[11,12]$, $I_2=[16,17]$, and $I_3=[21,22]$. Thus, the complete codeword is
$$
\setlength{\arraycolsep}{3pt}
\varphi_2^S(\bmx)=
\left(\begin{array}{cccc|cccccc|cc|ccc|cc|ccc|cc|ccc}
0&0&0&1&1&1&1&1&0&0&1&0&1&1&0&0&0&0&0&0&0&0&0&0&1\\
0&1&1&1&1&1&1&1&0&0&1&0&1&1&0&1&1&0&0&0&0&0&0&0&1\\
1&1&1&1&1&1&1&1&0&0&1&1&1&1&0&1&1&0&0&0&0&1&0&0&1
\end{array}\right).
$$
\end{example}

The next theorem guarantees that the code $\cC_2^S$ can correct a $t$-$(1,\ldots,1)$-composite-error.
\begin{theorem}\label{thm_t1CECC}
    The code $\cC_2^S$ is a $t$-$(1,\ldots,1)$-CECC with redundancy $t\ceilenv{\log_{Q_{q,k}}(p)}+2k+tk=t\log_{Q_{q,k}}(n)+O(1)$.
\end{theorem}
\begin{IEEEproof}
Let $\bmc=\varphi_2^S(\bmx)$ be the transmitted codeword, and let $\bmy_i\in\Sigma_q^n$ be the received sequence in row $i\in[k]$. At most $t$ rows contain errors, and each row contains at most one substitution. We show how to recover $\bmx$ from the received matrix $\bmy$.

We first determine which rows contain a substitution in the first $m$ columns. For each $i\in[k]$, compute $S_i=\Sum{\bmy_i\mid_{[m]}}\pmod q$. If $\bmy_i\mid_{[m]}$ is error-free, then $S_i=a_i$, and at least one of $\bmy_i[m+2i-1]$ and $\bmy_i[m+2i]$ still equals $a_i$. Conversely, if a substitution occurred in the first $m$ entries of row $i$, both copies of $a_i$ are error-free. The substitution changes the row sum by a nonzero integer $\delta$ with $|\delta|\le q-1$, so $S_i\ne a_i$. Therefore, the known set
\begin{equation*}
 I=\mathset{i\in[k]:S_i\notin\{\bmy_i[m+2i-1],\bmy_i[m+2i]\}}
\end{equation*}
consists exactly of the rows whose first $m$ entries contain a substitution. Put $s=|I|$ and $I^\prime=[k]\setminus I$. Then $s\le t$, and $\bmx_i=\bmy_i\mid_{[m]}$ is known for every $i\in I^\prime$. If $s=0$, all rows of $\bmx$ are known
and we are done. In the following, assume $s\ge1$.

Put $v_i=\overline{\vt(\bmx_i)}$ for $i\in[k]$. The values $v_i$ with $i\in I^\prime$ can be computed directly. For $i\in I$, the only substitution occurred in the first $m$ entries, so the remainder of row $i$ is error-free. In particular, $a_i=\bmy_i[m+2i-1]=\bmy_i[m+2i]$ is known. By \Cref{lem_qary1sub}, it remains to recover $v_i$ for every $i\in I$.

We next identify enough error-free blocks that store the checks $h_j$. Let $r_j$ and $I_j$ be as in
\Cref{construction_t1CECC}, and define the known set
\begin{equation*}
 J=\mathset{j\in[t]:
   \Sum{\bmy_i\mid_{I_j}}\equiv
     \bmy_i[r_j+\Delta+i]\pmod q
   \text{ for every }i\in I^\prime}.
\end{equation*}
For $i\in I^\prime$, a substitution in
$\bmy_i\mid_{I_j}$ would change its sum modulo $q$, while the
check symbol $\bmy_i[r_j+\Delta+i]$ would remain error-free.
Thus, the displayed congruence certifies that
$\bmy_i\mid_{I_j}$ is error-free. The corresponding rows with
$i\in I$ are already known to have no errors after the first
$m$ entries. Consequently, $\bmy\mid_{I_j}=\bmc\mid_{I_j}$ for every $j\in J$.

We claim that $|J|\ge s$. The rows indexed by $I$ account for $s$ substitutions. Thus, there are at most $t-s$ substitutions in the rows indexed by $I^\prime$. For every $j\notin J$, there exists $i\in I^\prime$ such that $\Sum{\bmy_i\mid_{I_j}}
 \not\equiv \bmy_i[r_j+\Delta+i]\pmod q$.
Therefore, a substitution occurred in row $i$, either within
$I_j$ or at position $r_j+\Delta+i$. In either case, its
column index belongs to $[r_j+1,r_j+\Delta+k]$. These intervals are pairwise disjoint
as $j$ ranges over $[t]$. Hence, a substitution can cause $\Sum{\bmy_i\mid_{I_j}}
 \not\equiv \bmy_i[r_j+\Delta+i]\pmod q$ for at most one index $j$, giving $\abs{[t]\setminus J}\le t-s$ and $|J|\ge s$.

For every $j\in J$, the error-free block
$\bmy\mid_{I_j}$ determines the integer
\begin{equation*}
 h_j=\sum_{\ell=1}^{\Delta}\iota_{q,k}\parenv{\bmy[r_j+\ell]}Q_{q,k}^{\ell-1}.
\end{equation*}
Write $I=\{i_1,\ldots,i_s\}$ with $i_1<\cdots<i_s$, and choose
$j_1<\cdots<j_s$ in $J$. For $\nu\in[s]$, the field element
\begin{equation*}
 w_\nu=h_{j_\nu}-\sum_{i\in I^\prime}m_{j_\nu,i}v_i
 \in\mathbb F_p
\end{equation*}
is known. The definition of $h_j$ gives
\begin{equation*}
 \begin{pmatrix}
  m_{j_1,i_1}&\cdots&m_{j_1,i_s}\\
  \vdots&\ddots&\vdots\\
  m_{j_s,i_1}&\cdots&m_{j_s,i_s}
 \end{pmatrix}
 \begin{pmatrix}v_{i_1}\\\vdots\\v_{i_s}\end{pmatrix}
 =\begin{pmatrix}w_1\\\vdots\\w_s\end{pmatrix}
 \quad\text{over }\mathbb F_p.
\end{equation*}
The coefficient matrix is an $s\times s$ submatrix of
$\bbM$ and is invertible by \Cref{lem_cauchymatrix}. Thus,
this system uniquely determines $v_{i_1},\ldots,v_{i_s}$ as
field elements. Since $0\le v_i<2m(q-1)<p$, their integer values are also uniquely determined. For every
$i\in I$, \Cref{lem_qary1sub} now recovers $\bmx_i$ from
$\bmy_i\mid_{[m]}$, $v_i$, and $a_i$. All rows of $\bmx$ are
therefore recovered.
This proves the required error-correction property.

Finally, the first $m$ columns of $\varphi_2^S(\bmx)$ equal
$\bmx$. Hence, $\varphi_2^S$ is injective and
$|\cC_2^S|=Q_{q,k}^m$. Its redundancy is therefore
\begin{equation*}
 \log_{Q_{q,k}}\frac{Q_{q,k}^n}{|\cC_2^S|}
 =n-m
 =2k+t(\Delta+k)
 =t\ceilenv{\log_{Q_{q,k}}p}+2k+tk.
\end{equation*}
Bertrand's postulate gives $p=\Theta(m)$.
Hence, $\Delta=\log_{Q_{q,k}}m+O(1)$ and
$n=m+O(\log m)$. Now the stated asymptotic
redundancy $t\log_{Q_{q,k}}n+O(1)$ follows.
\end{IEEEproof}

\section{Bounds on CDCCs}\label{sec_bounddeltion}
%%%%%%%%%%%%%%%%%%%%%%%%%%%%%%%%%%%%%
Let $A_{q,k}^D\parenv{n;\parenv{e_1,e_2,\ldots,e_{k}}}$ denote the maximum possible size of a $q$-ary $\parenv{e_1,e_2,\ldots,e_{k}}$-CDCC of length $n$.
The main focus of this section is to derive upper bounds on $A_{q,k}^D\parenv{n;\parenv{e_1,e_2,\ldots,e_{k}}}$. Notice that a $k$-resolution $e$-CDCC is also an $(e_1,\ldots,e_k)$-CDCC whenever $e=\sum_{i=1}^{k}e_i$, and that a $k$-resolution $t$-$(e_1,\ldots,e_t)$-CDCC is also an $(e_1,\ldots,e_k)$-CDCC where $e_{t+1}=\cdots=e_k=0$.
Therefore, bounds on $A_{q,k}^D\parenv{n;\parenv{e_1,e_2,\ldots,e_{k}}}$ also give bounds on the maximum sizes of $e$-CDCCs and $t$-$(e_1,\ldots,e_t)$-CDCCs.

The main results of this section are presented in \Cref{sec_CDCCbounds,sec_CDCCotherbounds}. In \Cref{sec_auxiliaryCDCC} below, we give some auxiliary results that will be used in \Cref{sec_CDCCbounds,sec_CDCCotherbounds}.

\subsection{Auxiliary Results}\label{sec_auxiliaryCDCC}
%%%%%%%%%%%%%%%%%%%%%%%%%%%%%%%%%%%%%%%%%%%%
We begin with some definitions from \cite{Kulkarni013IT}. A \emph{hypergraph} $\cH$ is a pair $\parenv{X,\cE}$, where $X$ is a finite set and $\cE$ is a family of nonempty subsets of $X$. Each subset in $\cE$ is called a \emph{hyperedge}. A \emph{matching} of $\cH$ is a set of pairwise disjoint hyperedges. Define $\nu\parenv{\cH}$ to be the size of a maximum matching of $\cH$. A \emph{fractional transversal} of $\cH$ is a nonnegative function $w$ on $X$ such that $\sum_{x\in E}w(x)\ge1$ for every $E\in \cE$. We have the following lemma.
\begin{lemma}\cite[Lemma 2.4]{Kulkarni013IT}\label{lem_matchingtransversal}
    It holds that $\nu\parenv{\cH}\le\sum_{x\in X}w(x)$ for every fractional transversal $w$ of $\cH$.
\end{lemma}

Let $Q\ge2$, $s\ge0$, and $m\ge s$ be integers. For $\bmx\in\Sigma_Q^m$ and $0\le s\le m$, let $\cD_s(\bmx)$ be the set of all subsequences of length $m-s$ of $\bmx$. A nonempty set $\cC\subseteq\Sigma_Q^m$ is called an $s$-\emph{deletion-correcting code} if $\cD_{s}(\bmx)\cap\cD_s(\bmx^\prime)=\emptyset$ for any distinct $\bmx,\bmx^\prime\in\cC$.

Consider the hypergraph $\cH_{Q,s,m}\triangleq\parenv{\Sigma_Q^{m-s},\mathset{\cD_s(\bmx):\bmx\in\Sigma_Q^m}}$.
For $\bmy\in\Sigma_Q^{m-s}$, set
\begin{equation}\label{eq_weight}
 w_{Q,s,m}(\bmy)=
 \begin{cases}
  1,&s=0\text{ or }s< m\le2s,\\
  \frac{1}{\abs{\cD_s(\bmy)}},&s\ge1\text{ and }m>2s.
 \end{cases}
\end{equation}
For $s=0$ or $s<m\le 2s$, it is easy to verify that $w_{Q,s,m}$ is a fractional transversal of $\cH_{Q,s,m}$. When $m>2s$, the proof of \cite[Theorem 4.1]{Kulkarni013IT} shows that $w_{Q,s,m}$ is a fractional transversal of $\cH_{Q,s,m}$. Moreover, we have the following lemma.
\begin{lemma}\label{lem_ordinarydelbound}
    For $Q\ge2$, $s\ge0$ and $m>s$, let $w_{Q,s,m}$ be defined in \eqref{eq_weight}. Denote $U_{Q,s,m}=\sum_{\bmy\in\Sigma_Q^{m-s}}w_{Q,s,m}(\bmy)$.
    \begin{enumerate}[$(i)$]
        \item(\cite[Lemma 2.4 and Theorem 4.1]{Kulkarni013IT}) Then $U_{Q,s,m}$ is an upper bound on the size of every $s$-deletion-correcting code in $\Sigma_Q^m$.
        \item(\cite[Theorem 3.1]{Kulkarni013IT}) For all $Q,m\ge2$, we have
        \begin{equation*}
            U_{Q,1,m}=\frac{Q^m-Q}{(Q-1)(m-1)}.
        \end{equation*}
        \item(\cite[Theorem 4.3]{Kulkarni013IT}) For fixed $Q$ and $s$ and sufficiently large $m$, we have
        \begin{equation*}
        U_{Q,s,m}=\parenv{1+o(1)}\frac{s!Q^m}{(Q-1)^sm^s}.
        \end{equation*}
    \end{enumerate}
\end{lemma}

Recall that $Q_{q,k}$ is the size of the ordered composite alphabet $\Phi_{q,k}$. Let $s\ge0$, $m>s$, $m_1\ge2$ and $R\ge1$ be integers. Consider the set
$\Sigma_{Q_{q,k}}^m\times(\Sigma_q^{m_1})^R$.
Each sequence in this set is a tuple
$\bmx=(\bmx^{(0)},\ldots,\bmx^{(R)})$, where
$\bmx^{(0)}\in\Sigma_{Q_{q,k}}^m$ and
$\bmx^{(j)}\in\Sigma_q^{m_1}$ for $j\in[R]$.

Suppose that $\bmx$ is transmitted through a channel which deletes $s$ symbols from $\bmx^{(0)}$
and one symbol from each $\bmx^{(j)}$, $j\in[R]$. For simplicity, we refer to this type of error as an \emph{$\parenv{s,1,\ldots,1}$-product-deletion}.
The decoder receives the resulting sequences separately, with their component indices known.
Thus, the set of possible channel outputs of $\bmx$ is $\cB^{\times}(\bmx)\triangleq\cD_s\parenv{\bmx^{(0)}}\times
 \prod_{j=1}^R\cD_1\parenv{\bmx^{(j)}}$. Denote $\Omega^{\times}=\Sigma_{Q_{q,k}}^{m-s}
 \times(\Sigma_q^{m_1-1})^R$. Define the hypergraph $\cH^{\times}=\parenv{X,\cE}$ with $X=\Omega^{\times}$ and $\cE=\mathset{\cB^{\times}(\bmx):\bmx\in \Sigma_{Q_{q,k}}^m\times(\Sigma_q^{m_1})^R}$.

Let $\cC\subseteq\Sigma_{Q_{q,k}}^m\times(\Sigma_q^{m_1})^R$ be a code that corrects an $\parenv{s,1,\ldots,1}$-product-deletion. Then
for any two distinct codewords
$\bmx,\bmx^\prime\in\cC$, we have $\cB^{\times}(\bmx)\cap\cB^{\times}(\bmx^\prime)=\emptyset$.
Therefore, $\mathset{\cB^\times(\bmx):\bmx\in\cC}$
is a matching of $\cH^\times$ of size $\abs{\cC}$. For each $\bmy=(\bmy^{(0)},\ldots,\bmy^{(R)})\in\Omega^\times$,
define
$$
 w(\bmy)=w_{Q_{q,k},s,m}\parenv{\bmy^{(0)}}
 \prod_{j=1}^R w_{q,1,m_1}\parenv{\bmy^{(j)}},
$$
where $w_{Q_{q,k},s,m}$ is a fractional transversal of $\cH_{Q_{q,k},s,m}$ and $w_{q,1,m_1}$ is a fractional transversal of $\cH_{q,1,m_1}$.
Clearly, $w(\bmy)\ge0$ for every $\bmy\in\Omega^\times$. Straightforward computation gives
\begin{align*}
 \sum_{\bmy\in\cB^{\times}(\bmx)}w(\bmy)
 &=\parenv{\sum_{\bmy^{(0)}\in\cD_s\parenv{\bmx^{(0)}}}w_{Q_{q,k},s,m}\parenv{\bmy^{(0)}}}
 \prod_{j=1}^R
 \parenv{
 \sum_{\bmy^{(j)}\in\cD_1\parenv{\bmx^{(j)}}}
 w_{q,1,m_{1}}\parenv{\bmy^{(j)}}}\\
 &\ge1.
\end{align*}
This implies that $w$ is a fractional transversal of the hypergraph $\cH^{\times}$.
\begin{lemma}\label{lem_productbound}
    Let $\cC$ be as above. Then we have $\abs{\cC}\le U_{Q_{q,k},s,m}\parenv{U_{q,1,m_1}}^R$, where $U_{Q_{q,k},s,m}$ and $U_{q,1,m_1}$ are as defined in \Cref{lem_ordinarydelbound} by setting $(Q,s,m)=\parenv{Q_{q,k},s,m}$ and $(Q,s,m)=(q,1,m_1)$, respectively.
\end{lemma}
\begin{IEEEproof}
    Let $w$ be the fractional transversal of $\cH^\times$ defined above. Since $\cC$ determines a matching of size $\abs{\cC}$, it follows from \Cref{lem_matchingtransversal} that
    \begin{align*}
        \abs{\cC}\le\sum_{\bmy\in\Omega^\times}w(\bmy)&=\parenv{\sum_{\bmy^{(0)}\in\Sigma_{Q_{q,k}}^{m-s}}w_{Q_{q,k},s,m}\parenv{\bmy^{(0)}}}\prod_{j=1}^R\parenv{\sum_{\bmy^{(j)}\in\Sigma_q^{m_1-1}}w_{q,1,m_1}\parenv{\bmy^{(j)}}}\\
        &=U_{Q_{q,k},s,m}\parenv{U_{q,1,m_1}}^R.
    \end{align*}
    This completes the proof.
\end{IEEEproof}

\subsection{Bounds on \texorpdfstring{$(e_1,\ldots,e_k)$}{(e1,...,ek)}-CDCCs}\label{sec_CDCCbounds}
For each $i\in[k]$, let
\begin{equation}\label{eq_row-alphabet}
 \cF_i=\mathset{[\,\underbrace{0,\ldots,0}_{i-1},\,a,\,
                 \underbrace{q-1,\ldots,q-1}_{k-i}\,]^\T:a\in\Sigma_q}.
\end{equation}
The $q$ ordered composite symbols in $\cF_i$ differ only in the $i$-th coordinate. Let $m_1\ge2$. Every sequence in $\cF_i^{m_1}$ can be viewed as a
$k\times m_1$ matrix. Its first $i-1$ rows consist
entirely of zeros, and its last $k-i$ rows consist
entirely of $q-1$. Only row $i$ varies. Thus, specifying row $i$ as a sequence in $\Sigma_q^{m_1}$ uniquely determines the entire matrix.

Let $e_0=\min\mathset{e_i:i\in[k]}$ and $R=\sum_{i=1}^{k}(e_i-e_0)$. We first consider the case $R=0$.
\begin{corollary}
    If $R=0$, then $e_1=\cdots=e_k=e_0$. It holds that
    \begin{equation}\label{eq:equal-counts}
    A_{q,k}^D(n;(e_0,\ldots,e_0))\le U_{Q_{q,k},e_0,n},
    \end{equation}
    for every $n>e_0$. In particular, we have
    \begin{equation*}
    A_{q,k}^D(n;(1,\ldots,1))
    \le\frac{{Q_{q,k}}^n-{Q_{q,k}}}{({Q_{q,k}}-1)(n-1)},\qquad n\ge2.
    \end{equation*}
\end{corollary}
\begin{IEEEproof}
    An $(e_0,\ldots,e_0)$-CDCC must correct $e_0$ whole-column deletions. Under the bijection $\iota_{q,k}:\Phi_{q,k}\to\Sigma_{Q_{q,k}}$, it is therefore an ordinary ${Q_{q,k}}$-ary $e_0$-deletion-correcting code. Then the two bounds follow directly from \Cref{lem_ordinarydelbound}.
\end{IEEEproof}
Next assume $R\ge1$. Fix a map $\lambda:[R]\to[k]$ such that $\abs{\mathset{j\in[R]:\lambda(j)=i}}=e_i-e_0$
for every $i$. For a sequence $\bms\in\Phi_{q,k}^n$ and a tuple $\parenv{e_1,e_2,\ldots,e_{k}}\in[0,n]^k$, define
\begin{equation}\label{eq_compositedelball}
\cB_{\parenv{e_1,e_2,\ldots,e_{k}}}^D(\bms)=\mathset{\parenv{\bmy_1,\bmy_2,\ldots,\bmy_{k}}:\bmy_i\in\cD_{e_i}(\bms_i),~\forall~ i\in[k]}.
\end{equation}
It is clear that $\cC$ is an $\parenv{e_1,e_2,\ldots,e_{k}}$-CDCC if and only if $\cB_{\parenv{e_1,e_2,\ldots,e_{k}}}^D(\bms)\cap\cB_{\parenv{e_1,e_2,\ldots,e_{k}}}^D(\bms^\prime)=\emptyset$ for any two distinct $\bms,\bms^\prime\in\cC$.

Now we are ready to give the first upper bound on $A_{q,k}^D\parenv{n;(e_1,\ldots,e_k)}$ for general $(e_1,\ldots,e_k)$.
\begin{theorem}\label{thm_CDCCgenbound}
Suppose $R\ge1$. Let $h,m_1$ be integers satisfying $2\le m_1\le h$ and $m:=n-Rh>e_0$. Put
\begin{equation*}
 M=\floorenv{\frac{h}{m_1}},\,
 G=\sparenv{1-\parenv{1-\parenv{\frac{q}{Q_{q,k}}}^{m_1}}^M}^R,\,
 \alpha=\frac{U_{Q_{q,k},e_0,m}}{{Q_{q,k}}^m}\parenv{\frac{U_{q,1,m_1}}{q^{m_1}}}^R.
\end{equation*}
Then $0<\alpha<1$ and
\begin{equation}\label{eq:exact-finite}
 A_{q,k}^D(n;(e_1,\ldots,e_k))
 \le {Q_{q,k}}^n\bigl[1-G(1-\alpha)\bigr].
\end{equation}
\end{theorem}
\begin{figure}[!t]
\centering
{\color{black}
\begin{tikzpicture}[
    x=1cm,y=1cm,font=\small,text=black,
    interval/.style={draw=black!70,line width=.6pt,fill=black!3},
    dimension/.style={Stealth-Stealth,draw=black!65,line width=.5pt}
]
    \draw[interval] (0,0) rectangle (3.0,.9);
    \draw[interval,fill=blue!5] (3.0,0) rectangle (7.8,.9);
    \draw[interval,fill=white] (7.8,0) rectangle (9.6,.9);
    \draw[interval,fill=blue!5] (9.6,0) rectangle (14.4,.9);
    \node[align=center] at (1.5,.45) {$J_0$\\[-1pt]$[1,m]$};
    \node[align=center] at (5.4,.45) {$J_1$\\[-1pt]$[m+1,m+h]$};
    \node at (8.7,.45) {$\cdots$};
    \node[align=center] at (12,.45) {$J_R$\\[-1pt]$[m+(R-1)h+1,n]$};
    \draw[dimension] (0,1.17)--node[above=1pt]
        {$[n],\qquad n=m+Rh$}(14.4,1.17);
    \draw[dimension] (0,-.25)--node[below=1pt]{$m$ positions}(3.0,-.25);
    \draw[dimension] (3.0,-.25)--node[below=1pt]{$h$ positions}(7.8,-.25);
    \draw[dimension] (9.6,-.25)--node[below=1pt]{$h$ positions}(14.4,-.25);
\end{tikzpicture}
}
\caption{The partition of $[n]$ into $J_0=[m]$ and the consecutive
intervals $J_j=[m+(j-1)h+1,m+jh]$, $j\in[R]$, where $m=n-Rh$.}
\label{fig_cdcc_intervals}
\end{figure}
\begin{figure}[!t]
\centering
{\color{black}
\begin{tikzpicture}[
    x=1cm,y=1cm,font=\small,text=black,
    block/.style={draw=black!70,line width=.6pt,fill=blue!5},
    dimension/.style={Stealth-Stealth,draw=black!65,line width=.5pt}
]
    \draw[block] (0,0) rectangle (3.75,.9);
    \draw[block,fill=white] (3.75,0) rectangle (5.05,.9);
    \draw[block] (5.05,0) rectangle (10.05,.9);
    \draw[block,fill=black!5,dashed] (10.05,0) rectangle (14.4,.9);
    \node[align=center,font=\footnotesize] at (1.875,.45)
        {$B_{j,1}$\\[-1pt]$[a_j+1,a_j+m_1]$};
    \node at (4.40,.45) {$\cdots$};
    \node[align=center,font=\footnotesize] at (7.55,.45)
        {$B_{j,M}$\\[-1pt]$[a_j+(M-1)m_1+1,a_j+Mm_1]$};
    \node[align=center,font=\footnotesize] at (12.225,.45)
        {Remainder\\[-1pt]$[a_j+Mm_1+1,a_j+h]$};
    \draw[dimension] (0,1.17)--node[above=1pt]
        {$J_j=[a_j+1,a_j+h],\qquad a_j=m+(j-1)h$}(14.4,1.17);
    \draw[dimension] (0,-.25)--node[below=1pt,font=\footnotesize]
        {$M$ intervals, each with $m_1$ positions}(10.05,-.25);
    \draw[dimension] (10.05,-.25)--node[below=1pt,font=\footnotesize]
        {$0\le h-Mm_1<m_1$ positions}(14.4,-.25);
\end{tikzpicture}
}
\caption{The partition of the first $Mm_1$ positions in $J_j$, $j\in[R]$.}
\label{fig_cdcc_blocks}
\end{figure}
\begin{IEEEproof}
Let $\cC\subseteq\Phi_{q,k}^n$ be an $(e_1,\ldots,e_k)$-CDCC. The idea is to establish a partition $\cC=\cC_1\cup\cC_2$ and bound the sizes of $\cC_1$ and $\cC_2$ from above separately. To that end, we choose a subset $\cG$ of $\Phi_{q,k}^n$, and define $\cC_1=\cC\cap\cG$ and $\cC_2=\cC\setminus\cG$. The set $\cG$ is chosen carefully such that $\cC\cap\cG$ is the union of some codes that correct product-deletions. Then we can bound $\abs{\cC\cap\cG}$ using \Cref{lem_productbound}. The size of $\cC_2$ is trivially bounded by $\abs{\Phi_{q,k}^n\setminus\cG}$.

\vspace{5pt}
(1) \textbf{The set} $\cG$

Partition $[n]$ into the following
$R+1$ intervals, as shown in \Cref{fig_cdcc_intervals}:
$$
 J_0=[m],\quad
 J_j=[m+(j-1)h+1,m+jh],\quad j\in[R].
$$
For each $j\in[R]$ and $u\in[M]$, define $B_{j,u}=
[m+(j-1)h+(u-1)m_1+1,m+(j-1)h+um_1]$. Thus, $B_{j,1},\ldots,B_{j,M}$ partition the first
$Mm_1$ positions of $J_j$ into $M$ intervals
of length $m_1$. The remaining $h-Mm_1$ positions, if any,  are not included in these intervals. See \Cref{fig_cdcc_blocks}.

Let $\lambda: [R]\to[k]$ be the map defined above. Define
$$
 \cG
 =
 \mathset{
 \bmc\in\Phi_{q,k}^n:
 \text{for every }j\in[R],\
 \bmc\mid_{B_{j,u}}\in\cF_{\lambda(j)}^{m_1}
 \text{ for some }u\in[M]
 }.
$$
The structure of sequences in $\cG$ is illustrated in \Cref{fig_cdcc_G_structure}. We first compute $\abs{\cG}$. Let $\bmc\in\Phi_{q,k}^n$. It is clear that  $\bmc\mid_{J_j}\in\Phi_{q,k}^h$ and $\bmc\mid_{B_{j,u}}\in\Phi_{q,k}^{m_1}$ for every $j\in[R]$ and $u\in[M]$. For every $i\in[k]$, we have $\abs{\Phi_{q,k}^{m_1}\setminus\cF_i^{m_1}}=Q_{q,k}^{m_1}-q^{m_1}$. Then for every $j\in[R]$, since $B_{j,1},\ldots,B_{j,M}$ are pairwise disjoint, the number of choices for $\bmc\mid_{J_j}$ with $\bmc\mid_{B_{j,u}}\in\cF_{\lambda(j)}^{m_1}$ for some $u\in[M]$ is $Q_{q,k}^h-\parenv{Q_{q,k}^{m_1}-q^{m_1}}^MQ_{q,k}^{h-Mm_1}=Q_{q,k}^h\parenv{1-\parenv{1-\frac{q^{m_1}}{Q_{q,k}^{m_1}}}^M}$. Since $J_0,J_1,\ldots,J_R$ are pairwise disjoint, we conclude that
\begin{align*}
    \abs{\cG}&=Q_{q,k}^m\sparenv{Q_{q,k}^h\parenv{1-\parenv{1-\frac{q^{m_1}}{Q_{q,k}^{m_1}}}^M}}^R=Q_{q,k}^n\sparenv{1-\parenv{1-\frac{q^{m_1}}{Q_{q,k}^{m_1}}}^M}^R=GQ_{q,k}^n.
\end{align*}

For each $\bmc\in\cG$ and $j\in[R]$, let $u_j(\bmc)
=\min\mathset{u\in[M]:\bmc\mid_{B_{j,u}}\in\cF_{\lambda(j)}^{m_1}}$. Thus, $u_j(\bmc)$ identifies the leftmost interval in $J_j$ whose corresponding substring belongs to $\cF_{\lambda(j)}^{m_1}$. For
$\bmu=(u_1,\ldots,u_R)\in[M]^R$, put
$$
 K_{\bmu}=[n]\setminus\parenv{
 J_0\cup\bigcup_{j=1}^R B_{j,u_j}}.
$$
For each $\bmz\in\Phi_{q,k}^{\abs{K_{\bmu}}}$, define $\cG_{\bmu,\bmz}=\mathset{\bmc\in\cG:u_j(\bmc)=u_j\text{ for all }j\in[R],\ \bmc\mid_{K_{\bmu}}=\bmz}$.
Let $\cP$ be the collection of all pairs $\parenv{\bmu,\bmz}$ with $\cG_{\bmu,\bmz}\ne\emptyset$. Each sequence in $\cG$ determines exactly one pair $(\bmu,\bmz)$, so we obtain the following partition of $\cG$:
$$
\cG=\cup_{(\bmu,\bmz)\in\cP}\cG_{\bmu,\bmz}.
$$

Let $\iota_{q,k}:\Phi_{q,k}\to\Sigma_{Q_{q,k}}$ be the bijection defined in \Cref{def_bijection}. For a given $\parenv{\bmu,\bmz}\in\cP$, define the map
$$
 \Psi_{\bmu,\bmz}:\cG_{\bmu,\bmz}
 \longrightarrow
 \Sigma_{Q_{q,k}}^m\times(\Sigma_q^{m_1})^R,
 \quad
 \bmc\longmapsto
 (\bmx^{(0)},\ldots,\bmx^{(R)}),
$$
where $\bmx^{(0)}=\iota_{q,k}\parenv{\bmc\mid_{[m]}}$ and $\bmx^{(j)}=\bmc_{\lambda(j)}\mid_{B_{j,u_j}}$ for every $j\in[R]$. Here and below, $\iota_{q,k}$ and its inverse act on sequences by applying
the corresponding map to each coordinate. By definition, each $\bmc\in\cG_{\bmu,\bmz}$ is uniquely determined by $\bmc\mid_{[m]}$ and $\parenv{\bmc_{\lambda(1)}\mid_{B_{1,u_1}},\ldots,\bmc_{\lambda(R)}\mid_{B_{R,u_R}}}$. Conversely, the substring on $J_0$ and the varying row
in each selected interval $B_{j,u_j}$ can be chosen
arbitrarily. All earlier intervals $B_{j,u}$ with
$u<u_j$ are fixed by $\bmz$ and remain outside
$\cF_{\lambda(j)}^{m_1}$, whereas each selected interval remains in $\cF_{\lambda(j)}^{m_1}$.
Hence, these choices preserve $u_j(\bmc)=u_j$ for every $j\in[R]$ and yield a sequence in $\cG_{\bmu,\bmz}$. Thus, $\Psi_{\bmu,\bmz}$ is a bijection. In particular, we have $\abs{\cG_{\bmu,\bmz}}=Q_{q,k}^mq^{Rm_1}$.
\begin{figure}[!t]
\centering
\begin{tikzpicture}[
    x=1cm,y=1cm,font=\small,text=black,
    interval/.style={draw=black!65,fill=black!3,line width=.6pt},
    good/.style={draw=blue!65!black,fill=blue!15,line width=.8pt},
    dimension/.style={Stealth-Stealth,draw=black!60,line width=.5pt}
]
\node at (6.9,4.45) {$\bmc\in\cG\subseteq\Phi_{q,k}^{n}$};
\draw[interval] (0,3.05) rectangle (2.6,3.65);
\node at (1.3,3.35) {Any $\bmx\in\Phi_{q,k}^{m}$};
\node at (1.3,3.98) {$J_0$: $m$ columns};
\foreach \xl/\xr/\xc/\lab in {
    2.6/5.8/4.2/1,
    6.6/9.8/8.2/j,
    10.6/13.8/12.2/R
}{
    \draw[interval] (\xl,3.05) rectangle (\xr,3.65);
    \node at (\xc,3.98) {$J_{\lab}$: $h$ columns};
}
\node at (6.2,3.35) {$\cdots$};
\node at (10.2,3.35) {$\cdots$};
\draw[good] (3.65,3.05) rectangle (4.45,3.65);
\draw[good] (7.8,3.05) rectangle (8.6,3.65);
\draw[good] (11.9,3.05) rectangle (12.7,3.65);
\node[font=\footnotesize] at (4.05,2.80) {At least one};
\node at (8.2,2.80) {$B_{j,u}$};
\node[font=\footnotesize] at (12.3,2.80) {At least one};
\draw[-{Stealth},draw=blue!65!black,line width=.7pt]
    (8.2,2.59)--(8.2,2.35)--(4.5,2.35)--(4.5,2.13);
\node at (4.5,1.88)
    {$\bmc\mid_{B_{j,u}}\in\cF_i^{m_1},\qquad i=\lambda(j)$};
\fill[blue!5] (1.8,0) rectangle (7.2,1.5);
\fill[blue!20] (1.8,.5) rectangle (7.2,1.0);
\draw[draw=blue!65!black,line width=.8pt] (1.8,0) rectangle (7.2,1.5);
\draw[draw=blue!65!black,line width=.6pt]
    (1.8,.5)--(7.2,.5) (1.8,1)--(7.2,1);
\node at (4.5,1.25) {All $0$};
\node at (4.5,.75) {$a_1\quad\cdots\quad a_{m_1}$};
\node at (4.5,.25) {All $q-1$};
\node[anchor=west] at (7.5,1.25) {Rows $1,\ldots,i-1$};
\node[anchor=west] at (7.5,.75)
    {Row $i$: $(a_1,\ldots,a_{m_1})\in\Sigma_q^{m_1}$};
\node[anchor=west] at (7.5,.25) {Rows $i+1,\ldots,k$};
\draw[dimension] (1.35,0)--node[left=2pt]{$k$ rows}(1.35,1.5);
\draw[dimension] (1.8,-.25)--node[below=1pt]{$m_1$ columns}(7.2,-.25);
\end{tikzpicture}
\caption{Structure of sequences in $\cG$. Each $J_j$, $j\in[R]$, contains at least one interval $B_{j,u}$ with $\bmc\mid_{B_{j,u}}\in\cF_{\lambda(j)}^{m_1}$.}
\label{fig_cdcc_G_structure}
\end{figure}

\vspace{10pt}
(2) \textbf{An upper bound on} $\abs{\cC}$

We show that every $\Psi_{\bmu,\bmz}\parenv{\cC\cap\cG_{\bmu,\bmz}}$, if nonempty, corrects an $\parenv{e_0,1,\ldots,1}$-product-deletion. Suppose on the
contrary that two distinct codewords
$\bmc,\bmc^\prime\in\cC\cap\cG_{\bmu,\bmz}$ have images
$$
 \Psi_{\bmu,\bmz}(\bmc)=
 (\bmx^{(0)},\ldots,\bmx^{(R)}),
 \quad
 \Psi_{\bmu,\bmz}(\bmc^\prime)
 =
 (\bmx^{\prime(0)},\ldots,\bmx^{\prime(R)})
$$
with a common received tuple
$(\bmy^{(0)},\ldots,\bmy^{(R)})$. Then we have $\bmy^{(0)}\in\cD_{e_0}(\bmx^{(0)})\cap\cD_{e_0}(\bmx^{\prime(0)})$ and $ \bmy^{(j)}\in\cD_1(\bmx^{(j)})\cap
\cD_1(\bmx^{\prime(j)})$ for every $j\in[R]$.

Next, we explicitly construct a common received tuple for $\bmc$ and $\bmc^\prime$ (see also \Cref{fig_cdcc_common_received}). For
each $j\in[R]$, define the positions in $J_j$ before and after the selected interval $B_{j,u_j}$ by $P_j=\{v\in J_j:v<\min B_{j,u_j}\}$ and $T_j=\{v\in J_j:v>\max B_{j,u_j}\}$, respectively. Both sets are contained in $K_{\bmu}$. Thus, for every $i\in[k]$, it holds that
$$
\bma_{i,j}:=\bmc_i\mid_{P_j}=\bmc_i'\mid_{P_j},\quad
 \bmb_{i,j}:=\bmc_i\mid_{T_j}=\bmc_i'\mid_{T_j}.
$$
A subsequence indexed by the empty set is the empty sequence. In particular,
these definitions also apply when $u_j=1$ or $u_jm_1=h$.

Let $\widehat{\bmc}^{(0)}=\iota_{q,k}^{-1}(\bmy^{(0)})\in\Phi_{q,k}^{m-e_0}$,
and let $\widehat{\bmc}_i^{(0)}$ be its $i$-th row. Since $\bmy^{(0)}\in\cD_{e_0}(\bmx^{(0)})\cap\cD_{e_0}(\bmx^{\prime(0)})$, we have $\widehat{\bmc}_i^{(0)}\in
 \cD_{e_0}(\bmc_i\mid_{J_0})
 \cap\cD_{e_0}(\bmc_i'\mid_{J_0})$
for every $i\in[k]$. For $i\in[k]$ and $j\in[R]$, define
$$
 \widehat{\bmc}_i^{(j)}=
 \begin{cases}
  0^{m_1},&i<\lambda(j),\\
  \bmy^{(j)},&i=\lambda(j),\\
  (q-1)^{m_1},&i>\lambda(j).
 \end{cases}
$$
Put $\varepsilon_{i,j}=1$ if $i=\lambda(j)$,
and $\varepsilon_{i,j}=0$ otherwise. It follows from the definition of $\cF_{\lambda(j)}$ and the choice of $\bmy^{(j)}$ that $\widehat{\bmc}_i^{(j)}\in
 \cD_{\varepsilon_{i,j}}(\bmc_i\mid_{B_{j,u_j}})
 \cap\cD_{\varepsilon_{i,j}}(\bmc_i'\mid_{B_{j,u_j}})$.
For each $i\in[k]$, concatenate these sequences in their original order, as shown in \Cref{fig_cdcc_common_received}~(c):
$$
 \bmv_i=\widehat{\bmc}_i^{(0)}
 \bma_{i,1}\widehat{\bmc}_i^{(1)}\bmb_{i,1}
 \cdots
 \bma_{i,R}\widehat{\bmc}_i^{(R)}\bmb_{i,R}.
$$
This is a subsequence of $\bmc_i$. Since $|P_j|+m_1+|T_j|=h$, we have
$$
 |\bmv_i|=(m-e_0)+\sum_{j=1}^R(h-\varepsilon_{i,j})
 =n-e_0-|\{j\in[R]:\lambda(j)=i\}|=n-e_i,
$$
which implies
$\bmv_i\in\cD_{e_i}(\bmc_i)\cap\cD_{e_i}(\bmc_i')$ for every $i\in[k]$.
Consequently, $(\bmv_1,\ldots,\bmv_k)\in
 \cB_{(e_1,\ldots,e_k)}^D(\bmc)
 \cap\cB_{(e_1,\ldots,e_k)}^D(\bmc')$, contrary to \eqref{eq_compositedelball} and the assumption that $\cC$ is an $(e_1,\ldots,e_k)$-CDCC.

We may therefore apply~\Cref{lem_productbound}
to $\Psi_{\bmu,\bmz}\parenv{\cC\cap\cG_{\bmu,\bmz}}$. Since $\Psi_{\bmu,\bmz}$ is injective, $\abs{\cC\cap\cG_{\bmu,\bmz}}=\abs{\Psi_{\bmu,\bmz}\parenv{\cC\cap\cG_{\bmu,\bmz}}}\le U_{Q_{q,k},e_0,m}U_{q,1,m_1}^R=\alpha\abs{\cG_{\bmu,\bmz}}$.
Summing over all pairs $(\bmu,\bmz)\in\cP$ gives
\begin{align*}
 \abs{\cC}=\sum_{(\bmu,\bmz)\in\cP}\abs{\cC\cap\cG_{\bmu,\bmz}}+
 \abs{\cC\setminus\cG}\le
 \alpha\sum_{(\bmu,\bmz)\in\cP}\abs{\cG_{\bmu,\bmz}}+
 \abs{\Phi_{q,k}^n\setminus\cG}=
 \alpha GQ_{q,k}^n+(1-G)Q_{q,k}^n=
 Q_{q,k}^n[1-G(1-\alpha)].
\end{align*}
Since $\cC$ was arbitrary, this proves the
stated upper bound.

Finally, all weights in~\eqref{eq_weight}
lie in $(0,1]$. Therefore,
\[
 0<U_{Q_{q,k},e_0,m}
 \le Q_{q,k}^{m-e_0}
 \le Q_{q,k}^m,
 \quad
 0<U_{q,1,m_1}\le q^{m_1-1}<q^{m_1}.
\]
Together with the definition of $\alpha$, these
inequalities give $0<\alpha<1$.
\end{IEEEproof}
\begin{figure}[!t]
\centering
{\color{black}
\begin{tikzpicture}[
    x=1cm,y=.62cm,font=\small,text=black,
    piece/.style={draw=black!65,line width=.6pt},
    initial/.style={piece,fill=blue!7},
    selected/.style={piece,fill=orange!12},
    fixed/.style={piece,fill=black!4},
    deletion/.style={-{Stealth[length=2mm]},draw=black!75,line width=.7pt},
    dimension/.style={Stealth-Stealth,draw=black!60,line width=.5pt}
]
    % (a) The initial interval: both rows yield the same subsequence.
    \node[font=\small\bfseries] at (2.1,6.45)
        {(a) Initial interval $J_0$};
    \draw[initial] (.65,4.5) rectangle (3.55,5.3);
    \draw[initial,line width=.9pt] (.65,2.65) rectangle (3.55,3.45);
    \draw[initial] (.65,.8) rectangle (3.55,1.6);
    \node at (2.1,4.9) {$\bmc_i\mid_{J_0}$};
    \node at (2.1,3.05) {$\widehat{\bmc}_i^{(0)}$};
    \node at (2.1,1.2) {$\bmc_i'\mid_{J_0}$};
    \draw[deletion] (2.1,4.46)--node[right=3pt,font=\footnotesize]
        {$e_0$ deletions}(2.1,3.49);
    \draw[deletion] (2.1,1.64)--node[right=3pt,font=\footnotesize]
        {$e_0$ deletions}(2.1,2.61);
    \node[font=\footnotesize] at (2.1,.2)
        {Length of the common part: $m-e_0$};

    % (b) A later interval: the two flanks are fixed and unchanged.
    \draw[black!25,dashed] (4.55,.7)--(4.55,5.85);
    \node[font=\small\bfseries] at (9.8,6.45)
        {(b) Later interval $J_j$, for each $j\in[R]$};
    \node at (7.45,5.65) {$P_j$};
    \node at (9.8,5.65) {$B_{j,u_j}$};
    \node at (12.15,5.65) {$T_j$};
    \foreach \ylow/\yhigh in {4.5/5.3,2.65/3.45,.8/1.6}{
        \draw[fixed] (6.7,\ylow) rectangle (8.2,\yhigh);
        \draw[selected] (8.2,\ylow) rectangle (11.4,\yhigh);
        \draw[fixed] (11.4,\ylow) rectangle (12.9,\yhigh);
    }
    \foreach \ycenter in {4.9,3.05,1.2}{
        \node at (7.45,\ycenter) {$\bma_{i,j}$};
        \node at (12.15,\ycenter) {$\bmb_{i,j}$};
    }
    \node[anchor=east] at (6.45,4.9) {$\bmc_i\mid_{J_j}$};
    \node[anchor=east,font=\footnotesize] at (6.45,3.05) {Common part};
    \node[anchor=east] at (6.45,1.2) {$\bmc_i'\mid_{J_j}$};
    \node at (9.8,4.9) {$\bmc_i\mid_{B_{j,u_j}}$};
    \node at (9.8,3.05) {$\widehat{\bmc}_i^{(j)}$};
    \node at (9.8,1.2) {$\bmc_i'\mid_{B_{j,u_j}}$};
    \draw[deletion] (9.8,4.46)--node[right=3pt,font=\footnotesize]
        {$\varepsilon_{i,j}$ deletions}(9.8,3.49);
    \draw[deletion] (9.8,1.64)--node[right=3pt,font=\footnotesize]
        {$\varepsilon_{i,j}$ deletions}(9.8,2.61);
    \node[font=\footnotesize] at (9.8,.2)
        {$\varepsilon_{i,j}=1$ if $i=\lambda(j)$, and $0$ otherwise};

    % (c) Concatenate one common part from every interval.
    \node[font=\small\bfseries] at (6.65,-.65)
        {(c) Concatenate the common parts in their original order};
    \node[anchor=east] at (.95,-1.5) {$\bmv_i=$};
    \draw[initial] (1.2,-1.9) rectangle (3.5,-1.1);
    \draw[fixed] (3.5,-1.9) rectangle (4.75,-1.1);
    \draw[selected] (4.75,-1.9) rectangle (6.75,-1.1);
    \draw[fixed] (6.75,-1.9) rectangle (8,-1.1);
    \draw[fixed] (8.8,-1.9) rectangle (10.05,-1.1);
    \draw[selected] (10.05,-1.9) rectangle (12.05,-1.1);
    \draw[fixed] (12.05,-1.9) rectangle (13.3,-1.1);
    \node at (2.35,-1.5) {$\widehat{\bmc}_i^{(0)}$};
    \node at (4.125,-1.5) {$\bma_{i,1}$};
    \node at (5.75,-1.5) {$\widehat{\bmc}_i^{(1)}$};
    \node at (7.375,-1.5) {$\bmb_{i,1}$};
    \node at (8.4,-1.5) {$\cdots$};
    \node at (9.425,-1.5) {$\bma_{i,R}$};
    \node at (11.05,-1.5) {$\widehat{\bmc}_i^{(R)}$};
    \node at (12.675,-1.5) {$\bmb_{i,R}$};
    \draw[dimension] (1.2,-2.2)--node[below=2pt,font=\footnotesize]
        {$m-e_0$ symbols}(3.5,-2.2);
    \draw[dimension] (3.5,-2.2)--node[below=2pt,font=\footnotesize]
        {$h-\varepsilon_{i,1}$ symbols}(8,-2.2);
    \draw[dimension] (8.8,-2.2)--node[below=2pt,font=\footnotesize]
        {$h-\varepsilon_{i,R}$ symbols}(13.3,-2.2);
\end{tikzpicture}
}
\caption{Construction of the common received sequence $\bmv_i$ for a
fixed row $i\in[k]$. In (a), deleting $e_0$ entire columns in $J_0$
gives the same row $\widehat{\bmc}_i^{(0)}$ from both codewords.
In (b), the parts $\bma_{i,j}$ and $\bmb_{i,j}$ are identical in the
two codewords and undergo no deletions. The selected interval $B_{j,u_j}$ gives
$\widehat{\bmc}_i^{(j)}$ after $\varepsilon_{i,j}$ deletions. Panel (c) joins these common parts to form $\bmv_i$.}
\label{fig_cdcc_common_received}
\end{figure}

The following corollary follows from
\Cref{thm_CDCCgenbound} and \Cref{lem_ordinarydelbound}~$(ii)$.
\begin{corollary}\label{cor_CDCCdgenbound}
Under the assumptions of \Cref{thm_CDCCgenbound},
\begin{equation}\label{eq:simplified-finite}
\begin{split}
 A_{q,k}^D(n;(e_1,\ldots,e_k))
 \le Q_{q,k}^n\sparenv{\frac{U_{Q_{q,k},e_0,m}}{Q_{q,k}^m}
 \parenv{\frac{1-q^{1-m_1}}{(q-1)(m_1-1)}}^R+R\e^{-M\parenv{\frac{q}{Q_{q,k}}}^{m_1}}}.
\end{split}
\end{equation}
\end{corollary}
\begin{IEEEproof}
By \Cref{lem_ordinarydelbound}~$(ii)$,
\[
 \frac{U_{q,1,m_1}}{q^{m_1}}
 =\frac{1-q^{1-m_1}}{(q-1)(m_1-1)}.
\]
Thus the first term in brackets in~\eqref{eq:simplified-finite} equals
$\alpha$. Put $p=(q/Q_{q,k})^{m_1}$ and $z=(1-p)^M$. Then
$G=(1-z)^R$ and $0\le z<1$. Therefore, $1-G=z\sum_{j=0}^{R-1}(1-z)^j
 \le Rz=R(1-p)^M\le R\cdot\e^{-Mp}$.
Since $0< G\le1$ and $\alpha>0$, it follows that
$1-G(1-\alpha)=1-G+G\alpha\le R\e^{-Mp}+\alpha$.
Substituting this inequality into~\eqref{eq:exact-finite} proves the
corollary.
\end{IEEEproof}
\begin{remark}
    The bounds in \Cref{thm_CDCCgenbound} and \Cref{cor_CDCCdgenbound} may be minimized over all integer pairs $(h,m_1)$ satisfying $2\le m_1\le h$ and $m=n-Rh>e_0$, with $M=\floorenv{h/m_1}$.
\end{remark}

The above corollary gives the following asymptotic bounds.
\begin{theorem}\label{thm_CDCCasymbound1}
Let $q,k\ge2$ and the vector $(e_1,\ldots,e_k)\ne\parenv{0,\ldots,0}$ be fixed. For sufficiently large $n$, it holds that
\begin{equation}\label{eq_CDCCasymptotic}
 A_{q,k}^D(n;(e_1,\ldots,e_k))
 \le(1+o(1))\frac{e_0!}{({Q_{q,k}}-1)^{e_0}}
 \parenv{\frac{\ln({Q_{q,k}}/q)}{q-1}}^R
 \frac{{Q_{q,k}}^n}{n^{e_0}(\ln n)^R}.
\end{equation}
Therefore, the redundancy of an $(e_1,\ldots,e_k)$-CDCC is at least
\begin{equation}\label{eq_CDCCredundancyasymp}
 e_0\log_{Q_{q,k}} n+R\log_{Q_{q,k}}\ln n-O(1).
\end{equation}
\end{theorem}
\begin{IEEEproof}
If $R=0$, then $e_0\ge1$, because the deletion counts are not all zero.
Combining~\eqref{eq:equal-counts} with
\Cref{lem_ordinarydelbound}~(iii) gives
\eqref{eq_CDCCasymptotic} in this case.

Suppose now that $R\ge1$. Since $k\ge2$,
$Q_{q,k}\ge\binom{q+1}{2}>q$. Hence
$c:=\ln(Q_{q,k}/q)>0$. In \Cref{cor_CDCCdgenbound}, choose
\begin{equation}\label{eq:asymptotic-parameters}
 h=\floorenv{\frac{n}{\ln n}},\quad
 m_1=\floorenv{\frac{\ln n-(R+4)\ln\ln n}{c}},\quad
 m=n-Rh.
\end{equation}
When $n$ is sufficiently large, these parameters satisfy $m=n(1+o(1))$, $m_1=\frac{\ln n}{c}(1+o(1))$ and $M=\floorenv{\frac{h}{m_1}}=\frac{cn}{(\ln n)^2}(1+o(1))$.
In particular, $2\le m_1\le h$ and $m>e_0$ for all sufficiently large $n$. Thus, we can apply \Cref{cor_CDCCdgenbound}.

For $e_0\ge1$, \Cref{lem_ordinarydelbound}~$(iii)$ gives
\[
 \frac{U_{Q_{q,k},e_0,m}}{Q_{q,k}^m}
 =(1+o(1))\frac{e_0!}{(Q_{q,k}-1)^{e_0}m^{e_0}}.
\]
The same formula holds for $e_0=0$, since
$U_{Q_{q,k},0,m}=Q_{q,k}^m$ by~\eqref{eq_weight}.
Also, $m_1\to\infty$ as $n\to\infty$, so $q^{1-m_1}=o(1)$ and
\[
 \frac{1-q^{1-m_1}}{(q-1)(m_1-1)}
 =(1+o(1))\frac{1}{(q-1)m_1}.
\]
It follows that the first term in brackets in~\eqref{eq:simplified-finite}
can be written as
\begin{align*}
 &\frac{U_{Q_{q,k},e_0,m}}{Q_{q,k}^m}
   \parenv{\frac{1-q^{1-m_1}}{(q-1)(m_1-1)}}^R\\
 &\qquad=(1+o(1))
   \frac{e_0!}{(Q_{q,k}-1)^{e_0}(q-1)^R m^{e_0}m_1^R}\\
 &\qquad=(1+o(1))
   \frac{e_0!}{(Q_{q,k}-1)^{e_0}}
   \parenv{\frac{c}{q-1}}^R
   \frac{1}{n^{e_0}(\ln n)^R}.
\end{align*}
To estimate the second term, the choice of $m_1$ in
\eqref{eq:asymptotic-parameters} gives $\parenv{q/Q_{q,k}}^{m_1}
 =\e^{-cm_1}\ge(\ln n)^{R+4}/n$.
Moreover, $M\ge cn/(2(\ln n)^2)$ for all sufficiently large $n$.
Therefore, $M\parenv{q/Q_{q,k}}^{m_1}
 \ge\frac{c}{2}(\ln n)^{R+2}$.
Multiplying the second term by $n^{e_0}(\ln n)^R$, we obtain
\begin{align*}
n^{e_0}(\ln n)^R R
 \e^{-M\parenv{\frac{q}{Q_{q,k}}}^{m_1}}\le R\e^{
       -\frac{c}{2}(\ln n)^{R+2}+e_0\ln n+R\ln\ln n}
 \longrightarrow0.
\end{align*}
The limit follows because $e_0$ and $R$ are fixed and
$(\ln n)^{R+2}/\ln n\to\infty$. Thus the second term is
$o(n^{-e_0}(\ln n)^{-R})$. Substituting both estimates into
\eqref{eq:simplified-finite} proves~\eqref{eq_CDCCasymptotic}.
Taking logarithms in \eqref{eq_CDCCasymptotic}
and using the definition of redundancy yields
\eqref{eq_CDCCredundancyasymp}.
\end{IEEEproof}

\subsection{Bounds on \texorpdfstring{$e$}{e}-CDCCs and \texorpdfstring{$t$-$(e_1,\ldots,e_t)$}{t-(e1,...,et)}-CDCCs}
\label{sec_CDCCotherbounds}
%%%%%%%%%%%%%%%%%%%%%%%%%%%%%%%%%
Let $A_{q,k}^D(n;e)$ and $A_{q,k}^D(n;t,(e_1,\ldots,e_t))$ denote
the maximum sizes of $q$-ary $k$-resolution $e$-CDCCs and
$t$-$(e_1,\ldots,e_t)$-CDCCs of length $n$, respectively.

An $e$-CDCC must be an $(e_1,\ldots,e_k)$-CDCC
for every tuple $(e_1,\ldots,e_k)\in[0,n]^k$
satisfying $\sum_{i=1}^k e_i=e$. Hence, for $0\le e\le kn$, it holds that
\begin{equation}\label{eq_eCDCC_reduction}
 A_{q,k}^D(n;e)
 \le\min_{\substack{(e_1,\ldots,e_k)\in[0,n]^k\\
                         e_1+\cdots+e_k=e}}
       A_{q,k}^D(n;(e_1,\ldots,e_k)).
\end{equation}
We obtain the following explicit bounds by choosing the $e_i$ to differ by at most one.
\begin{corollary}\label{cor_eCDCC}
Fix $q,k\ge2$ and an integer $e\ge1$. Write $e=ks+r$, where
$s=\floorenv{e/k}$ and $0\le r<k$. For sufficiently large $n$, we have
\begin{equation*}
 A_{q,k}^D(n;e)
 \le(1+o(1))\frac{s!}{(Q_{q,k}-1)^s}
 \parenv{\frac{\ln(Q_{q,k}/q)}{q-1}}^r
 \frac{Q_{q,k}^n}{n^s(\ln n)^r}.
\end{equation*}
Consequently, every $q$-ary $k$-resolution $e$-CDCC has redundancy at least $s\log_{Q_{q,k}}n+r\log_{Q_{q,k}}\ln n-O(1)$.
\end{corollary}
\begin{IEEEproof}
Choose
\begin{equation*}
 e_i=
 \begin{cases}
  s+1,&1\le i\le r,\\
  s,&r<i\le k.
 \end{cases}
\end{equation*}
Since $r<k$, these counts satisfy
$\sum_{i=1}^k e_i=ks+r=e$,
$\min_{i\in[k]}e_i=s$, and $\sum_{i=1}^k(e_i-s)=r$.
Thus, \eqref{eq_eCDCC_reduction} allows us to apply \Cref{thm_CDCCasymbound1} with $(e_0,R)=(s,r)$. Specifically, the bounds on $A_{q,k}^D(n;e)$ and the redundancy follow from \eqref{eq_CDCCasymptotic} and \eqref{eq_CDCCredundancyasymp}, respectively.
\end{IEEEproof}

We next consider $t$-$(e_1,\ldots,e_t)$-CDCCs.
It is clear that a $k$-resolution $t$-$(e_1,\ldots,e_t)$-CDCC must be a $(e_1,\ldots,e_k)$-CDCC, where $e_i=0$ for every $i>t$. Therefore, when $0\le e_i\le n$ for every $i\in[t]$, it holds that
\begin{equation}\label{eq_teCDCC_reduction}
 A_{q,k}^D(n;t,(e_1,\ldots,e_t))
 \le A_{q,k}^D(n;(e_1,\ldots,e_t,\underbrace{0,\ldots,0}_{k-t})).
\end{equation}
The zero entries are omitted when $t=k$.

\begin{corollary}\label{cor_teCDCC}
Fix $q,k\ge2$, $1\le t\le k$, and nonnegative integers
$e_1,\ldots,e_t$ with $E=\sum_{i=1}^t e_i\ge1$. Put
\begin{equation*}
 e_0=
 \begin{cases}
  0,&t<k,\\
  \min_{i\in[k]}e_i,&t=k,
 \end{cases}
 \quad R=E-ke_0.
\end{equation*}
When $n$ is sufficiently large, we have
\begin{equation*}
 A_{q,k}^D(n;t,(e_1,\ldots,e_t))
 \le(1+o(1))\frac{e_0!}{(Q_{q,k}-1)^{e_0}}
 \parenv{\frac{\ln(Q_{q,k}/q)}{q-1}}^R
 \frac{Q_{q,k}^n}{n^{e_0}(\ln n)^R}.
\end{equation*}
Consequently, every $q$-ary $k$-resolution $t$-$(e_1,\ldots,e_t)$-CDCC has redundancy at least $e_0\log_{Q_{q,k}}n+R\log_{Q_{q,k}}\ln n-O(1)$.

In particular, if $t<k$, then $e_0=0$, $R=E$, and
\begin{equation*}
 A_{q,k}^D(n;t,(e_1,\ldots,e_t))
 \le(1+o(1))\parenv{\frac{\ln(Q_{q,k}/q)}{q-1}}^E
 \frac{Q_{q,k}^n}{(\ln n)^E}.
\end{equation*}
The corresponding redundancy lower bound is
$E\log_{Q_{q,k}}\ln n-O(1)$.
\end{corollary}
\begin{IEEEproof}
Let $e_i=0$ for $t<i\le k$.
Then $\min_{i\in[k]}e_i=e_0$ and
$\sum_{i=1}^k(e_i-e_0)=E-ke_0=R$. Now the conclusion follows from \eqref{eq_teCDCC_reduction} and \Cref{thm_CDCCasymbound1}.
\end{IEEEproof}

For example, a $t$-$(1,\ldots,1)$-CDCC with $1\le t<k$ has
size at most
\begin{equation*}
 (1+o(1))\parenv{\frac{\ln(Q_{q,k}/q)}{q-1}}^t
 \frac{Q_{q,k}^n}{(\ln n)^t}.
\end{equation*}
For $t=k$, we instead have $e_0=1$ and $R=0$.
By \eqref{eq_teCDCC_reduction},
\eqref{eq:equal-counts}, and
\Cref{lem_ordinarydelbound}~$(ii)$, we obtain, for $n\ge2$,
\begin{equation*}
 A_{q,k}^D(n;k,(1,\ldots,1))
 \le\frac{Q_{q,k}^n-Q_{q,k}}{(Q_{q,k}-1)(n-1)}.
\end{equation*}

\section{Constructions of CDCCs}\label{sec_consdel}
%%%%%%%%%%%%%%%%%%%%%%%%%%%%%%%%%%%%%%%%%%%%%%%%%
This section presents constructions of $q$-ary $k$-resolution $t$-$(1,\ldots,1)$-CDCCs. The redundancies of the constructed codes have optimal leading term. 

The constructions use ordered composite sequences containing prescribed patterns, as described in \Cref{sec_localization}. When at least one row is error-free, these patterns, together with the nondecreasing property of each column, allow us to
locate the deletion in each affected row within an interval of logarithmic length. A short checksum then suffices to recover each affected row.

We consider the cases $t<k$ and $t=k$ separately. When $t<k$, there is always at least one error-free row. When $t=k$, every row may suffer a deletion. To handle this case, we additionally protect the first row of an intermediate codeword using a differential Varshamov--Tenengolts checksum.
Once this row is recovered, we can apply the decoding procedure for the case $t<k$.

\subsection{Locating Deletions}\label{sec_localization}
%%%%%%%%%%%%%%%%%%%%%%%%%%%%%%%%%%%%%%%%%%%%%%%%
For a given positive integer $m$, define
\begin{equation}\label{eq_blockparameters}
 B_m=2Q_{q,k}^2\ceilenv{\ln(2m)},\quad
 N_m=B_m\floorenv{m/B_m}.
\end{equation}
We have $B_m=\Theta(\ln m)$. Throughout this section, assume $m\ge 2B_{m}$, which holds for all sufficiently large $m$.

For each $i\in[\floorenv{m/B_m}]$, define the interval
\begin{equation}\label{eq_partitioninterval}
 I_i=[(i-1)B_{m}+1,iB_{m}].
\end{equation}
For each $\bmc\in\Phi_{q,k}^m$, we partition the length-$N_m$ prefix of $\bmc$ into $\floorenv{m/B_{m}}$ substrings of length $B_{m}$: $\bmc\mid_{I_i}$, $i\in[\floorenv{m/B_{m}}]$. Fix two ordered composite letters
$$
\sigma_0=[0,\ldots,0]^\T,\quad\sigma_{q-1}=[q-1,\ldots,q-1]^\T.
$$
Let
\begin{equation}\label{eq:restricted-set}
 \cR_{q,k,m}=\mathset{\bmc\in\Phi_{q,k}^m:
 \forall i\in[\floorenv{m/B_{m}}]\text{ there is }
 j\in[(i-1)B_{m}+1,iB_{m}-1]\text{ with }\bmc[j]=\sigma_0,\ \bmc[j+1]=\sigma_{q-1}}.
\end{equation}
In other words, an ordered composite sequence $\bmc$ is included in $\cR_{q,k,m}$ if and only if every substring $\bmc\mid_{I_i}$ contains the pattern $\sigma_0\sigma_{q-1}$.

\begin{lemma}\label{lem_restrictionsize}
For $B_{m}$ in~\eqref{eq_blockparameters}, we have $\abs{\cR_{q,k,m}}\ge Q_{q,k}^m/2$.
\end{lemma}
\begin{IEEEproof}
Choose $\bmc$ uniformly at random from $\Phi_{q,k}^m$. Then $\bmc[1],\ldots,\bmc[m]$ are chosen independently and uniformly at random from $\Phi_{q,k}$. We estimate the probability that $\bmc\in\cR_{q,k,m}$.

Notice that $B_{m}$ is an even integer. Each substring $\bmc\mid_{I_i}$ can be partitioned into $B_{m}/2$ disjoint adjacent pairs beginning at positions $1,3,\ldots,B_{m}-1$, and each of these pairs equals $\sigma_0\sigma_{q-1}$ with probability $Q_{q,k}^{-2}$. Therefore,
\begin{equation*}
\begin{split}
 \Pr\sparenv{\bmc\mid_{I_i}\text{ does not contain }\sigma_0\sigma_{q-1}}\le\parenv{1-\frac{1}{Q_{q,k}^2}}^{\frac{B_{m}}{2}}
 \le\e^{-\frac{B_{m}}{2Q_{q,k}^2}}
 =\e^{-\ceilenv{\ln(2m)}}\le \frac1{2m},
\end{split}
\end{equation*}
for every $i\in\sparenv{\floorenv{m/B_{m}}}$.
Then the union bound gives
$$
\Pr\sparenv{\exists i\in\sparenv{\floorenv{\frac{m}{B_{m}}}}\text{ s.t. }\bmc\mid_{I_i}\text{ does not contain }\sigma_0\sigma_{q-1}}\le\frac{\floorenv{m/B_{m}}}{2m},
$$
which implies $\Pr\sparenv{\bmc\in\cR_{q,k,m}}\ge1-\frac{\floorenv{m/B_{m}}}{2m}$.
Since $\bmc$ is chosen uniformly at random from $\Phi_{q,k}^m$, we have
\begin{equation*}
\abs{\cR_{q,k,m}}=\Pr\sparenv{\bmc\in\cR_{q,k,m}}\cdot\abs{\Phi_{q,k}^m}\ge Q_{q,k}^m\parenv{1-\frac{\floorenv{m/B_{m}}}{2m}}\ge\frac{Q_{q,k}^m}{2}.
\end{equation*}
This completes the proof.
\end{IEEEproof}

In Appendix~\ref{sec_restrictedencoder}, we present an encoder $\mathsf{E}_m$ that encodes any sequence in $\Phi_{q,k}^{m-1}$ into $\cR_{q,k,m}$ and the corresponding decoder.

For two sequences $\bmx,\bmz\in\Sigma_q^m$, we write $\bmx\le\bmz$ (resp. $\bmx\ge\bmz$) if $\bmx[i]\le\bmz[i]$ (resp. $\bmx[i]\ge\bmz[i]$) for every $i\in[m]$.
\begin{definition}\label{def:admissible}
For a sequence $\bmz\in\Sigma_q^m$, define $\cA_{\le}(\bmz)$ to be the set of $\bmx\in\Sigma_q^m$ such that $\bmx\le\bmz$ and, for each interval $I_i$ defined in \eqref{eq_partitioninterval}, there exists $j\in[(i-1)B_{m}+1,iB_{m}-1]$ satisfying
\begin{equation}\label{eq:row-pair}
 \bmx[j]=\bmz[j]=0,\quad
 \bmx[j+1]=\bmz[j+1]=q-1.
\end{equation}
Define $\cA_{\ge}(\bmz)$ by the same condition~\eqref{eq:row-pair} and the inequality $\bmz\le\bmx$.
\end{definition}
If $\bmc\in\cR_{q,k,m}$, it holds that $\bmc_i\in\cA_{\le}(\bmc_{i+1})$ for every $1\le i<k$, and $\bmc_i\in\cA_{\ge}(\bmc_{i-1})$ for every $1<i\le k$.

For $\bmx\in\Sigma_q^m$ and $d\in[m]$, recall that $\bmx\mid_{[m]\setminus\{d\}}$ is the sequence obtained from $\bmx$ by deleting the $d$-th coordinate $\bmx[d]$.
\begin{lemma}\label{lem_localization}
Let $\bmx,\bmx^\prime,\bmz\in\Sigma_q^m$. Suppose that either $\bmx,\bmx^{\prime}\in\cA_{\le}(\bmz)$ or $\bmx,\bmx^{\prime}\in\cA_{\ge}(\bmz)$. If
$\bmx\mid_{[m]\setminus\{d\}}=\bmx^\prime\mid_{[m]\setminus\{d^\prime\}}$ for $d,d^\prime\in[m]$, then
\begin{equation*}
 |d-d^{\prime}|<2B_{m}.
\end{equation*}
The statement also holds when $\bmx=\bmx^{\prime}$.
\end{lemma}
\begin{IEEEproof}
When $d=d^{\prime}$, the result holds. Otherwise, we may assume $d<d^{\prime}$. The equality $\bmx\mid_{[m]\setminus\{d\}}=\bmx^\prime\mid_{[m]\setminus\{d^\prime\}}$ gives
\begin{equation}\label{eq:alignment}
 \bmx[j+1]=\bmx^{\prime}[j]\quad\text{for every }j\in[d,d^{\prime}-1].
\end{equation}
Suppose for a contradiction that $d^{\prime}-d\ge2B_{m}$. Set $a=1+B_{m}\ceilenv{(d-1)/B_{m}}$. Then $d\le a\le d+B_{m}-1$ and $a+B_{m}-1\le d+2B_{m}-2<d^{\prime}$. Thus $[a,a+B_{m}-1]=I_i$ for some $I_i$ defined in \eqref{eq_partitioninterval} and $[a,a+B_{m}-1]\subseteq[d,d^{\prime}-1]$.

If $\bmx,\bmx^{\prime}\in\cA_{\le}(\bmz)$, choose $j\in[a,a+B_{m}-2]$ that satisfies \eqref{eq:row-pair} for $\bmx$ and $\bmz$. It follows from \eqref{eq:row-pair} and \eqref{eq:alignment} that $q-1=\bmx[j+1]=\bmx^{\prime}[j]\le\bmz[j]=0$, a contradiction. If $\bmx,\bmx^{\prime}\in\cA_{\ge}(\bmz)$, choose $j\in[a,a+B_{m}-2]$ that satisfies \eqref{eq:row-pair} for $\bmx^\prime$ and $\bmz$. Then \eqref{eq:row-pair} and \eqref{eq:alignment} give $q-1=\bmz[j+1]\le\bmx[j+1]=\bmx^{\prime}[j]=0$, again a contradiction. Hence $d^{\prime}-d<2B_{m}$ in both cases.
\end{IEEEproof}

We explain why \Cref{lem_localization} is helpful. Suppose that $\bmc\in\cR_{q,k,m}$ is the transmitted codeword and $\bmy$ is obtained from $\bmc$ by a $t$-$(1,\ldots,1)$-composite-deletion. By observing the length of each row of $\bmy$, we can tell which rows suffered a deletion. Suppose that $\bmy_i$ is error-free and that $\bmy_j\in\cD_1(\bmc_j)$. Then $\bmc_i=\bmy_i$. For $d\in[m]$ and $\gamma\in\Sigma_q$, let $\ins_{d,\gamma}(\bmy_j)$ be the sequence obtained from $\bmy_j$ by inserting $\gamma$ at the $d$-th position.  It must be that $\bmc_j=\ins_{d,\gamma}(\bmy_j)$ for some $d\in[m]$ and $\gamma\in\Sigma_q$.

If $j<i$, it is clear that $\bmc_j\in\cA_{\le}(\bmc_i)$. We search for the smallest $d^\prime\in[m]$ such that there is a sequence $\bmc_j^\prime:=\ins_{d^\prime,\gamma^\prime}(\bmy_j)\in\cA_{\le}(\bmc_i)$ for some $\gamma^\prime\in\Sigma_q$. Then \Cref{lem_localization} guarantees $0\le d-d^\prime<2B_{m}$. If $j>i$, the same conclusion holds by a similar argument. So in either case, we can determine the position of the single deletion in $\bmc_j$ to within an interval of length $2B_{m}$.

To recover $\bmc_j$, we need a code which can correct one deletion given knowledge of the position of the deleted symbol to within an interval of length $2B_{m}$. For $\bmx\in\Sigma_q^m$, define the sequence $\psi(\bmx)\in\Sigma_q^m$ as follows:
\begin{equation}\label{eq_differentialsequence}
 \psi(\bmx)[j]=
 \begin{cases}
  (\bmx[j]-\bmx[j+1])\bmod q,&j<m,\\
  \bmx[m],&j=m.
 \end{cases}
\end{equation}
Define
\begin{equation}\label{eq_localsyndrome}
 \begin{split}
 P_m&=2B_{m},\\
 a_m(\bmx)&=\vt(\psi(\bmx))\bmod q(P_m+1),\\
 b_m(\bmx)&=\sum_{j=1}^m\psi(\bmx)[j]\bmod(q+1),\\
 \mathsf{H}_m(\bmx)&=a_m(\bmx)+q(P_m+1)b_m(\bmx).
 \end{split}
\end{equation}
Clearly, $\mathsf{H}_m(\bmx)$ is an integer in $[[q(q+1)(P_m+1)]]$. The two residues $a_m$ and $b_m$ are recovered from $\mathsf{H}_m$ as $a_m=\mathsf{H}_m\bmod q(P_m+1)$ and $b_m=\floorenv{\frac{\mathsf{H}_m}{q(P_m+1)}}$.

The following lemma helps to correct a single deletion when given knowledge of the position of the deleted symbol to within an interval of length $P_m$.
\begin{lemma}(\cite[Lemma~5]{Tuan2024IT})\label{lem_Pbounded}
Let $\bmx,\bmx^{\prime}\in\Sigma_q^m$. If $\mathsf{H}_m(\bmx)=\mathsf{H}_m(\bmx^\prime)$ and
$\bmx\mid_{[m]\setminus\{d\}}=\bmx^\prime\mid_{[m]\setminus\{d^\prime\}}$ for some $d,d^\prime\in[m]$ with
$|d-d^\prime|<P_m$, then $\bmx=\bmx^{\prime}$.
\end{lemma}

\subsection{\texorpdfstring{$t$-$(1,\ldots,1)$}{t-(1,...,1)}-CDCCs with \texorpdfstring{$t<k$}{t<k}}\label{sec_t<k}
%%%%%%%%%%%%%%%%%%%%%%%%%%%%%%%%%%%%%%%%%%%
Our construction is based on the following lemma.
\begin{lemma}\label{lem_rowdecoder}
Let $P_m$ be as in~\eqref{eq_localsyndrome} and $\bmz\in\Sigma_q^m$. Fix $\cA=\cA_{\le}(\bmz)$ or $\cA=\cA_{\ge}(\bmz)$. For every $\bmx\in\cA$, the known set $\cA$, a sequence $\bmy\in\cD_1(\bmx)$, and $\mathsf{H}_m(\bmx)$ uniquely determine $\bmx$.
\end{lemma}
\begin{IEEEproof}
Let $h=\mathsf{H}_m(\bmx)$. For $\bmy$, form the set $ \cK=\mathset{\ins_{d,\gamma}(\bmy):
 d\in[m],\ \gamma\in\Sigma_q,\
 \ins_{d,\gamma}(\bmy)\in\cA,\
 \mathsf{H}_m(\ins_{d,\gamma}(\bmy))=h}$.
Clearly, we have $\bmx\in\cK$, implying $\abs{\cK}\ge1$. If $\bmx,\bmx^{\prime}\in\cK$ are respectively obtained from $\bmy$ by inserting symbols at positions $d,d^{\prime}$, then \Cref{lem_localization} gives $|d-d^{\prime}|<P_{m}$. The equality $\mathsf{H}_m(\bmx)=\mathsf{H}_m(\bmx^\prime)$ and \Cref{lem_Pbounded} imply $\bmx=\bmx^{\prime}$. Thus $\mathcal K=\{\bmx\}$.
\end{IEEEproof}

Recall that the map $\expan{Q,L}{\cdot}$ converts a nonnegative integer $a\in\mathset{0,1,\ldots,Q^L-1}$ into a sequence $\expan{Q,L}{a}\in\Sigma_Q^L$. Now we are ready to present our constructions of $t$-$\parenv{1,\ldots,1}$-CDCCs.
\begin{construction}\label{construction_t1cdcc}
    Let $q\ge2$, $1\le t< k$ and $m$ be integers. Let $B_m$ and $P_m$ be defined in \eqref{eq_blockparameters} and \eqref{eq_localsyndrome}, respectively.  Assume $m\ge P_m$. Let $p$ be the smallest prime greater than $\max\mathset{q(q+1)(P_m+1),k+t}$. Let $\bbM=(m_{j,i})_{t\times k}$ be the matrix in \Cref{def_cauchymatrix}. Let $\mathsf{E}_m:\Phi_{q,k}^{m-1}\to\cR_{q,k,m}$ be the encoder in Appendix~\ref{sec_restrictedencoder}. For $\bmx\in\Phi_{q,k}^{m-1}$ and every $j\in[t]$, define
\begin{equation*}
 A_j(\bmx)=\sum_{i=1}^k m_{j,i}\mathsf{H}_m\parenv{\bmw_i}\in\F_p,
\end{equation*}
where $\bmw=\mathsf{E}_m(\bmx)$. Represent $A_j(\bmx)$ by its integer in $[[p]]$. Denote $\Delta=\ceilenv{\log_{Q_{q,k}}(p)}$, $\bm0=[0,\ldots,0]^\T$ and $\bm1=[1,\ldots,1]^\T$. Let $n=m+2t+t\Delta$.

Let $\iota_{q,k}:\Phi_{q,k}\to\Sigma_{Q_{q,k}}$ be the bijection in \Cref{def_bijection}. Define the map $\varphi_1^D:\Phi_{q,k}^{m-1}\to\Phi_{q,k}^n$ by
\begin{equation*}%\label{eq:cauchy-encoding}
 \varphi_1^D(\bmx)=\mathsf{E}_m(\bmx)\,\bm0\,\bm1\,\bmd_1(\bmx)\cdots
 \bm0\,\bm1\,\bmd_t(\bmx),
\end{equation*}
where $\bmd_j(\bmx)=\iota_{q,k}^{-1}\parenv{\expan{Q_{q,k},\Delta}{A_j(\bmx)}}$ for every $j\in[t]$. Define $\cC_1^D=\mathset{\varphi_1^D(\bmx):\bmx\in\Phi_{q,k}^{m-1}}$.
\end{construction}

\begin{theorem}\label{thm_t1cdcc}
    The code $\cC_1^D$ is a $k$-resolution $t$-$(1,\ldots,1)$-CDCC with redundancy $2t+t\ceilenv{\log_{Q_{q,k}}(p)}+1=t\log_{Q_{q,k}}\ln(n)+O(1)$.
\end{theorem}
\begin{IEEEproof}
Let $\bmc=\varphi_1^D(\bmx)$ be transmitted, and put
$\bmw=\mathsf{E}_m(\bmx)\in\cR_{q,k,m}$. Let $\bmy_i$ be the received sequence in row $i$, for $i\in[k]$. We aim to recover $\bmx$ from $\parenv{\bmy_1,\ldots,\bmy_k}$. Let $\mathsf{D}_m$ be the decoder given in Appendix~\ref{sec_restrictedencoder}. Since $\bmx=\mathsf{D}_m(\bmw)$, it suffices to recover $\bmw$.

The set $U=\{i\in[k]:|\bmy_i|=n-1\}$ is known from the lengths of $\bmy_1,\ldots,\bmy_k$. Clearly, $|U|\le t<k$. For $i\in U$, let $d_i\in[n]$ be the
position of the deleted symbol in $\bmc_i$. For $i\notin U$, we have $\bmy_i=\bmc_i$. For $j\in[t]$, define $r_j=m+(j-1)(\Delta+2)+1$ and $T_j=[r_j+2,r_j+\Delta+1]$. By the construction, $\bmc_i[r_j]=0$ and $\bmc_i[r_j+1]=1$ for
every $i\in[k]$, and $\bmc\mid_{T_j}=\bmd_j(\bmx)$. Partition
$[n]$ into the intervals
\begin{equation*}
 S_0=[1,r_1],\quad
 S_j=[r_j+1,r_{j+1}]\;(1\le j<t),\quad
 S_t=[r_t+1,n].
\end{equation*}
Thus, $S_0=[m+1]$ contains the data part and the first zero
column, while $T_j\subseteq S_j$ for every $j\in[t]$.

We first determine which interval $S_j$ contains each deletion.
For every $i\in U$ and $j\in[t]$, the single-deletion relation gives
\begin{equation*}
 \bmy_i[r_j]=
 \begin{cases}
  \bmc_i[r_j+1]=1,&d_i\le r_j,\\
  \bmc_i[r_j]=0,&d_i>r_j.
 \end{cases}
\end{equation*}
Consequently, the integer $\ell_i:=\min\bigl(\{j\in[t]:\bmy_i[r_j]=1\}\cup\{t+1\}\bigr)-1$
is known and satisfies $d_i\in S_{\ell_i}$.

Let $I=\{i\in U:\ell_i=0\}$ and $s=|I|$. For every
$i\in[k]\setminus I$, no deletion occurs
in $\bmc_i\mid_{[m]}$. Hence, $\bmw_i=\bmy_i\mid_{[m]}$ is known. For
$i\in I$, put $\bmv_i=\bmy_i\mid_{[m-1]}$. Since $d_i\le m+1$, we have 
$$
\bmv_i=\bmw_i\mid_{[m]\setminus\{\min\{d_i,m\}\}}
 \in\cD_1(\bmw_i).
$$
If $s=0$, all rows of $\bmw$ are known and we are done. In the following, assume $s\ge1$.

We next recover $\mathsf H_m(\bmw_i)$ for $i\in I$. Define $J=[t]\setminus\{\ell_i:i\in U\setminus I\}$. Then no row has a deletion in $S_j$ for $j\in J$. Moreover, $|J|\ge t-|U\setminus I|=t-|U|+s\ge s$. For $j\in J$ and $i\in[k]$, set $\varepsilon_{i,j}=1$ if
$i\in U$ and $\ell_i<j$, and set $\varepsilon_{i,j}=0$ otherwise.
If $\varepsilon_{i,j}=1$, the deletion precedes $T_j$; otherwise,
row $i$ has no deletion or its deletion follows $T_j$. Therefore, for every $i\in[k]$, we have
\begin{equation*}
 \bigl(\bmd_j(\bmx)\bigr)_i
 =\bmc_i\mid_{T_j}
 =\bmy_i\mid_{[r_j+2-\varepsilon_{i,j},\,
                    r_j+\Delta+1-\varepsilon_{i,j}]}.
\end{equation*}
The right-hand side is known for every row, so the entire block $\bmd_j(\bmx)$ is recovered. By the definition of
$\expan{Q_{q,k},\Delta}{\cdot}$, this gives the integer $A_j(\bmx)$ by
\begin{equation*}
 A_j(\bmx)=\sum_{a=1}^{\Delta}
 \iota_{q,k}\bigl(\bmd_j(\bmx)[a]\bigr)Q_{q,k}^{a-1},
\end{equation*}
for every $j\in J$.

Write $I=\{i_1,\ldots,i_s\}$ with $i_1<\cdots<i_s$, and choose
$j_1<\cdots<j_s$ in $J$. Put $h_i=\mathsf H_m(\bmw_i)$ for
$i\in[k]$. The values $h_i$ with $i\notin I$ are known because
the corresponding rows $\bmw_i$ are known. For $\nu\in[s]$, the
field element
\begin{equation*}
 b_\nu=A_{j_\nu}(\bmx)
       -\sum_{i\in[k]\setminus I}m_{j_\nu,i}h_i
 \in\mathbb F_p
\end{equation*}
is therefore known. The definition of $A_{j_\nu}$ gives
\begin{equation*}
 \begin{pmatrix}
  m_{j_1,i_1}&\cdots&m_{j_1,i_s}\\
  \vdots&\ddots&\vdots\\
  m_{j_s,i_1}&\cdots&m_{j_s,i_s}
 \end{pmatrix}
 \begin{pmatrix}h_{i_1}\\\vdots\\h_{i_s}\end{pmatrix}
 =\begin{pmatrix}b_1\\\vdots\\b_s\end{pmatrix}
 \quad\text{over }\mathbb F_p.
\end{equation*}
The coefficient matrix is an $s\times s$ submatrix of $\bbM$ and
is invertible by \Cref{lem_cauchymatrix}. Hence this system
uniquely determines $h_{i_1},\ldots,h_{i_s}$ as field elements.
Since
\begin{equation*}
 0\le h_i<q(q+1)(P_m+1)<p,
\end{equation*}
their integer values are also uniquely determined.

It remains to recover the rows $\bmw_i$, $i\in I$, from
$\bmv_i$ and $h_i$. Choose $i_0\in[k]\setminus U$, which exists
because $|U|\le t<k$, and let
$\bmz=\bmy_{i_0}\mid_{[m]}=\bmw_{i_0}$. For each $i\in I$, define
\begin{equation*}
 \cA_i=
 \begin{cases}
  \cA_{\le}(\bmz),&i<i_0,\\
  \cA_{\ge}(\bmz),&i>i_0.
 \end{cases}
\end{equation*}
Every column of $\bmw$ is nondecreasing, and each substring
$\bmw\mid_{I_a}$ contains $\sigma_0\sigma_{q-1}$ because
$\bmw\in\cR_{q,k,m}$, where $I_a$ is defined in \Cref{eq_partitioninterval}. Thus $\bmw_i\in\cA_i$ for every $i\in I$.
By \Cref{lem_rowdecoder}, the known set $\cA_i$, the subsequence
$\bmv_i\in\cD_1(\bmw_i)$, and the value
$h_i=\mathsf H_m(\bmw_i)$ uniquely determine $\bmw_i$. Specifically,
\begin{equation*}
 \{\bmu\in\cA_i:
       \bmv_i\in\cD_1(\bmu),\ \mathsf H_m(\bmu)=h_i\}
 =\{\bmw_i\}.
\end{equation*}
All rows of $\bmw$ are therefore recovered. This proves the required deletion-correction property.

Finally, the first $m$ columns of $\varphi_1^D(\bmx)$ equal
$\mathsf{E}_m(\bmx)$. The injectivity of $\mathsf{E}_m$ implies that
$\varphi_1^D$ is injective and $|\cC_1^D|=Q_{q,k}^{m-1}$.
Consequently, the redundancy is
\begin{equation*}
 \log_{Q_{q,k}}\frac{Q_{q,k}^n}{|\cC_1^D|}
 =n-(m-1)
 =1+2t+t\ceilenv{\log_{Q_{q,k}}p}=t\log_{Q_{q,k}}\ln(n)+O(1).
\end{equation*}
The last equality follows from the fact that $p=\Theta\parenv{P_m}=\Theta\parenv{\ln(m)}$ ( Bertrand's postulate) and $m=\Theta(n)$. This completes the proof.
\end{IEEEproof}

\subsection{\texorpdfstring{$t$-$(1,\ldots,1)$}{t-(1,...,1)}-CDCCs with \texorpdfstring{$t=k$}{t=k}}\label{sec_t=k}
%%%%%%%%%%%%%%%%%%%%%%%%%%%%%%%%%%%%%%%%%%%%%%%%%%%%
In this case, every row of the codeword may suffer a deletion. Thus, \Cref{construction_t1cdcc} does not apply. Our idea is to recover the first row of the codeword and then apply \Cref{construction_t1cdcc}. To that end, we need a code that corrects a single deletion at an arbitrary position.

For a positive integer $m$ and $\bmx\in\Sigma_q^m$, define
\begin{equation*}
 \mathsf{G}_m(\bmx)=\vt(\psi(\bmx))\bmod(qm),
\end{equation*}
where $\psi$ is given by \eqref{eq_differentialsequence}.
\begin{lemma}(\cite[Theorem~3]{Tuan2024IT})\label{lem:global}
If $\bmx,\bmx^{\prime}\in\Sigma_q^m$ satisfy
$\mathsf{G}_m(\bmx)=\mathsf{G}_m(\bmx^{\prime})$ and
$\cD_1(\bmx)\cap\cD_1(\bmx^{\prime})\ne\emptyset$, then
$\bmx=\bmx^{\prime}$.
\end{lemma}

\begin{construction}\label{construction_all1cdcc}
Let $q\ge2$, $k\ge2$, and $m$ be integers. Let $B_m$ and $P_m$
be defined in \eqref{eq_blockparameters} and
\eqref{eq_localsyndrome}, respectively, and assume $m\ge P_m$.
Let $p$ be the smallest prime greater than
$\max\{q(q+1)(P_m+1),2k-1\}$. Put
\begin{equation*}
 \Delta_1=\ceilenv{\log_{Q_{q,k}}p},\quad
 n_1=m+2(k-1)+(k-1)\Delta_1,\quad
 \Delta_2=\ceilenv{\log_{Q_{q,k}}(qn_1)},\quad
 n=n_1+2+\Delta_2.
\end{equation*}
Let $\bbM=(m_{j,i})_{(k-1)\times k}$ be the matrix over $\F_p$ in \Cref{def_cauchymatrix}
with $t=k-1$. Let $\mathsf{E}_m:\Phi_{q,k}^{m-1}\to\cR_{q,k,m}$ be the
encoder in Appendix~\ref{sec_restrictedencoder}.

For $\bmx\in\Phi_{q,k}^{m-1}$, denote $\bmw=\mathsf{E}_m(\bmx)$. For every $j\in[k-1]$, define $A_j(\bmx)=\sum_{i=1}^k m_{j,i}\mathsf H_m(\bmw_i)\in\F_p$ and $\bmd_j(\bmx)=\iota_{q,k}^{-1}\parenv{\expan{Q_{q,k},\Delta_1}{A_j(\bmx)}}$, where $A_j(\bmx)$ is represented by its integer in $[[p]]$.
With $\bm0=[0,\ldots,0]^\T$ and $\bm1=[1,\ldots,1]^\T$, set
\begin{equation}\label{eq_all1intermediate}
 \bmz=\bmw\,\bm0\,\bm1\,\bmd_1(\bmx)\cdots
       \bm0\,\bm1\,\bmd_{k-1}(\bmx)\in\Phi_{q,k}^{n_1}.
\end{equation}
Compute
\begin{equation*}
 a(\bmx)=\mathsf{G}_{n_1}(\bmz_1),\quad
 \bmb(\bmx)=\iota_{q,k}^{-1}
       \parenv{\expan{Q_{q,k},\Delta_2}{a(\bmx)}}.
\end{equation*}
Define the map $\varphi_2^D:\Phi_{q,k}^{m-1}\to\Phi_{q,k}^n$
by
\begin{equation}\label{eq_all1encoding}
 \varphi_2^D(\bmx)=\bmz\,\bm0\,\bm1\,\bmb(\bmx),
\end{equation}
and set
$\cC_2^D=\mathset{\varphi_2^D(\bmx):\bmx\in\Phi_{q,k}^{m-1}}$.
\end{construction}

\begin{theorem}\label{thm_all1cdcc}
The code $\cC_2^D$ is a $k$-resolution $(1,\ldots,1)$-CDCC
with redundancy
\begin{equation*}
 \ceilenv{\log_{Q_{q,k}}(qn_1)}
       +(k-1)\ceilenv{\log_{Q_{q,k}}p}+2k+1
 =\log_{Q_{q,k}}n+(k-1)\log_{Q_{q,k}}\ln n+O(1).
\end{equation*}
\end{theorem}
\begin{IEEEproof}
Let $\bmc=\varphi_2^D(\bmx)$ be the transmitted codeword, and let $\bmy_i$ be the received sequence in row $i\in[k]$. We show how to recover $\bmx$ from $\bmy_1,\ldots,\bmy_k$.

Let $\bmz$ be the intermediate codeword in
\eqref{eq_all1intermediate}. It is a codeword of \Cref{construction_t1cdcc} with $t=k-1$ and length $n_1$. By \Cref{thm_t1cdcc}, the intermediate code corrects at most one deletion in each of at most $k-1$ rows of $\bmz$.

Put $U=\{i\in[k]:|\bmy_i|=n-1\}$. For $i\in U$, let $d_i\in[n]$ be the deleted coordinate. By \eqref{eq_all1encoding},
$\bmc_i[n_1+1]=0$ and $\bmc_i[n_1+2]=1$, and hence
\begin{equation*}
 \bmy_i[n_1+1]=
 \begin{cases}
  \bmc_i[n_1+2]=1,&d_i\le n_1+1,\\
  \bmc_i[n_1+1]=0,&d_i>n_1+1.
 \end{cases}
\end{equation*}
Therefore, the sets $I=\{i\in U:\bmy_i[n_1+1]=1\}$ and $J=\{i\in U:\bmy_i[n_1+1]=0\}$ are known, and $U=I\mathbin{\dot\cup}J$. For every $i\in[k]$, form
\begin{equation}\label{eq_all1receivedprefix}
 \bmv_i=
 \begin{cases}
  \bmy_i\mid_{[n_1-1]},&i\in I,\\
  \bmy_i\mid_{[n_1]},&i\notin I.
 \end{cases}
\end{equation}
If $i\notin I$, then $\bmv_i=\bmz_i$. If $i\in I$, then $\bmv_i=\bmz_i\mid_{[n_1]\setminus\{\min\{d_i,n_1\}\}}\in\cD_1(\bmz_i)$.

If $J\ne\emptyset$, then
$|I|=|U|-|J|\le k-1$. Consequently,
$(\bmv_1,\ldots,\bmv_k)$ is obtained from $\bmz$ by at most
one deletion in each of at most $k-1$ rows. We can apply \Cref{thm_t1cdcc} to recover $\bmx$.

Suppose now that $J=\emptyset$. In this case, it might be that each row of $\bmz$ suffers one deletion. To apply \Cref{thm_t1cdcc}, we first recover $\bmz_1$, the first row of $\bmz$. For $i\in[k]$, set $\varepsilon_i=1$ if $i\in I$, and set $\varepsilon_i=0$ otherwise. Since $J=\emptyset$, every deletion occurs at a coordinate at most $n_1+1$. This implies that the last $\Delta_2$ symbols of each row are known up to the determined shift $\varepsilon_i$. More precisely, $\parenv{\bmb(\bmx)}_i
 =\bmy_i\mid_{[n_1+3-\varepsilon_i,\,n-\varepsilon_i]}$.
These rows determine $\bmb(\bmx)$ and hence
\begin{equation*}
 \mathsf{G}_{n_1}(\bmz_1)=a(\bmx)=\sum_{h=1}^{\Delta_2}
       \iota_{q,k}(\bmb(\bmx)[h])Q_{q,k}^{h-1}.
\end{equation*}
If $1\notin I$, then $\bmz_1=\bmv_1$ is already known.
Otherwise, \Cref{lem:global} gives
\begin{equation*}
 \mathset{\ins_{h,\gamma}(\bmv_1):
    h\in[n_1],\ \gamma\in\Sigma_q,\
    \mathsf{G}_{n_1}(\ins_{h,\gamma}(\bmv_1))=a(\bmx)}
 =\{\bmz_1\}.
\end{equation*}
In either case, the entire row $\bmz_1$ is recovered. The tuple
$(\bmz_1,\bmv_2,\ldots,\bmv_k)$ now has one deletion in each of at most
$k-1$ rows. Now we can apply \Cref{thm_t1cdcc} to recover $\bmx$. This proves the required
deletion-correction property.

The first $m$ columns of $\varphi_2^D(\bmx)$ equal
$\mathsf{E}_m(\bmx)$. Since $\mathsf{E}_m$ is injective,
$\varphi_2^D$ is also injective and $|\cC_2^D|=Q_{q,k}^{m-1}$.
Its redundancy is therefore
\begin{equation*}
 \log_{Q_{q,k}}\frac{Q_{q,k}^n}{|\cC_2^D|}
 =n-(m-1)
 =\Delta_2+(k-1)\Delta_1+2k+1.
\end{equation*}
For fixed $q,k$, we have $P_m=\Theta(\ln m)$, and Bertrand's
postulate gives $p=\Theta(\ln m)$. Hence
$\Delta_1=\log_{Q_{q,k}}\ln m+O(1)$,
$n_1=m+O(\log\log m)$, and
$\Delta_2=\log_{Q_{q,k}}m+O(1)$. Since $n=m+O(\log m)$, these estimates prove the stated asymptotic redundancy.
\end{IEEEproof}

\section{Conclusion}\label{sec_conclusion}
%%%%%%%%%%%%%%%%%%%%%%%%%%%%%%%%%%%%%%%%%%%%%%%%
We studied composite-error correcting codes and composite-deletion correcting codes for the $q$-ary $k$-resolution ordered composite DNA channel.
We introduced two additional models that prescribe error budgets for selected channels whose indices are not specified in advance. We established upper bounds on code sizes and gave explicit constructions with encoding and decoding algorithms. The redundancies of these constructions have optimal leading terms.

Several problems remain open. For
$(1,\ldots,1)$-CDCCs, it remains to close the gap between the redundancy of our construction and the lower bound while retaining efficient encoding and decoding. More generally, constructing codes that approach the derived bounds when multiple substitutions or deletions may occur in an affected channel is an important direction for further research. For DNA storage with $q=4$, it is also useful to impose GC-content and runlength constraints on each row while preserving the required error-correction properties.

\appendices
\section{Code Equivalence}\label{sec_equiv}
%%%%%%%%%%%%%%%%%%%%%%%%%%%%%%%%%%%%%%%
This section studies equivalence relations among families of CECCs and among families of CDCCs.
Let $n\ge1$ be the code length. For each tuple
$\parenv{e_1,e_2,\ldots,e_k}\in[0,n]^k$, let $\mathscr{C}_{q,k}^S\parenv{n;e_1,e_2,\ldots,e_{k}}$ (resp. $\mathscr{C}_{q,k}^D\parenv{n;e_1,e_2,\ldots,e_{k}}$)
denote the family of all $\parenv{e_1,e_2,\ldots,e_{k}}$-CECCs (resp. CDCCs) in $\Phi_{q,k}^n$.

We say that two families, $\mathscr{C}_{q,k}^S\parenv{n;e_1,e_2,\ldots,e_{k}}$ and $\mathscr{C}_{q,k}^S\parenv{n;e_1^\prime,e_2^\prime,\ldots,e_{k}^\prime}$, are \textbf{equivalent} if there is a bijection $g:\Phi_{q,k}^n\rightarrow\Phi_{q,k}^n$ satisfying the following two properties:
\begin{itemize}
    \item For every code $\cC\in\mathscr{C}_{q,k}^S\parenv{n;e_1,e_2,\ldots,e_{k}}$, its image $g(\cC)\triangleq\mathset{g(\bmc):\bmc\in\cC}$ belongs to $\mathscr{C}_{q,k}^S\parenv{n;e_1^\prime,e_2^\prime,\ldots,e_{k}^\prime}$.
    \item For every code $\cC\in\mathscr{C}_{q,k}^S\parenv{n;e_1^\prime,e_2^\prime,\ldots,e_{k}^\prime}$, its preimage $g^{-1}\parenv{\cC}$ belongs to $\mathscr{C}_{q,k}^S\parenv{n;e_1,e_2,\ldots,e_{k}}$.
\end{itemize}

Equivalence for the families $\mathscr{C}_{q,k}^D\parenv{n;e_1,e_2,\ldots,e_{k}}$ and $\mathscr{C}_{q,k}^D\parenv{n;e_1^\prime,e_2^\prime,\ldots,e_{k}^\prime}$ is defined similarly.

We begin with $\parenv{e_1,e_2,\ldots,e_{k}}$-CECCs. \Cref{pro_subreverse,pro_subequal} below generalize \cite[Proposition 3]{BesartDollma202509} and \cite[Proposition 4]{BesartDollma202509} to general alphabets. While \Cref{pro_subequal} generalizes \cite[Proposition 4]{BesartDollma202509}, its proof is novel.
\begin{proposition}\label{pro_subreverse}
    For any $\parenv{e_1,e_2,\ldots,e_{k}}\in[0,n]^k$, the families $\mathscr{C}_{q,k}^S\parenv{n;e_1,e_2,\ldots,e_{k}}$ and $\mathscr{C}_{q,k}^S\parenv{n;e_{k},e_{k-1},\ldots,e_1}$ are equivalent.
\end{proposition}
\begin{IEEEproof}
    For any $\bms=\sparenv{s_1,s_2,\ldots,s_{k}}^\T\in\Phi_{q,k}$, define $g(\bms)=\sparenv{q-1-s_{k},\ldots,q-1-s_2,q-1-s_1}^\T$. Clearly, $g$ is a bijection on $\Phi_{q,k}$. Maps $g$ and $g^{-1}$ can be naturally extended component-wise to bijections on $\Phi_{q,k}^n$. By abuse of notation, we still denote them by $g$ and $g^{-1}$, respectively.

    For a given $\cC\in\mathscr{C}_{q,k}^S\parenv{n;e_1,e_2,\ldots,e_{k}}$, let $\cC^\prime=g(\cC)$. We will show that $\cC^\prime$ is an $\parenv{e_{k},e_{k-1},\ldots,e_1}$-CECC. Suppose that the transmitted codeword is $\bmc^\prime=g(\bmc)$ and $\bmy^\prime_i$, the output of the $i$-th channel, satisfies $\dth{\bmy^\prime_i,\bmc_i^\prime}\le e_{k+1-i}$. To decode $\bmc^\prime$ from $\bmy_1^\prime,\bmy_2^\prime,\ldots,\bmy_{k}^\prime$, it suffices to recover $\bmc$. For $i\in[k]$, define $\bmy_i=q-1-\bmy_{k+1-i}^\prime$. One can verify that $\dth{\bmy_i,\bmc_i}\le e_{i}$ for all $i\in[k]$. Since $\cC$ is an $\parenv{e_1,e_2,\ldots,e_{k}}$-CECC, we can decode $\bmc$ from $\bmy_1,\bmy_2,\ldots,\bmy_{k}$. Thus, $\cC^\prime$ is indeed an $\parenv{e_{k},e_{k-1},\ldots,e_{1}}$-CECC.

    As $g^{-1}=g$, it also holds that $g^{-1}\parenv{\cC}\in\mathscr{C}_{q,k}^S\parenv{n;e_1,e_2,\ldots,e_{k}}$ for any $\cC\in\mathscr{C}_{q,k}^S\parenv{n;e_{k},e_{k-1},\ldots,e_1}$. This completes the proof.
\end{IEEEproof}

Throughout this section, let $\bme_i\triangleq\parenv{0,\ldots,0,1,0,\ldots,0}$ denote the length-$k$ vector whose $i$-th entry is $1$ and whose other entries are $0$.
\begin{proposition}\label{pro_subequal}
    For any $i,j\in[k]$ and any $s\in[n]$, the families $\mathscr{C}_{q,k}^S\parenv{n;s\cdot\bme_i}$ and $\mathscr{C}_{q,k}^S\parenv{n;s\cdot\bme_j}$ are equivalent.
\end{proposition}
\begin{IEEEproof}
   It suffices to prove the case $s=1$, as the general case follows similarly. Moreover, by transitivity and \Cref{pro_subreverse}, we only need to show that for any $1\le i\le k-1$, the families $\mathscr{C}_{q,k}^S\parenv{n;\bme_i}$ and $\mathscr{C}_{q,k}^S\parenv{n;\bme_{i+1}}$ are equivalent.

    For any $\bms=[s_1,s_2,\ldots, s_{k}]^\T\in \Phi_{q,k}$, define $g(\bms)=[q-1-s_{k},s_1+q-1-s_{k},s_2+q-1-s_{k},\ldots, s_{k-1}+q-1-s_{k}]^\T$. Since $[s_1,s_2,\ldots, s_{k}]^{\T}$ is nondecreasing and $s_{k}\le q-1$, its image $g(\bms)$ is also a nondecreasing vector in $\Sigma_q^k$ and thus belongs to $\Phi_{q,k}$. The map $g$ is a bijection, with its inverse given by
    \begin{equation}\label{eq_inversesubg}
        g^{-1}\parenv{\sparenv{s_1,s_2,\ldots,s_{k}}^\T}=\sparenv{s_2-s_1,\ldots,s_{k}-s_1,q-1-s_{1}}^\T
    \end{equation}
    for any $\sparenv{s_1,s_2,\ldots,s_{k}}^\T\in\Phi_{q,k}$. Both $g$ and $g^{-1}$ extend component-wise to bijections on $\Phi_{q,k}^n$.

    Let $\cC\subseteq\Phi_{q,k}^n$ be an $\bme_i$-CECC and $\cC^\prime=g(\cC)$. Next, we prove that $\cC^\prime$ is an $\bme_{i+1}$-CECC. Suppose that $\bmc^\prime\in\cC^\prime$ is transmitted and the $(i+1)$-th channel introduces a substitution, yielding output $\bmy^\prime$. We need to decode $\bmc^\prime$ from $\bmy^\prime$. Let $\bmc=g^{-1}\parenv{\bmc^\prime}$. It suffices to recover $\bmc$ from $\bmy^\prime$.

    Since $i+1\ge2$, the first row of $\bmy^\prime$ is error-free. Thus, it follows from \eqref{eq_inversesubg} that if $\bmy^\prime$ is a sequence in $\Phi_{q,k}^n$, then $g^{-1}\parenv{\bmy^\prime}$ is obtained from $\bmc$ by a substitution error in the $i$-th row. As $\cC$ is an $\bme_i$-CECC, we can decode $\bmc$ from $g^{-1}\parenv{\bmy^\prime}$.

    Now suppose that $\bmy^\prime$ is not a sequence in $\Phi_{q,k}^n$. Then there is exactly one index $m\in[n]$ such that $\bmy^\prime[m]$ is not a nondecreasing vector. This implies that the substitution error occurred in the $m$-th column. Construct a sequence $\bmy$ as follows:
    \begin{itemize}
        \item For $j\in[n]\setminus\{m\}$, let $\bmy[j]=g^{-1}\parenv{\bmy^\prime[j]}$.
        \item For the column $\bmy[m]$, set $\bmy_{\ell}[m]=\bmy^\prime_{\ell+1}[m]-\bmy^\prime_1[m]$ for $\ell\in\mathset{1,2,\ldots,k-1}\setminus\{i\}$; set $\bmy_{k}[m]=q-1-\bmy^\prime_1[m]$; let $\bmy_i[m]$ be any symbol in $\Sigma_q$.
    \end{itemize}
    Since $\bmy^\prime_1[m]=\bmc^\prime_1[m]$, \Cref{eq_inversesubg} implies $\bmy_{\ell}[m]=\bmc_{\ell}[m]$ for all $\ell\ne i$. It is clear that $\bmy$ is obtained from $\bmc$ by at most one substitution in the $i$-th row. As $\cC$ is an $\bme_i$-CECC, we can recover $\bmc$ from $\bmy$.

    In both cases, $\bmc$ can be decoded. Therefore, $\cC^\prime$ is an $\bme_{i+1}$-CECC. A symmetric argument using $g^{-1}$ shows that $g^{-1}(\cC)\in\mathscr{C}_{q,k}^S\parenv{n;\bme_{i}}$ for any $\cC\in\mathscr{C}_{q,k}^S\parenv{n;\bme_{i+1}}$, establishing the equivalence.
\end{IEEEproof}

We now turn to $\parenv{e_1,e_2,\ldots,e_{k}}$-CDCCs. The following proposition, analogous to \Cref{pro_subreverse}, has a similar proof.

\begin{proposition}\label{prop_delreverse}
   The two classes of codes $\mathscr{C}_{q,k}^D\parenv{n;e_1,e_2,\ldots,e_{k}}$ and $\mathscr{C}_{q,k}^D\parenv{n;e_{k},e_{k-1},\ldots,e_1}$ are equivalent for any $\parenv{e_1,e_2,\ldots,e_{k}}\in[0,n]^k$.
\end{proposition}

Recall that $A_{q,k}^S\parenv{n;\parenv{e_1,e_2,\ldots,e_{k}}}$ and $A_{q,k}^D\parenv{n;\parenv{e_1,e_2,\ldots,e_{k}}}$ are the maximum possible sizes of an
$\parenv{e_1,\ldots,e_k}$-CECC and an
$\parenv{e_1,\ldots,e_k}$-CDCC, respectively. As an immediate consequence of \Cref{pro_subreverse,pro_subequal,prop_delreverse}, we have the following result.
\begin{corollary}
    For any $\parenv{e_1,e_2,\ldots,e_{k}}\in[0,n]^k$, any $s\in[n]$ and any $i,j\in[k]$, we have
    \begin{enumerate}[$(i)$]
        \item $A_{q,k}^S\parenv{n;\parenv{e_1,e_2,\ldots,e_{k}}}=A_{q,k}^S\parenv{n;\parenv{e_{k},e_{k-1},\ldots,e_{1}}}$ and $A_{q,k}^S\parenv{n;s\cdot\bme_i}=A_{q,k}^S\parenv{n;s\cdot\bme_j}$;
        \item $A_{q,k}^D\parenv{n;\parenv{e_1,e_2,\ldots,e_{k}}}=A_{q,k}^D\parenv{n;\parenv{e_{k},e_{k-1},\ldots,e_{1}}}$.
    \end{enumerate}
\end{corollary}

\section{Encoding and Decoding for \texorpdfstring{$\cR_{q,k,m}$}{R(q,k,m)}}\label{sec_restrictedencoder}
%%%%%%%%%%%%%%%%%%%%%%%%%%%%%%%%%%%%%%%%%%%%%%%%%
Let $B_m$, $N_m$, the intervals $I_i$, and the ordered composite letters $\sigma_0,\sigma_{q-1}$ be as defined in \Cref{sec_localization}. Throughout this section, assume $m\ge2B_m$ and write $Q=Q_{q,k}$.

Recall that $I_i=[(i-1)B_m+1,iB_m]$. A sequence
$\bmw\in\Phi_{q,k}^m$ belongs to $\cR_{q,k,m}$ if and only if, for every $i\in[\floorenv{m/B_m}]$, there is an index $h\in[(i-1)B_m+1,iB_m-1]$ such that $\bmw[h]=\sigma_0$ and $\bmw[h+1]=\sigma_{q-1}$.

This section constructs an encoder $\mathsf{E}_m:\Phi_{q,k}^{m-1}\to\cR_{q,k,m}$ and a
corresponding decoder $\mathsf{D}_m$ that recovers the message from its encoded sequence. The idea is similar to that of
\Cref{alg_encoder10CECC,alg_decoder10CECC}.
We specify an order for each of the sets
$\Phi_{q,k}^{m-1}$ and $\cR_{q,k,m}$, and number the sequences
in each set starting from zero.
For each $0\le i<Q^{m-1}$, the encoder $\mathsf{E}_m$ maps
the sequence numbered $i$ in $\Phi_{q,k}^{m-1}$ to the
sequence numbered $i$ in $\cR_{q,k,m}$.
By \Cref{lem_restrictionsize} and $Q\ge2$, we have
$|\cR_{q,k,m}|\ge Q^{m-1}$, so the latter sequence always
exists. Given an encoded sequence, the decoder
$\mathsf{D}_m$ computes its rank in $\cR_{q,k,m}$ and
recovers the message assigned that number.

We first show how to count the sequences in $\cR_{q,k,m}$
with a prescribed prefix. These counts will be used both
to compute the rank of a sequence in $\cR_{q,k,m}$ and
to find the sequence with a given rank.
We then present the encoding and decoding algorithms.
\subsection{Counting Sequences with a Prescribed Prefix}
%%%%%%%%%%%%%%%%%%%%%%%%%%%%%%%%%%%%%%%%%%%%%%%%%
Let $\bm{\epsilon}$ denote the empty sequence. For
$0\le j\le m$ and $\bmu\in\Phi_{q,k}^j$, define
\begin{equation}\label{eq_restricted-prefixcount}
 K(\bmu)=\abs{\mathset{\bmv\in\Phi_{q,k}^{m-j}:
                         \bmu\bmv\in\cR_{q,k,m}}}.
\end{equation}
Thus, $K(\bmu)$ counts the sequences in $\cR_{q,k,m}$ whose
length-$j$ prefix equals $\bmu$. In particular,
$K(\bm{\epsilon})=|\cR_{q,k,m}|$. If $j=m$, then
$K(\bmu)=1$ for $\bmu\in\cR_{q,k,m}$ and $K(\bmu)=0$
otherwise. If $j<m$, separating the sequences according to their $(j+1)$-st composite symbol gives
\begin{equation}\label{eq_restrictedcountsum}
 K(\bmu)=\sum_{\tau\in\Phi_{q,k}}K(\bmu\tau).
\end{equation}

We will show how to compute these counts without listing all possible prefixes. If, for some $i$ with $iB_m\le j$, the substring $\bmu\mid_{I_i}$ does not contain $\sigma_0\sigma_{q-1}$, then $K(\bmu)=0$. Indeed, adding symbols at
positions $j+1,\ldots,m$ cannot change $\bmu\mid_{I_i}$. Suppose henceforth that
\begin{equation}\label{eq_prefixcondition}
 \text{for every }i\in[\floorenv{j/B_m}],\quad
 \bmu\mid_{I_i}\text{ contains }\sigma_0\sigma_{q-1}.
\end{equation}
The condition is vacuous when $j<B_m$.
Define $s(\bmu)\in\{0,1\}$ by
\begin{equation*}
 s(\bmu)=
 \begin{cases}
  1,&j\ge1\text{ and }\bmu[j]=\sigma_0,\\
  0,&\text{otherwise}.
 \end{cases}
\end{equation*}
To define a second integer $\xi(\bmu)\in\{0,1\}$, first suppose
$j<N_m$ and $B_m\nmid j$, and put
$i=\floorenv{j/B_m}+1$. In this case,
$(i-1)B_m<j<iB_m$, so the symbols already specified in $I_i$ are precisely those at positions $(i-1)B_m+1,\ldots,j$. Set $\xi(\bmu)=1$ if there is an integer $(i-1)B_m+1\le h<j$ such that
\begin{equation*}
 \bmu[h]=\sigma_0,\quad \bmu[h+1]=\sigma_{q-1},
\end{equation*}
and set $\xi(\bmu)=0$ otherwise. If $B_m\mid j$ or $j\ge N_m$,
set $\xi(\bmu)=0$. In particular,
$s(\bm{\epsilon})=\xi(\bm{\epsilon})=0$.

For a fixed $j$, the value of $K(\bmu)$ for $\bmu$ satisfying \eqref{eq_prefixcondition} depends only on $s(\bmu)$ and $\xi(\bmu)$. Indeed, if $j\ge N_m$, every choice of the remaining $m-j$ composite symbols is allowed. If $j<N_m$ and $B_m\mid j$, the next position is the first position of $I_{j/B_m+1}$, and all conditions on the remaining intervals involve only the composite symbols to be added. Finally, if $j<N_m$ and
$B_m\nmid j$, let $i=\floorenv{j/B_m}+1$. If $\xi(\bmu)=1$, the substring $\bmu\mid_{[(i-1)B_m+1,j]}$ already contains $\sigma_0\sigma_{q-1}$. If $\xi(\bmu)=0$, an occurrence of $\sigma_0\sigma_{q-1}$ in $I_i$ must either use positions $j,j+1$ or have both positions greater than $j$. The first possibility requires $s(\bmu)=1$ and the next letter to be $\sigma_{q-1}$. The second possibility involves only the composite symbols to be added. The conditions on $I_{i+1},I_{i+2},\ldots$ also involve only those symbols.

We next iteratively compute a table $F_j(s,\xi)$, for $0\le j\le m$ and $s,\xi\in\{0,1\}$, such that
\begin{equation}\label{eq_tablemeaning}
 K(\bmu)=F_j\bigl(s(\bmu),\xi(\bmu)\bigr)
 \quad\text{whenever \eqref{eq_prefixcondition} holds.}
\end{equation}
For $N_m\le j\le m$, set $F_j(s,\xi)=Q^{m-j}$ for all $s,\xi\in\{0,1\}$. For each symbol $\tau\in\Phi_{q,k}$, let $s_\tau=1$ if $\tau=\sigma_0$, and let $s_\tau=0$
otherwise. For $0\le j<N_m$, define
\begin{equation*}
 d_j(s,\xi,\tau)=
 \begin{cases}
  1,&\xi=1\text{ or }
       (B_m\nmid j,\ s=1,\ \tau=\sigma_{q-1}),\\
  0,&\text{otherwise}.
 \end{cases}
\end{equation*}
Let $\bmv=\bmu\tau$, so that $\bmv[j+1]=\tau$.
Then $d_j(s(\bmu),\xi(\bmu),\tau)=1$ if and only if
there exists an integer $h$ satisfying
\begin{equation*}
 B_m\floorenv{j/B_m}+1\le h\le j,\quad
 \bmv[h]=\sigma_0,\quad
 \bmv[h+1]=\sigma_{q-1}.
\end{equation*}
Otherwise, $d_j(s(\bmu),\xi(\bmu),\tau)=0$. In particular, when $j=iB_m$ with $i\ge1$, we have $d_j(s(\bmu),\xi(\bmu),\tau)=0$.

For $j=N_m-1,N_m-2,\ldots,0$, compute
\begin{equation}\label{eq_restrictedrecurrence}
 F_j(s,\xi)=
 \begin{cases}
  \displaystyle\sum_{\tau\in\Phi_{q,k}}
       F_{j+1}\bigl(s_\tau,d_j(s,\xi,\tau)\bigr),
       &B_m\nmid(j+1),\\[6pt]
  \displaystyle\sum_{\substack{\tau\in\Phi_{q,k}\\
                          d_j(s,\xi,\tau)=1}}
       F_{j+1}(s_\tau,0),
       &B_m\mid(j+1).
 \end{cases}
\end{equation}
\begin{proposition}\label{prop_restrictedcounts}
    The $m+1$ functions $F_0,F_1,\ldots,F_m$ defined iteratively above satisfy \eqref{eq_tablemeaning} for every $0\le j\le m$.
\end{proposition}
\begin{IEEEproof}
The claim trivially holds when $N_m\le j\le m$. We now use backward induction on $j$. Fix $0\le j<N_m$,
and assume that \eqref{eq_tablemeaning} holds for every $\bmu^\prime\in\Phi_{q,k}^{j+1}$ satisfying \eqref{eq_prefixcondition} with $j$ replaced by $j+1$. Let $\bmu\in\Phi_{q,k}^j$
satisfy \eqref{eq_prefixcondition}, and put
$s=s(\bmu)$ and $\xi=\xi(\bmu)$. For
$\tau\in\Phi_{q,k}$, write $\bmv=\bmu\tau$. Since
$\bmv[j+1]=\tau$, we always have $s(\bmv)=s_\tau$.
We consider the following two cases.

\textbf{Case 1:} $B_m\nmid(j+1)$.
In this case,
$\floorenv{(j+1)/B_m}=\floorenv{j/B_m}$, so
$\bmv$ satisfies \eqref{eq_prefixcondition} with $j+1$
in place of $j$ for every $\tau\in\Phi_{q,k}$.
Also, $j+1<N_m$, since $B_m\mid N_m$. By the definitions of $\xi$ and $d_j$, we have $\xi(\bmv)=d_j(s,\xi,\tau)$.
The induction hypothesis therefore gives
$K(\bmu\tau)=F_{j+1}(s_\tau,d_j(s,\xi,\tau))$.
Using \eqref{eq_restrictedcountsum} and the first case of \eqref{eq_restrictedrecurrence}, we obtain
\begin{equation*}
 K(\bmu)
 =\sum_{\tau\in\Phi_{q,k}}K(\bmu\tau)
 =\sum_{\tau\in\Phi_{q,k}}
      F_{j+1}(s_\tau,d_j(s,\xi,\tau))
 =F_j(s,\xi).
\end{equation*}

\textbf{Case 2:} $B_m\mid(j+1)$.
Put $i=(j+1)/B_m$. Then $\floorenv{j/B_m}=i-1$,
and the condition on every interval $I_\ell$ with $1\le\ell<i$
already holds for $\bmu$ (see \eqref{eq_prefixcondition}). Thus, $\bmv$ satisfies
\eqref{eq_prefixcondition}, with $j+1$ in place of $j$, if and only if there exists an integer $h$ such that
\begin{equation*}
 (i-1)B_m+1\le h< iB_m,\quad
 \bmv[h]=\sigma_0,\quad
 \bmv[h+1]=\sigma_{q-1}.
\end{equation*}
By the definition of $d_j$, this is equivalent to
$d_j(s,\xi,\tau)=1$. If $d_j(s,\xi,\tau)=0$, the substring $\bmv\mid_{I_i}$ does not contain $\sigma_0\sigma_{q-1}$. No extension of $\bmv$ can change that substring, so $K(\bmu\tau)=0$.
If $d_j(s,\xi,\tau)=1$, then $\bmv$ satisfies \eqref{eq_prefixcondition} and $\xi(\bmv)=0$, since $B_m\mid(j+1)$. The induction hypothesis gives
$K(\bmu\tau)=F_{j+1}(s_\tau,0)$. Consequently, \eqref{eq_restrictedcountsum} and the second
case of \eqref{eq_restrictedrecurrence} yield
\begin{equation*}
 K(\bmu)
 =\sum_{\tau\in\Phi_{q,k}}K(\bmu\tau)
 =\sum_{\substack{\tau\in\Phi_{q,k}\\d_j(s,\xi,\tau)=1}}
      F_{j+1}(s_\tau,0)
 =F_j(s,\xi).
\end{equation*}
Both cases prove \eqref{eq_tablemeaning} at length $j$.
The result follows by backward induction.
\end{IEEEproof}

For use in the encoding and decoding algorithms, we record
the counts after appending one letter. Let $0\le j<m$ and
$\bmu\in\Phi_{q,k}^j$ satisfy \eqref{eq_prefixcondition}, and
put $s=s(\bmu)$ and $\xi=\xi(\bmu)$. The proof of
\Cref{prop_restrictedcounts} gives, for every
$\tau\in\Phi_{q,k}$,
\begin{equation}\label{eq_restrictedextension}
 K(\bmu\tau)=
 \begin{cases}
  Q^{m-j-1},&j\ge N_m,\\[2pt]
  F_{j+1}(s_\tau,d_j(s,\xi,\tau)),
       &j<N_m,\ B_m\nmid(j+1),\\[2pt]
  F_{j+1}(s_\tau,0),
       &j<N_m,\ B_m\mid(j+1),\ d_j(s,\xi,\tau)=1,\\[2pt]
  0,&j<N_m,\ B_m\mid(j+1),\ d_j(s,\xi,\tau)=0.
 \end{cases}
\end{equation}
If $\bmu$ does not satisfy \eqref{eq_prefixcondition}, then
$K(\bmu\tau)=0$ for every $\tau\in\Phi_{q,k}$.

\subsection{Encoding and Decoding Algorithms}
%%%%%%%%%%%%%%%%%%%%%%%%%%%%%%%%%%%%%%%%%%%%
Recall that we give a bijection $\iota_{q,k}:\Phi_{q,k}\to\Sigma_Q$ in \Cref{def_bijection}. For $a\in\Sigma_Q$, put $\tau_a=\iota_{q,k}^{-1}(a)$. Then $\Phi_{q,k}=\mathset{\tau_0,\tau_1,\ldots,\tau_{Q-1}}$. Order the symbols in $\Phi_{q,k}$ by $\tau_0<\tau_1<\cdots<\tau_{Q-1}$. In other words, for any two ordered composite symbols $\sigma,\sigma^\prime\in\Phi_{q,k}$, we have $\sigma<\sigma^\prime$ if and only if $\iota_{q,k}(\sigma)<\iota_{q,k}(\sigma^\prime)$.

For two distinct sequences 
$\bmu,\bmv\in\Phi_{q,k}^m$, write
$\bmu<_{\mathrm{lex}}\bmv$ if
$\bmu[h]<\bmv[h]$ at
$h=\min\{i\in[m]:\bmu[i]\ne\bmv[i]\}$.
For $\bmw\in\cR_{q,k,m}$, define its rank, starting at zero, by $\operatorname{rank}_{\cR_{q,k,m}}(\bmw)
=\abs{\mathset{\bmv\in\cR_{q,k,m}:\bmv<_{\mathrm{lex}}\bmw}}$. Thus, rank zero identifies the first sequence, rank one the second sequence, and so on.

For a prefix $\bmu$ of length less than $m$ and an integer $0\le a\le Q$, put
\begin{equation}\label{eq_branchcounts}
 L(\bmu,a)=\sum_{b=0}^{a-1}K(\bmu\tau_b).
\end{equation}
Here, $L(\bmu,0)\triangleq0$. By \eqref{eq_restrictedcountsum}, we have $L(\bmu,Q)=K(\bmu)$. For $0\le a<Q$, the value $L(\bmu,a)$ counts the sequences
in $\cR_{q,k,m}$ whose prefix equals $\bmu$ and whose next symbol precedes $\tau_a$. If $\bmu$ satisfies \eqref{eq_prefixcondition}, each count $K(\bmu\tau_b)$ in \eqref{eq_branchcounts} is evaluated by \eqref{eq_restrictedextension} with $j=|\bmu|$, $s=s(\bmu)$, $\xi=\xi(\bmu)$, and $\tau=\tau_b$. Otherwise, all these counts are zero.

Given $\bmw\in\cR_{q,k,m}$, let $\bmu_0=\bm{\epsilon}$,
$\bmu_j=\bmw[1]\cdots\bmw[j]$ for $j\in[m]$, and
$a_{j+1}=\iota_{q,k}(\bmw[j+1])$ for $0\le j<m$. Then
\begin{equation}\label{eq_restrictedrank}
 \operatorname{rank}_{\cR_{q,k,m}}(\bmw)
 =\sum_{j=0}^{m-1}L(\bmu_j,a_{j+1}).
\end{equation}
Indeed, every sequence $\bmv\in\cR_{q,k,m}$ with
$\bmv<_{\mathrm{lex}}\bmw$ has a unique first differing position $j+1$. It begins with $\bmu_j\tau_b$ for some $b<a_{j+1}$, and there are $K(\bmu_j\tau_b)$ such sequences. Summing over these disjoint possibilities gives \eqref{eq_restrictedrank}.

For a message $\bmx\in\Phi_{q,k}^{m-1}$, define
\begin{equation}\label{eq_messageindex}
 z(\bmx)=\sum_{h=1}^{m-1}\iota_{q,k}(\bmx[h])Q^{h-1}.
\end{equation}
This is an integer in $[[Q^{m-1}]]=\mathset{0,1,\ldots,Q^{m-1}-1}$. The message $\bmx$ is recovered by applying $\iota_{q,k}^{-1}$ to $\expan{Q,m-1}{z(\bmx)}$. By
\Cref{lem_restrictionsize} and $Q\ge2$, $ K(\bm{\epsilon})=|\cR_{q,k,m}|\ge Q^m/2\ge Q^{m-1}$.

We define $\bmw=\mathsf{E}_m(\bmx)$ to be the unique sequence in $\cR_{q,k,m}$ with rank $z(\bmx)$. \Cref{alg_restrictedencoder} constructs $\bmw$ one symbol at a time, starting with $\bmu_0=\bm{\epsilon}$ and $\rho_0=z(\bmx)$. 

At iteration $j$, the algorithm has constructed the
length-$j$ prefix $\bmu_j$ of $\bmw$. The integer $\rho_j$ denotes the number of sequences
$\bmv\in\cR_{q,k,m}$ that have prefix $\bmu_j$ and satisfy $\bmv<_{\mathrm{lex}}\bmw$.
For each $a\in[[Q]]$, among all sequences in $\cR_{q,k,m}$ with prefix $\bmu_j$, exactly $L(\bmu_j,a)$ have their $(j+1)$-st symbol in $\{\tau_b:0\le b<a\}$, and exactly $K(\bmu_j\tau_a)$ have their $(j+1)$-st symbol equal to $\tau_a$. Therefore, $\bmw[j+1]=\tau_{a_{j+1}}$, where
$a_{j+1}$ is the unique index satisfying $L(\bmu_j,a_{j+1})\le\rho_j<L(\bmu_j,a_{j+1}+1)$.
The algorithm then sets $\bmu_{j+1}=\bmu_j\tau_{a_{j+1}}$ and $\rho_{j+1}=\rho_j-L(\bmu_j,a_{j+1})$.
The integer $\rho_{j+1}$ thus counts the sequences
$\bmv\in\cR_{q,k,m}$ that have prefix $\bmu_{j+1}$ and satisfy $\bmv<_{\mathrm{lex}}\bmw$.
\begin{algorithm}[!t]
\caption{The encoder $\mathsf{E}_m$}\label{alg_restrictedencoder}
\KwIn{$\bmx\in\Phi_{q,k}^{m-1}$ and the values
$F_j(s,\xi)$, $0\le j\le m$, $s,\xi\in\{0,1\}$,
from \Cref{prop_restrictedcounts}.}
\KwOut{$\bmw=\mathsf{E}_m(\bmx)\in\cR_{q,k,m}$.}
$\bmu_0\gets\bm{\epsilon}$; $\rho_0\gets z(\bmx)$ as in
\eqref{eq_messageindex}\;
\For{$j=0,1,\ldots,m-1$}{
 Compute $K(\bmu_j\tau_a)$ for $a=0,1,\ldots,Q-1$
 from \eqref{eq_restrictedextension} with
 $s=s(\bmu_j)$, $\xi=\xi(\bmu_j)$, and $\tau=\tau_a$\;
 Compute $L(\bmu_j,a)$ for $a=0,1,\ldots,Q$
 from \eqref{eq_branchcounts}\;
 Choose the unique $a_{j+1}\in[[Q]]$ satisfying
 $L(\bmu_j,a_{j+1})\le\rho_j<L(\bmu_j,a_{j+1}+1)$\;
 $\bmu_{j+1}\gets\bmu_j\tau_{a_{j+1}}$\;
 $\rho_{j+1}\gets\rho_j-L(\bmu_j,a_{j+1})$\;
}
$\bmw\gets\bmu_m$\;
\Return{$\bmw$}\;
\end{algorithm}

We define
$\mathsf{D}_m:\Phi_{q,k}^m\to\Phi_{q,k}^{m-1}\cup\{\bot\}$,
where $\bot$ is an output symbol not in $\Phi_{q,k}^{m-1}$.
For an input $\bmw\notin\cR_{q,k,m}$, the decoder returns
$\bot$. For $\bmw\in\cR_{q,k,m}$, the decoder computes $z_m=\operatorname{rank}_{\cR_{q,k,m}}(\bmw)$ using \eqref{eq_restrictedrank}. If $z_m\ge Q^{m-1}$, the decoder returns $\bot$. If $z_m<Q^{m-1}$, the decoder returns $\iota_{q,k}^{-1}(\expan{Q,m-1}{z_m})$.
\Cref{alg_restricteddecoder} gives these computations explicitly.
\begin{algorithm}[!t]
\caption{The decoder $\mathsf{D}_m$}\label{alg_restricteddecoder}
\KwIn{$\bmw\in\Phi_{q,k}^m$ and the values
$F_j(s,\xi)$, $0\le j\le m$, $s,\xi\in\{0,1\}$,
from \Cref{prop_restrictedcounts}.}
\KwOut{The message $\bmx$ if $\bmw=\mathsf{E}_m(\bmx)$;
otherwise $\bot$.}
\For{$i=1,2,\ldots,\floorenv{m/B_m}$}{
 \If{there is no $h\in[(i-1)B_m+1,iB_m-1]$ with
 $\bmw[h]=\sigma_0$ and $\bmw[h+1]=\sigma_{q-1}$}{
  \Return{$\bot$}\;
 }
}
$\bmu_0\gets\bm{\epsilon}$; $z_0\gets0$\;
\For{$j=0,1,\ldots,m-1$}{
 $a_{j+1}\gets\iota_{q,k}(\bmw[j+1])$\;
 Compute $K(\bmu_j\tau_b)$ for $0\le b<a_{j+1}$ from
 \eqref{eq_restrictedextension} with $s=s(\bmu_j)$,
 $\xi=\xi(\bmu_j)$, and $\tau=\tau_b$\;
 Compute $L(\bmu_j,a_{j+1})$ from \eqref{eq_branchcounts}
 (the value is zero if $a_{j+1}=0$)\;
 $z_{j+1}\gets z_j+L(\bmu_j,a_{j+1})$\;
 $\bmu_{j+1}\gets\bmu_j\bmw[j+1]$\;
}
\If{$z_m\ge Q^{m-1}$}{\Return{$\bot$}\;}
$\bmx\gets\iota_{q,k}^{-1}(\expan{Q,m-1}{z_m})$\;
\Return{$\bmx$}\;
\end{algorithm}

\begin{proposition}%\label{prop_restrictedencoder}
For every $\bmx\in\Phi_{q,k}^{m-1}$, the encoder in
\Cref{alg_restrictedencoder} returns
$\mathsf{E}_m(\bmx)\in\cR_{q,k,m}$ with $\operatorname{rank}_{\cR_{q,k,m}}\parenv{\mathsf{E}_m(\bmx)}=z(\bmx)$. The decoder in \Cref{alg_restricteddecoder} satisfies $\mathsf{D}_m(\mathsf{E}_m(\bmx))=\bmx$
for every $\bmx\in\Phi_{q,k}^{m-1}$, and $\mathsf{D}_m(\bmw)=\bot$ for every
$\bmw\in\Phi_{q,k}^m\setminus\mathsf{E}_m(\Phi_{q,k}^{m-1})$.
\end{proposition}
\begin{IEEEproof}
We first prove by induction on $j$ that
\Cref{alg_restrictedencoder} determines $\bmu_j$ and
$\rho_j$ uniquely and that $0\le\rho_j<K(\bmu_j)$ for every $0\le j\le m$. For $j=0$, the initialization gives $\bmu_0=\bm{\epsilon}$ and $\rho_0=z(\bmx)$. By \Cref{lem_restrictionsize},
$K(\bm{\epsilon})=|\cR_{q,k,m}|\ge Q^m/2$. Since $Q\ge2$, we have $0\le\rho_0<Q^{m-1}\le K(\bm{\epsilon})$.

Suppose that $\bmu_j$ and $\rho_j$ are uniquely determined and satisfy $0\le\rho_j<K(\bmu_j)$ for some $j<m$. The inequality $K(\bmu_j)>0$ implies that there is a sequence $\bmv\in\Phi_{q,k}^{m-j}$ with $\bmu_j\bmv\in\cR_{q,k,m}$. Therefore, $\bmu_j$ satisfies \eqref{eq_prefixcondition}. Then every $K(\bmu_j\tau_a)$ can be computed from \eqref{eq_restrictedextension}. By \eqref{eq_branchcounts},
\begin{equation*}
 \begin{aligned}
  L(\bmu_j,0)&=0,\qquad L(\bmu_j,Q)=K(\bmu_j),\\
  L(\bmu_j,a+1)-L(\bmu_j,a)
     &=K(\bmu_j\tau_a)\ge0\qquad(a\in[[Q]]).
 \end{aligned}
\end{equation*}
Hence the integer intervals $[L(\bmu_j,a),L(\bmu_j,a+1)-1]$ with $a\in[[Q]]$ and $ K(\bmu_j\tau_a)>0$ are pairwise disjoint, and their union is $[0,K(\bmu_j)-1]$.
Exactly one of these intervals contains $\rho_j$, so
$a_{j+1}$ exists and is unique. The assignments of
$\bmu_{j+1}$ and $\rho_{j+1}$ then give $0\le\rho_{j+1}=\rho_j-L(\bmu_j,a_{j+1})<L(\bmu_j,a_{j+1}+1)-L(\bmu_j,a_{j+1})=K(\bmu_{j+1})$. This completes the induction.

At $j=m$, \eqref{eq_restricted-prefixcount} gives
$K(\bmu_m)=1$ if $\bmu_m\in\cR_{q,k,m}$, and
$K(\bmu_m)=0$ otherwise. Since $K(\bmu_m)>0$, it follows that $\bmu_m\in\cR_{q,k,m}$ and $\rho_m=0$.
For the returned sequence $\bmw=\bmu_m$,
\eqref{eq_restrictedrank} and
the equality $\rho_j-\rho_{j+1}=L(\bmu_j,a_{j+1})$ yield
\begin{equation}\label{eq_rankEm}
 \operatorname{rank}_{\cR_{q,k,m}}(\bmw)=\sum_{j=0}^{m-1}L(\bmu_j,a_{j+1})=\sum_{j=0}^{m-1}(\rho_j-\rho_{j+1})
 =\rho_0-\rho_m=z(\bmx).
\end{equation}
By the definition of $\mathsf{E}_m$, this proves that \Cref{alg_restrictedencoder} returns $\mathsf{E}_m(\bmx)$.

We next verify \Cref{alg_restricteddecoder}. If
$\bmw\notin\cR_{q,k,m}$, there is an interval $I_i$ for which no $h\in[(i-1)B_m+1,iB_m-1]$ satisfies
$\bmw[h]=\sigma_0$ and $\bmw[h+1]=\sigma_{q-1}$.
The first loop therefore returns $\bot$. Suppose now that $\bmw\in\cR_{q,k,m}$. For each $j$, its prefix $\bmu_j$ satisfies \eqref{eq_prefixcondition}, so \eqref{eq_restrictedextension} gives the values $K(\bmu_j\tau_b)$ used by the decoder. Starting from $z_0=0$, the assignments of $z_{j+1}$ give
\begin{equation*}
 z_m=\sum_{j=0}^{m-1}L(\bmu_j,a_{j+1})
     =\operatorname{rank}_{\cR_{q,k,m}}(\bmw),
\end{equation*}
where the second equality is \eqref{eq_restrictedrank}.

For every $\bmx\in\Phi_{q,k}^{m-1}$, we have $\operatorname{rank}_{\cR_{q,k,m}}(\mathsf{E}_m(\bmx))
 =z(\bmx)<Q^{m-1}$. Thus, if $z_m\ge Q^{m-1}$, then $\bmw$ cannot equal
$\mathsf{E}_m(\bmx)$ for any $\bmx\in\Phi_{q,k}^{m-1}$,
and the decoder returns $\bot$.
If $z_m<Q^{m-1}$, put
$\bmx=\iota_{q,k}^{-1}(\expan{Q,m-1}{z_m})$.
Then \eqref{eq_messageindex} gives $z(\bmx)=z_m$.
By \eqref{eq_rankEm}, we have $\operatorname{rank}_{\cR_{q,k,m}}(\mathsf{E}_m(\bmx))=z(\bmx)=z_m=\operatorname{rank}_{\cR_{q,k,m}}(\bmw)$. Distinct sequences in $\cR_{q,k,m}$ have distinct ranks. Therefore, $\mathsf{E}_m(\bmx)=\bmw$. Thus, the decoder returns a sequence in $\Phi_{q,k}^{m-1}$ exactly when $\bmw\in\mathsf{E}_m\parenv{\Phi_{q,k}^{m-1}}$.
In particular, if $\bmw=\mathsf{E}_m(\bmx)$, then
$z_m=z(\bmx)<Q^{m-1}$. By \eqref{eq_messageindex}, the output of the decoder is $ \mathsf{D}_m(\mathsf{E}_m(\bmx))
 =\iota_{q,k}^{-1}\bigl(\expan{Q,m-1}{z(\bmx)}\bigr)
 =\bmx$.
This completes the proof.
\end{IEEEproof}

\bibliographystyle{IEEEtran}
\bibliography{ref}

@article{Agrawal2004Annals,
  author  = {Manindra Agrawal and Neeraj Kayal and Nitin Saxena},
  title   = {{PRIMES} is in {P}},
  journal = {Ann. Math.},
  volume  = {160},
  number  = {2},
  pages   = {781--793},
  year    = {2004},
  month={Sept.}
}

@article{Agrawal2019Annals,
  author  = {Manindra Agrawal and Neeraj Kayal and Nitin Saxena},
  title   = {Errata: {PRIMES} is in {P}},
  journal = {Ann. Math.},
  volume  = {189},
  number  = {1},
  pages   = {317--318},
  year    = {2019},
  month={Jan.}
}

@article{science2012,
  title={Next-{G}eneration {D}igital {I}nformation {S}torage in {DNA}},
  author={{Church}, {George} M and {Gao}, {Yuan} and {Kosuri}, {Sriram}},
  journal={Science},
  volume={337},
  number={6102},
  pages={1628--1628},
  year={2012},
  Month={Sep.}
}

@article{nature2013,
  title={Towards {P}ractical, {H}igh-capacity, {L}ow-{M}aintenance {I}nformation {S}torage in {S}ynthesized {DNA}},
  author={{Goldman}, {Nick} and {Bertone}, {Paul} and {Chen}, {Siyuan} and {Dessimoz}, {Christophe} and {LeProust}, {Emily} M and {Sipos}, {Botond} and {Birney}, {Ewan}},
  journal={Nature},
  volume={494},
  number={7435},
  pages={77--80},
  year={2013},
  Month={Feb.}
}

@article{Kulkarni013IT,
  author = {Ankur A. Kulkarni and Negar Kiyavash},
  title = {Nonasymptotic {U}pper {B}ounds for {D}eletion {C}orrecting {C}odes},
  journal = {IEEE Trans. Inf. Theory},
  volume = {59}, 
  number = {8}, 
  pages = {5115--5130},
  year = {2013}, 
  month ={Aug.}
}

@ARTICLE{Yazdi2015TMBMC,
  author={S. M. Hossein Tabatabaei Yazdi and Han Mao Kiah and Eva Garcia-Ruiz and Jian Ma and Huimin Zhao and Olgica Milenkovic},
  journal={IEEE Trans. Mol. Biol. Multiscale Commun.}, 
  title={{DNA}-{B}ased {S}torage: {T}rends and {M}ethods}, 
  year={2015},
 month={Sep.},
  volume={1},
  number={3},
  pages={230--248}
  }

@article{YANIV2017Science,
author={YANIV Erlich and DINA Zielinski},
title={{DNA} {F}ountain enables a robust and efficient storage architecture},
journal={Science},
year={2017},
month={Mar.},
volume={355},
number={6328},
pages={950--954}
}

@ARTICLE{Schalkwijk1972IT,
  author={Schalkwijk, J.},
  journal={IEEE Trans. Inf. Theory}, 
  title={An {A}lgorithm for {S}ource {C}oding}, 
  year={1972},
  month={May},
  volume={18},
  number={3},
  pages={395--399},
  }

@article{Yeongjae2019SR,
author={Yeongjae Choi and Taehoon Ryu and Amos C. Lee and others},
title={High information capacity {DNA}-based data storage with augmented encoding characters using degenerate bases},
journal={Sci. Rep.},
year={2019},
volume={9},
pages={6582}
}

@article{Leon2019NB,
author={Leon Anavy and Inbal Vaknin and  Orna Atar and Roee Amit and Zohar Yakhini},
title={Data storage in {DNA} with fewer synthesis cycles using composite {DNA} letters},
journal={Nat. Biotechnol.},
year={2019},
month={Oct.},
volume={37},
pages={1229--1236}
}

@ARTICLE{Makarychev2022IT,
  author={Makarychev, Konstantin and Rácz, Miklós Z. and Rashtchian, Cyrus and Yekhanin, Sergey},
  journal={IEEE Trans. Inf. Theory}, 
  title={Batch {O}ptimization for {DNA} {S}ynthesis}, 
  year={2022},
  month={Nov.},
  volume={68},
  number={11},
  pages={7454--7470}
  }

@article{Yan2023SR,
author={Yiqing Yan and  Nimesh Pinnamaneni and Sachin Chalapati and  Conor Crosbie and  Raja Appuswamy},
title={Scaling logical density of {DNA} storage with enzymatically‑ligated composite motifs},
journal={Sci. Rep.},
year={2023},
volume={13},
pages={15978}
}

@article{Press2024SR,
author={Inbal Preuss and  Michael Rosenberg and Zohar Yakhini and   Leon Anavy },
title={Efficient {DNA}‑based data storage using shortmer combinatorial encoding},
journal={Sci. Rep.},
year={2024},
volume={14},
pages={7731}
}

@ARTICLE{Tuan2024IT,
  author={Nguyen, Tuan Thanh and Cai, Kui and Siegel, Paul H.},
  journal={IEEE Trans. Inf. Theory}, 
  title={A {N}ew {V}ersion of $q$-{A}ry {V}arshamov-{T}enengolts {C}odes {W}ith {M}ore {E}fficient {E}ncoders: {T}he {D}ifferential {VT} {C}odes and {T}he {D}ifferential {S}hifted {VT} {C}odes}, 
  year={2024},
  month={Oct.},
  volume={70},
  number={10},
  pages={6989--7004}
}

@ARTICLE{Ohad2023IT,
  author={Elishco, Ohad and Huleihel, Wasim},
  journal={IEEE Trans. Inf. Theory}, 
  title={Optimal {R}eference for {DNA} {S}ynthesis}, 
  year={2023},
  month={Nov.},
  volume={69},
  number={11},
  pages={6941--6955}
  }

@article{Preuss2024TMBMSC,
author={Inbal Preuss and Ben Galili and Zohar Yakhini and Leon Anavy},
title={Sequencing {C}overage {A}nalysis for {C}ombinatorial {DNA}-{B}ased {S}torage {S}ystems},
year={2024},
month={Jun.},
journal={IEEE Trans. Mol. Biol. Multiscale Commun.},
volume={10},
number={2},
pages={297--316}
}

@ARTICLE{Wenkai2025IT,
  author={Zhang, Wenkai and Wang, Zhiying},
  journal={IEEE Trans. Inf. Theory}, 
  title={Codes for {L}imited-{M}agnitude {P}robability {E}rror in {DNA} {S}torage}, 
  year={2025},
  month={Jul.},
  volume={71},
  number={7},
  pages={5063--5081}
  }

@article{Roman2025TCOMM,
author={Roman Sokolovskii and Parv Agarwal and Luis Alberto Croquevielle and Zijian Zhou and Thomas Heinis},
title={Coding {O}ver {C}oupon {C}ollector {C}hannels for {C}ombinatorial {M}otif-{B}ased {DNA} {S}torage},
year={2025},
month={Jun.},
journal={IEEE Trans. Commun.},
volume={73},
number={6},
pages={3750--3760}
}

@ARTICLE{Immink2024TMBMSC,
  author={Schouhamer Immink, Kees A. and Cai, Kui and Nguyen, Tuan Thanh and Weber, Jos H.},
  journal={IEEE Trans. Mol. Biol. Multiscale Commun.}, 
  title={Constructions and {P}roperties of {E}fficient {DNA} {S}ynthesis {C}odes}, 
  year={2024},
  month={Jun.},
  volume={10},
  number={2},
  pages={289--296}
  }

@ARTICLE{Cohen2025JSAIT,
  author={Cohen, Tomer and Yaakobi, Eitan},
  journal={IEEE J. Sel. Areas Inf. Theory}, 
  title={Optimizing the {D}ecoding {P}robability and {C}overage {R}atio of {C}omposite {DNA}}, 
  year={2025},
  month={},
  volume={6},
  number={},
  pages={417--431}
  }

@ARTICLE{Zuo2025DCC,
  author={Zuo Ye and Omer Sabary and Ryan Gabrys and Eitan Yaakobi and Ohad Elishco},
  journal={Des. Codes Cryptogr.}, 
  title={More on codes for combinatorial composite {DNA}}, 
  year={2025},
  month={Aug.},
  volume={93},
  number={8},
  pages={3437--3462}
  }

@INPROCEEDINGS{Lenz2020ISIT,
  author={Lenz, Andreas and Liu, Yi and Rashtchian, Cyrus and Siegel, Paul H. and Wachter-Zeh, Antonia and Yaakobi, Eitan},
  booktitle={Proc. IEEE Int. Symp. Inf. Theory (ISIT)}, 
  title={Coding for {E}fficient {DNA} {S}ynthesis}, 
  year={2020},
  month={Jun.},
  volume={},
  number={},
  pages={2885--2890},
  address={Los Angeles, CA, USA}
  }

@INPROCEEDINGS{Abu-Sini2023ISIT,
  author={Abu-Sini, Maria and Lenz, Andreas and Yaakobi, Eitan},
  booktitle={Proc. IEEE Int. Symp. Inf. Theory (ISIT)}, 
  title={{DNA} {S}ynthesis {U}sing {S}hortmers}, 
  year={2023},
  month={Jun.},
  volume={},
  number={},
  pages={585--590},
  address={Taipei, Taiwan}
  }

@INPROCEEDINGS{Chrisnata2023ISIT,
  author={Chrisnata, Johan and Kiah, Han Mao and Long Phuoc Pham, Van},
  booktitle={Proc. IEEE Int. Symp. Inf. Theory (ISIT)}, 
  title={Deletion {C}orrecting {C}odes for {E}fficient {DNA} {S}ynthesis}, 
  year={2023},
  month={Jun.},
  volume={},
  number={},
  pages={352--357},
  address={Taipei, Taiwan}
  }

@INPROCEEDINGS{Frederik2024ISIT,
author={Frederik Walter and Omer Sabary and Antonia Wachter-Zeh and Eitan Yaakobi},
booktitle={Proc. IEEE Int. Symp. Inf. Theory (ISIT)},   
title={Coding for {C}omposite {DNA} to {C}orrect {S}ubstitutions, {S}trand {L}osses, and {D}eletions},
year={2024},
month={Jul.}, 
volume={},  
number={},  
pages={97--102},  
address={Athens, Greece}
}

@INPROCEEDINGS{Omer2024ISIT,
author={Omer Sabary and Inbal Preuss and Ryan Gabrys and Zohar Yakhini and Leon Anavy and Eitan Yaakobi},
booktitle={Proc. IEEE Int. Symp. Inf. Theory (ISIT)},  
title={Error-{C}orrecting {C}odes for {C}ombinatorial {DNA} {C}omposite},
year={2024},
month={Jul.}, 
volume={},  
number={},  
pages={109--114},  
address={Athens, Greece}
}

@INPROCEEDINGS{Tuan2024ISIT,
  author={Nguyen, Tuan Thanh and Cai, Kui and Schouhamer Immink, Kees A.},
  booktitle={Proc. IEEE Int. Symp. Inf. Theory (ISIT)}, 
  title={Efficient {DNA} {S}ynthesis {C}odes with {E}rror {C}orrection and {R}unlength {L}imited {C}onstraint}, 
  year={2024},
  month={Jul.},
  volume={},
  number={},
  pages={669--674},
  address={Athens, Greece}
  }

@INPROCEEDINGS{Yajuan2025ISIT,
  author={Liu, Yajuan and Duman, Tolga M.},
  booktitle={Proc. IEEE Int. Symp. Inf. Theory (ISIT)}, 
  title={Constrained {E}rror-{C}orrecting {C}odes for {E}fficient {DNA} {S}ynthesis}, 
  year={2025},
  month={Jun.},
  volume={},
  number={},
  pages={},
  address={Ann Arbor, MI, USA}
  }

@INPROCEEDINGS{WangChen2025ITW,
  author={Wang, Chen and Nguyen, Tuan Thanh and Cai, Kui and Zhang, Yiwei},
  booktitle={Proc. IEEE Inf. Theory Workshop (ITW)}, 
  title={Correcting {E}rrors in {C}omposite {DNA}: {C}hannel {M}odel and {C}ode {D}esign}, 
  year={2025},
  month={Sept.},
  volume={},
  number={},
  pages={1--6},
  address={Sydney, Australia}
  }

@INPROCEEDINGS{Cohen2025ISTC,
  author={Cohen, Tomer and Wang, Zhiying and Yaakobi, Eitan and Yakhini, Zohar},
  booktitle={Proc. 13th Int. Symp. Topics in Coding (ISTC)}, 
  title={Rank {M}odulated {C}omposite {E}ncoding for {D}ata {S}torage in {DNA}}, 
  year={2025},
  month={Aug.},
  volume={},
  number={},
  pages={},
  address={Los Angeles, USA}
  }

@INPROCEEDINGS{Frederik2025ISIT,
author={Frederik Walter and Yonatan Yehezkeally},
booktitle={Proc. IEEE Int. Symp. Inf. Theory (ISIT)}, 
title={Coding for {S}trand {B}reaks in {C}omposite {DNA}},
year={2025},
month={Jun.},
volume={},
number={},
pages={},
address={Ann Arbor, MI, USA}
}

@INPROCEEDINGS{Tuan2025ISIT,
  author={Nguyen, Tuan Thanh and Wang, Chen and Cai, Kui and Zhang, Yiwei and Yakhini, Zohar},
  booktitle={Proc. IEEE Int. Symp. Inf. Theory (ISIT)}, 
  title={Constrained {C}oding for {C}omposite {DNA}: {C}hannel {C}apacity and {E}fficient {C}onstructions}, 
  year={2025},
  month={Jun.},
  volume={},
  number={},
  pages={},
  address={Ann Arbor, MI, USA}
  }

@book{LucGabor2001,
  title={Combinatorial Methods in Density Estimation},
  author={Luc Devroye and G\'{a}bor Lugosi},
  series={Springer Series in Statistics},
  volume={},
  year={2001},
  publisher={Springer New York, NY},
  edition={1}
}

@book{Pan2001,
  author    = {Victor Y. Pan},
  title     = {Structured Matrices and Polynomials:
               Unified Superfast Algorithms},
  publisher = {Birkh{\"a}user, Boston, MA},
  year      = {2001},
  edition={1}
}

@article{BesartDollma202509,
author={Besart Dollma and Ohad Elishco and Eitan Yaakobi},
title={Coding for {O}rdered {C}omposite {DNA} {S}equences},
year={2025},
month={Sept.},
journal={arXiv:2509.26119},
volume={}
}

@misc{kabal2018,
  title={Combinatorial {C}oding and {L}exicographic {O}rdering},
  author={Kabal, Peter},
  year={2018},
  url={https://www.mmsp.ece.mcgill.ca/Documents/Reports/2018/KabalR2018.pdf}
}
\end{document}